\documentclass[acmsmall]{acmart}

\usepackage{makecell}

\usepackage[T1]{fontenc}
\usepackage[utf8]{inputenc}
\usepackage{terms} 
\usepackage{graphicx}
\usepackage{import} 
\usepackage{lipsum} 
\usepackage{hyperref}

\usepackage{amssymb}
\usepackage{listings} 
\usepackage{csquotes} 
\usepackage[title]{appendix}
\usepackage{parcolumns} 
\usepackage{subcaption}
\usepackage{pmboxdraw} 
\usepackage{multicol} 
\usepackage{wrapfig} 
\usepackage[noend]{algpseudocode} 
\usepackage{algorithm} 
\usepackage{tikz} 
\usepackage[]{todonotes} 
\usepackage{framed} 
\usepackage{proof} 
\usepackage{color}
\usepackage[nounderscore]{syntax} 
\usepackage{mathpartir} 
\usepackage[super]{nth}
\usepackage{tabularx}
\usepackage{float}
\usepackage{textcomp}
\usepackage{mathtools}
\usepackage{wrapfig}
\usepackage{pifont}
\usepackage[frozencache, cachedir=minted-cache]{minted2}
\usepackage{pgf}
\usepackage{newunicodechar}
\usepackage[inline]{enumitem}
\usepackage{threeparttable}


\newunicodechar{ℙ}{\ensuremath{\mathbb{P}}}
\newunicodechar{ℝ}{\ensuremath{\mathbb{R}}}
\newunicodechar{∀}{\ensuremath{\forall}}
\newunicodechar{⋀}{\ensuremath{\wedge}}
\newunicodechar{≤}{\ensuremath{\leq}}
\newunicodechar{ᵇ}{\ensuremath{^b}}
\newunicodechar{…}{\ensuremath{\cdots}}

\newunicodechar{₁}{\ensuremath{_1}}
\newunicodechar{₂}{\ensuremath{_2}}
\newunicodechar{₃}{\ensuremath{_3}}
\newunicodechar{₄}{\ensuremath{_4}}
\newunicodechar{₅}{\ensuremath{_5}}
\newunicodechar{ι}{\ensuremath{\iota}}
\newunicodechar{ₗ}{\ensuremath{_l}}
\newunicodechar{ₐ}{\ensuremath{_a}}
\newunicodechar{ₖ}{\ensuremath{_k}}
\newunicodechar{ᵢ}{\ensuremath{_i}}
\newunicodechar{ⱼ}{\ensuremath{_j}}
\newunicodechar{ᵥ}{\ensuremath{_v}}
\newunicodechar{ₕ}{\ensuremath{_h}}

\newunicodechar{ᵣ}{\ensuremath{_r}}
\newunicodechar{ʳ}{\ensuremath{^r}}
\newunicodechar{ˡ}{\ensuremath{^l}}
\newunicodechar{ℕ}{\ensuremath{\mathbb{N}}}
\newunicodechar{∀}{\ensuremath{\forall}}
\newunicodechar{≡}{\ensuremath{\equiv}}
\newunicodechar{≈}{\ensuremath{\approx}}
\newunicodechar{∎}{\ensuremath{\blacksquare}}
\newunicodechar{⊛}{\ensuremath{\circledast}}
\newunicodechar{⊗}{\ensuremath{\otimes}}
\newunicodechar{⊕}{\ensuremath{\oplus}}
\newunicodechar{φ}{\ensuremath{\phi}}
\newunicodechar{ψ}{\ensuremath{\psi}}
\newunicodechar{ε}{\ensuremath{\epsilon}}
\newunicodechar{λ}{\ensuremath{\lambda}}
\newunicodechar{σ}{\ensuremath{\sigma}}
\newunicodechar{′}{{'}}
\newunicodechar{∷}{\ensuremath{\dblcolon}}
\newunicodechar{↔}{\ensuremath{\leftrightarrow}}
\newunicodechar{↦}{\ensuremath{\mapsto}}
\newunicodechar{∘}{\ensuremath{\circ}}
\newunicodechar{⁻}{\ensuremath{^{-}}}
\newunicodechar{∙}{\ensuremath{\boldsymbol{\cdot}}}
\newunicodechar{▹}{\ensuremath{\triangleright}}
\newunicodechar{Σ}{\ensuremath{\Sigma}}
\newunicodechar{∈}{\ensuremath{\in}}
\newunicodechar{⟦}{\ensuremath{\llbracket}}
\newunicodechar{⟧}{\ensuremath{\rrbracket}}
\newunicodechar{⇒}{\ensuremath{\Rightarrow}}
\newunicodechar{⟪}{\ensuremath{\xlbrace}}
\newunicodechar{⟫}{\ensuremath{\xrbrace}}
\newunicodechar{⊠}{\ensuremath{\boxtimes}}
\newunicodechar{⊞}{\ensuremath{\boxplus}}
\newunicodechar{∋}{\ensuremath{\ni}}
\newunicodechar{∶}{\ensuremath{\bm{:}}}
\newunicodechar{↦}{\ensuremath{\mapsto}}
\newunicodechar{ℂ}{\ensuremath{\mathbb{C}}} 

\newunicodechar{ʹ}{\ensuremath{\prime}}
\newunicodechar{≢}{\ensuremath{\not\equiv}}
\newunicodechar{ω}{\ensuremath{\omega}}
\newunicodechar{ₙ}{\ensuremath{_n}}
\newunicodechar{₀}{\ensuremath{_0}}
\newunicodechar{⋆}{\ensuremath{\cdot}}
\newunicodechar{♯}{\ensuremath{\sharp}}
\newunicodechar{♭}{\ensuremath{\flat}}
\newunicodechar{ₛ}{\ensuremath{_s}}
\newunicodechar{≅}{\ensuremath{\cong}}
\newunicodechar{ᵗ}{\ensuremath{^t}}
\newunicodechar{⊡}{\ensuremath{\boxdot}}
\newunicodechar{∔}{\ensuremath{^+}}
\newunicodechar{ℝ}{\ensuremath{\mathbb{R}}}
\newunicodechar{⊤}{\ensuremath{\top}}
\newunicodechar{ℤ}{\ensuremath{\mathbb{Z}}}
\newunicodechar{𝕋}{\ensuremath{\mathbb{T}}}
\newunicodechar{⊖}{\ensuremath{\ominus}}

\usetikzlibrary{positioning, fit, calc, shapes.geometric, arrows.meta, backgrounds}

\pgfdeclarelayer{background}
\pgfsetlayers{background,main}

\newcommand{\sgrammar}[1]{\todo[color=orange!40]{Sentence OK?}}
\newcommand{\listingsize}{\small}
\newcommand{\correction}[1]{#1}

\renewcommand{\vehicle}{\textsc{Vehicle}\xspace} 
\newcommand{\keymaera}{KeYmaera~X}

\newcommand{\network}{\mathcal{N}}
\newcommand{\spec}{\mathcal{S}}
\newcommand{\world}{\mathcal{W}}

\newcommand{\cmark}{\textcolor{green}{\ding{51}}}
\newcommand{\xmark}{\textcolor{red}{\ding{55}}}

\DeclareUnicodeCharacter{211A}{$\mathbb{Q}$}
\DeclareUnicodeCharacter{2200}{$\forall$}

\lstdefinelanguage{agda}{
  morekeywords={data,where,module,import,open,record,field,constructor,infixl,infixr,infix,Set,forall},
  sensitive=true,
  morecomment=[l]{--},
  morecomment=[s]{\{-}{-\}},
  morestring=[b]"
}

\lstdefinelanguage{coq}{
  morekeywords={Definition,Lemma,Theorem,Proof,Qed,Require,Import,Notation,Context,Variant,Record},
  sensitive=true,
  morecomment=[s]{(*}{*)},
  morestring=[b]"
}

\lstdefinelanguage{example}{
  basicstyle=\small\ttfamily\color{VioletRed},
  moredelim=[l][\color{blue}]{: },
  moredelim=[l][\color{black}]{:=},
  moredelim=[s][\color{blue}]{new}{\ },
  morekeywords=[1]{State, Observation, controller, },
  keywordstyle=[1]\bfseries\color{blue},
  morekeywords=[2]{o, s},
  keywordstyle=[2]\itshape\color{black},
  morekeywords=[3]{where},
  keywordstyle=[3]\color{YellowOrange},
  literate={:=}{{\textcolor{OliveGreen}{:=}}}2 {,}{{\textcolor{OliveGreen}{,}}}1
  {\{}{{\textcolor{OliveGreen}{\{}}}1 {\}}{{\textcolor{OliveGreen}{\}}}}1
  {->}{{\textcolor{black}{\(\rightarrow\)}}}1 {=}{{\textcolor{OliveGreen}{=}}}1
  {+}{{\textcolor{OliveGreen}{+}}}1
}

\colorlet{isaTypeVarColor}{blue}

\newcommand{\aref}[1]{\hyperref[#1]{Appendix~\ref*{#1}}}

\setcopyright{cc}
\setcctype{by}
\acmDOI{10.1145/3828680}
\acmYear{2026}
\acmJournal{PACMPL}
\acmVolume{10}
\acmNumber{ICFP}
\acmArticle{282}
\acmMonth{8}
\acmSubmissionID{icfp26main-p34-p}
\received{2026-02-25}
\received[accepted]{2026-05-13}

\begin{document}

\title[Compositional Neural-Cyber-Physical System Verification in the ITP of Your Choice]{Compositional Neural-Cyber-Physical System Verification in the Interactive Theorem Prover of Your Choice}

\author{Matthew L. Daggitt}
\orcid{0000-0002-2552-3671}
\affiliation{
    \institution{University of Western Australia}
    \city{Perth}
    \country{Australia}
}
\email{matthew.daggitt@uwa.edu.au}

\author{Ekaterina Komendantskaya}
\orcid{0000-0002-3240-0987}
\affiliation{
    \institution{Heriot-Watt University}
    \city{Edinburgh}
    \country{UK}
}
\affiliation{
\institution{University of Southampton}
\city{Southampton}
\country{UK}
}

\author{Alistair Sirman}
\orcid{0009-0009-3342-089X}
\affiliation{
\institution{University of Southampton}
\city{Southampton}
\country{UK}
}

\author{Alessandro Bruni}
\orcid{0000-0003-2946-9462}
\affiliation{%
   \institution{IT University of Copenhagen}
   \city{Copenhagen}
   \country{Denmark}
}

\author{Samuel Teuber}
\orcid{0000-0001-7945-9110}
\affiliation{
    \institution{Karlsruhe Institute of Technology}
    \city{Karlsruhe}
    \country{Germany}
}

\author{Josh Smart}
\orcid{0009-0004-9452-1911}
\affiliation{
    \institution{University of Southampton}
    \city{Southampton}
    \country{UK}
}

\author{Grant Passmore}
\orcid{0009-0000-5335-1056}
\affiliation{
\institution{Imandra Inc.}
\city{Austin}
\state{Texas}
\country{USA}
}


\begin{abstract}
Formal verification of neuro-symbolic cyber-physical systems, such as drones,
medical devices and robots, is complicated.
Neural components must be trained to be optimal with respect to the available data as well as the safety specifications, and then verified using specialised solvers.
Symbolic models of the ``cyber'' and ``physical'' behaviour of the system must
be constructed and verified in interactive theorem provers (ITPs), often
requiring mature mathematical libraries to reason about the interplay of
discrete and continuous dynamics, preferably obtaining infinite time-horizon
guarantees. Finally, the results of the two already challenging verification
tasks need to be integrated into a single proof. 
In this paper we present a compositional methodology for constructing such proofs.  
The \vehicle framework provides a functional, domain-specific language for specifying, training, and verifying neural components. We extend \vehicle{} to allow integration with any ITP with minimal effort, thereby bridging the gap between the neural and symbolic proofs.
First, we describe how \vehicle{}’s standard bidirectional type checker can be reused to transpile neural specifications into an intermediate representation targeting multiple theorem provers.
Second, we integrate \vehicle with Rocq, Isabelle/HOL,
Agda and the industrial prover Imandra; and showcase a
generic infinite time-horizon safety proof of a discrete cyber-physical system with a neural network controller in each ITP. Finally, to put the idea of \emph{compositional neural-cyber-physical system verification} to the test, we use the Mathematical Components libraries in Rocq to verify infinite time-horizon safety of a medical device, modelled as a \emph{continuous} cyber-physical system with a neural controller. To our knowledge, this is the first result of this kind in a general purpose ITP; and a result that was only feasible thanks to the compositionality provided by \vehicle{}'s functional interface.  
\end{abstract}

\maketitle
\renewcommand{\shortauthors}{M. Daggitt et al.}

\section{Introduction}
\label{section:introduction}

\looseness=-1
Verification of cyber-physical systems (CPS) is a well-established research area that combines automated verification, reachability analysis, and differential equation solving~\cite{ARCH25}. 
There are several mature tools such as CORA~\cite{Althoff2015ARCH}, JuliaReach~\cite{10.1145/3302504.3311804} and \keymaera~\cite{FultonMQVP15}, and their progress is evaluated annually in the ARCHCOMP competition~\cite{ARCH25}. 
In parallel, the Interactive Theorem Prover (ITP) community has also shown
significant interest in verifying such systems, e.g. in
Isabelle/HOL~\cite{MuniveFGSLH24,MuniveS22,FosterMGS21,BohrerRVVP17}, Rocq~\cite{GeuversKSW10,OuchaniKH20,Ricketts17,BohrerRVVP17}, PVS~\cite{WhiteTSM24} or ACL2~\cite{boyer1990use}.
Generally, a CPS \correction{may be} formalised and verified in a general-purpose ITP directly (examples in this paper follow this route), or reasoned about via domain-specific languages embedded in the ITP: e.g. \emph{differentiable dynamic logic} (underlying  \keymaera) has been implemented in PVS~\cite{WhiteTSM24}, Isabelle/HOL~\cite{BohrerRVVP17,FosterMGS21} and Rocq~\cite{BohrerRVVP17}.  

\looseness=-1
A new generation of cyber-physical systems deploys neural network
components~\cite{JulianACASXuDNN,DBLP:conf/eann/Lopez-MiguelAGV23}
, leveraging the ability of neural networks to learn complex control strategies while remaining lightweight enough for real-time deployment. 
This has proven effective in a wide range of applications~\cite{ARCH25}, from CERN cooling towers~\cite{DBLP:conf/eann/Lopez-MiguelAGV23} to pilot advisory systems~\cite{JulianACASXuDNN}. 
In this paper we refer to such systems as \emph{neural-cyber-physical} systems (NCPS). Verifying the safety of these systems has been recognised as an important challenge in cyber-physical system verification~\cite{ARCH_COMP19_Category_Report,CordeiroDGIJKKLMSW25}.
As illustrative example, consider the classic cyber-physical verification problem of \citet{boyer1990use} where a car travels along a straight road while subject to cross-wind.
The car's sensors provide noisy measurements of its environment, and in the
original paper a symbolic controller is used to update its steering and velocity.
The original safety property proven by \citet{boyer1990use} states:

\begin{wrapfigure}[9]{r}{0.45\textwidth} 
    \vspace{-2em}
    \centering
    \includegraphics[width=.4\textwidth]{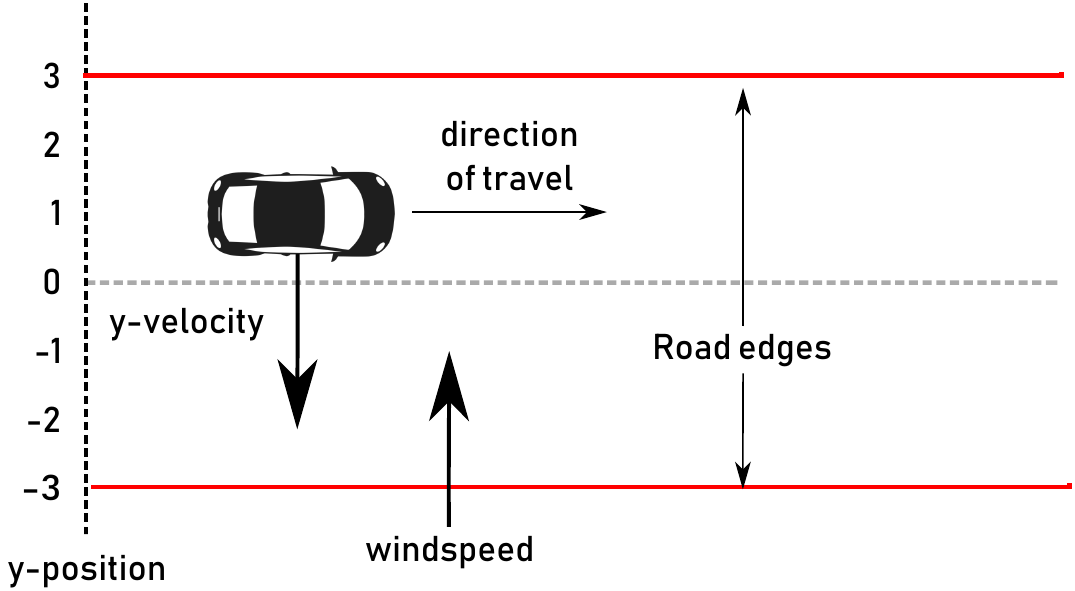}
    \vspace{-1em}
    \caption{Our running example of NCPS: A model of an autonomous car compensating for crosswinds.}
    \vspace{1em}
    \label{fig:wind-controller}
\end{wrapfigure}
\begin{theorem}\label{quote:theorem-1}
    Given that the wind can shift by no more than 1 per unit time, and the
    sensor is never off by more than 0.25 units, the car will never leave the
    road.
\end{theorem}
\noindent 
In the CPS literature, this type of result is called an \emph{infinite time-horizon} guarantee, as it ensures system safety for the entire execution of the system, rather than for a bounded time period.

Replacing the symbolic controller with a neural network transforms this example into a neural-cyber-physical-system where the car's software invokes the neural network at fixed time intervals to generate driving actions.
While the safety property remains the same,
we cannot prove it without also reasoning about the neural network's behavior.

\looseness=-1
Let us decompose the problem. An NCPS  can be represented as $\world(\spec(\network))$ where $\world$ is a model of the world the system operates in (i.e. the \emph{physical}), $\spec$ is a model of the conventional software components (i.e. the \emph{cyber}), and $\network$ is the neural network controller (i.e. the \emph{neural}).
The verification goal is then to establish a safety property $\Phi(\world(\spec(\network)))$, for example that the car always remains within the road boundaries.
However, while the cyber-physical components of the system, $\world$ and $\spec$, can be represented and reasoned about \emph{symbolically}, the neural component, $\network$, is a \emph{sub-symbolic} model\correction{, i.e.\ it represents knowledge as numerical values rather than as explicit, human-readable symbols and hence admits no semantically meaningful decomposition.}
For example, the world model $\world$ and symbolic software $\spec$ may contain the usual definitions of distance and velocity as functions of acceleration and time. 
In contrast, when $\network$ processes data that monitors distance, velocity and acceleration, its internal
representation of this data is opaque and devoid of interpretable symbolic meaning, preventing us from decomposing $\network$ to reason about it symbolically and interactively.
\begin{center}
\emph{The challenge is to develop methods that can combine proofs of safety of the symbolic cyber-physical system with proofs about the properties of its neural components.}
\end{center}

\subsection{Existing Automated Approaches}
\label{sec:automated-solvers}

\looseness=-1
Motivated by this problem, a variety of techniques have been developed to prove the correctness of the neural component, $\network$, in isolation~\cite{albarghouthi2021introductionneuralnetworkverification}.
Linear programming-based neural network verification tools~\cite{katz_reluplex_2017,wu2024marabou,bak2020improved,nnv2_cav2023} encode the neural network and desired property as mostly linear constraints, and apply linear constraint solvers that leverage sophisticated branch-and-bound techniques to handle the non-linearities.
Abstract interpretation-based tools~\cite{singh_abstract_2019,zhang_efficient_2018} trade completeness for scalability to larger neural networks.
Recent variants combine abstract interpretation with gradient descent~\cite{xufast} and Lagrangian optimisation~\cite{DBLP:conf/nips/WangZXLJHK21} to recover completeness via branch-and-bound.
While most tools limit themselves to verifying reachable (interval) bounds of a neural network or specifications given as linear constraints,
some recent tools support nonlinear/polynomial constraints~\cite{TeuberVerSAILLE2024,DBLP:conf/tacas/ShiJKJHZ25}. We will refer to such tools as \emph{neural solvers} in this paper.

\looseness=-1
These automated techniques have been extended to analyse a neural network $\network$ and its physical environment $\world(S(\cdot))$ in combination using a coupled reachability analysis, thereby obtaining safety guarantees about the entire system automatically~\cite{huang2022polar,wang2023polar,Althoff2015ARCH,10.1145/3302504.3311804,ivanov2021verisig}. They bypass the symbolic/sub-symbolic gap by implementing a combined reachability analysis in the same language: e.g. MATLAB in case of CORA~\cite{Althoff2015ARCH}, or C++ in case of PolarExpress~\cite{huang2022polar,wang2023polar}. 
However, this convenience comes at a cost:
rather than establishing invariants and infinite-time horizon safety properties,
their guarantees are usually confined to \emph{bounded} time horizons; moreover, performance degrades rapidly as system complexity or horizons grow~\cite{KesslerKCVFOMM25}. 
This restricts the range of neural-cyber-physical systems that can be verified automatically. For example, these approaches would not help to prove Theorem~\ref{quote:theorem-1}. 

\subsection{The Compositional Approach}

An alternative approach that overcomes these limitations takes inspiration from the decomposition $\world(\spec(\network))$. Starting from the desired system-level safety property~$\Phi$, one first derives a specification $\psi$ of the neural network such that the proof of $\Phi(\world(\spec(\network)))$ can be decomposed as follows:
\begin{equation}\label{eq:prop-impl}
  \psi(\network)
  \qquad
  \wedge
  \qquad
  \forall f. \: \psi(f) \Rightarrow \Phi(\world(\spec(f)))
\end{equation}
As shown in Figure~\ref{fig:system-decomposition}, this separates reasoning about the \emph{symbolic} cyber-physical system $\world(\spec(\cdot))$ from reasoning about the \emph{sub-symbolic} neural component $\network$.
This enables the use of existing ITPs to model and reason about complex symbolic dynamics (e.g. cyber-physical, probabilistic, perturbed, or partially observable systems) and obtain stronger guarantees about the full system (e.g.\ infinite-time horizon safety).
\correction{The exact fragment of logic that $\psi$ belongs to is unimportant as long as there is a specialised neural solver that supports it.}
%

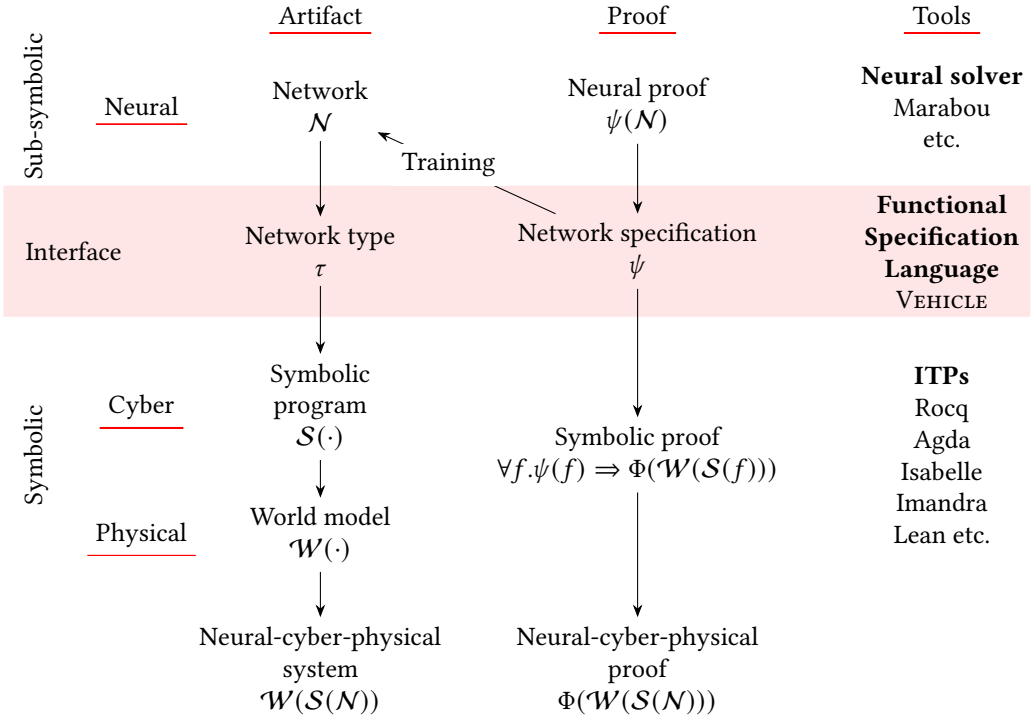
\begin{figure}[t]
\begin{tikzpicture}[
    >=Stealth, 
    node distance=0.5cm and 1.5cm, 
    align=center, 
    red_underline/.style={text=black, path picture={ 
        \draw[red, thick] (path picture bounding box.south west) -- (path picture bounding box.south east);
    }},
    item/.style={anchor=center}
]
    \draw[white, fill=red!10] (-4.2,-4) rectangle ++(13.6,1.75) {};
    
    \node [red_underline] (head_art) {Artifact};
    \node [item, below=0.5cm of head_art] (net_n) {Network\\$\network$};
    \node [item, below=1.0cm of net_n] (type_tau) {Network type\\$\tau$};
    \node [item, below=0.9cm of type_tau] (sym_prog) {Symbolic\\program\\$\spec(\cdot)$};
    \node [item, below=of sym_prog] (world_mod) {World model\\$\world(\cdot)$};
    \node [item, below=0.6cm of world_mod] (res_art) {Neural-cyber-physical \\system\\$\world(\spec(\network))$};

    \draw[->] (net_n) -- (type_tau);
    \draw[->] (type_tau) -- (sym_prog);
    \draw[->] (sym_prog) -- (world_mod);
    \draw[->] (world_mod) -- (res_art);

    \node [red_underline] (head_prf) [right=3cm of head_art] {Proof};
    \node [item, at={(head_prf |- net_n.center)}] (phi_n) {Neural proof\\$\psi(\network)$};
    \node [item, at={(head_prf |- type_tau.center)}] (spec_psi) {Network specification\\$\psi$};
    \node [item, below=1.7cm of spec_psi] (impl_form) {Symbolic proof \\ $\forall f . \psi(f) \Rightarrow \Phi(\world(\spec(f)))$};
    \node [item, at={(head_prf |- res_art.center)}] (res_prf) {Neural-cyber-physical \\proof \\$\Phi(\world(\spec(\network)))$};

    \draw[->] (phi_n) -- (spec_psi);
    \draw[->] (spec_psi) -- (impl_form);
    \draw[->] (impl_form) -- (res_prf);
    \draw[->] (spec_psi) -- node[midway, fill=white,pos=0.6] {Training} (net_n);

    \node [red_underline] (head_tool) [right=3cm of head_prf] {Tools};
    \node [item, at={(head_tool |- net_n.east)}] (tools_top) {\textbf{Neural solver}\\Marabou\\etc.};
    \node [item, at={(head_tool |- type_tau.east)}] (vehicle) {\textbf{Functional} \\\textbf{Specification} \\\textbf{Language} \\ \vehicle};
    \node [item, at={(head_tool |- impl_form.east)}] (itps) {\textbf{ITPs}\\Rocq\\Agda\\Isabelle\\Imandra\\Lean\ etc.};

    
    \node (lbl_neur) [left=1cm of net_n, red_underline] {Neural};
    \node (lbl_cyb) [red_underline] at (lbl_neur.south|-sym_prog.west) {Cyber};
    \node (lbl_phys) [red_underline] at (lbl_neur.south|-world_mod.west) {Physical};

    \node (lbl_sub) [anchor=east, left=0.5cm of lbl_neur] {\rotatebox{90}{Sub-symbolic}};
    \node (lbl_sym) [anchor=east] at (lbl_sub.east|-impl_form.west) {\rotatebox{90}{Symbolic}};
    \node[text=black,anchor=center,xshift=0.5cm] (interface) at (lbl_sub.center|-spec_psi.west) {Interface};


\end{tikzpicture}
\caption{The architecture of a compositional proof for a neural-cyber-physical system, with specifications in the \vehicle DSL, in red, acting as an interface between existing neural solvers and existing ITPs. }  
\label{fig:system-decomposition}
\end{figure}

\looseness=-1
Although the compositional approach is appealing due to its potential for stronger guarantees, the only attempt to verify neural-cyber-physical systems in this way that we are aware of is Teuber \emph{et al.}~\cite{TeuberAngelsAndDemons2025,TeuberVerSAILLE2024}.
However, their work did not integrate the proof of $\psi(\network)$, obtained from a specialised neural solver (NCubeV~\cite{TeuberVerSAILLE2024}), and the proof of $\forall f. \: \psi(f) \Rightarrow \Phi(\world(\spec(f)))$, formalised in \keymaera~\cite{FultonMQVP15}.
Therefore, despite achieving infinite-time horizon guarantees,
their results rely on the manual pen-and-paper composition of the proofs and lacks mechanised guarantees about alignment across the symbolic/sub-symbolic gap. We thus refine our thesis as follows:

\begin{center}
\emph{The challenge is to develop \textbf{compositional} methods for proving the safety of neural-cyber-physical systems.}  
\end{center}

\looseness=-1
\correction{This is where functional programming offers a compelling solution. Functional languages provide a common framework in which both neural network interfaces and mathematical specifications ($\psi$) can be expressed with precise semantics. At the same time, they align naturally with the functional foundations of interactive theorem provers and can be translated readily into the declarative formats expected by neural network solvers. As a result, functional programming provides an effective foundation for integrating symbolic and sub-symbolic proofs.}
Based on these observations, \citet{daggitt_et_al:LIPIcs.FSCD.2025.2} proposed \vehicle, a functional domain specific language (DSL), for writing down the types and high-level specifications for neural networks. 
As shown in Figure~\ref{fig:system-decomposition}, \vehicle serves as the interface between the tools for verifying neural networks that operate in the sub-symbolic world and ITPs for reasoning about the larger symbolic system. 
Concretely, \citet{daggitt_et_al:LIPIcs.FSCD.2025.2} provided a high-level, non-technical description of how \vehicle{} specifications are compiled to (i) loss functions for training the neural network, (ii) queries for the specialised neural solver Marabou~\cite{wu2024marabou} to verifying the resulting network, and (iii) Agda code for reasoning about the larger cyber-physical system.

\subsection{Contributions}

In this paper we make the following contributions:
\begin{enumerate}
  \item \textbf{The first description of the internal implementation of the \vehicle core language.} 
    The three ecosystems targeted by \vehicle{} -- machine learning frameworks, neural network verifiers and ITPs -- all have very different expressiveness and capabilities. Generating code for each of them requires unique analyses of the semantics of the specification $\psi$. In Section~\ref{sec:vehicle-core} we provide a concrete description of the architecture of the \vehicle{} type-checker that allows us to reuse it to perform the backend-specific analyses in a modular fashion.\footnote{\correction{Although implemented, we do not provide a detailed description of the implementation of the learning and verification backends in this paper.}}
  \item \textbf{A novel ITP intermediate language.} 
  In the original proposal~\cite{daggitt_et_al:LIPIcs.FSCD.2025.2}, \vehicle{} only supported exporting neural network specifications to Agda. 
  To address this limitation, in Section~\ref{section:intermediate-language} we describe how the type-checker described in Section~\ref{sec:vehicle-core} can be used to compile the core \vehicle{} language to a novel intermediate language used by the ITP backend that facilitates exporting the specification to any ITP with minimal additional code. We believe this technique is applicable to other domains which require translating Boolean specifications into ITPs.
  \item \textbf{Integration with 4 mainstream ITPs.} 
  \looseness=-1
  In Section~\ref{section:itp-implementations}, we describe
  implementing support for exporting \vehicle{} specifications to four mainstream general-purpose ITPs: Agda~\cite{norell2009dependently}, Rocq/MathComp~\cite{RocqManual,mathcomp,affeldt_mathcomp-analysis_2026}%
  , Isabelle/HOL~\cite{DBLP:books/sp/NipkowPW02} and Imandra~\cite{passmore2020imandra}.
  Together these systems cover both dependently/non-dependently typed systems and academic/industrial systems, demonstrating that our approach is applicable across the diverse ecosystem of ITPs. 
  We describe the prover-specific design choices required in each case and compare the strengths and limitations of these systems for expressing neural-cyber-physical specifications.
  In particular, as an original contribution, we present a new Tensor library for MathComp~\cite{mathcomp} that is required to model neural network specifications.
  We prove infinite-horizon safety (Theorem~\ref{quote:theorem-1}) of the discrete NCPS in Figure~\ref{fig:wind-controller} in each of the four provers. 
  
  \item \textbf{Infinite time-horizon safety of a continuous NCPS in a general-purpose ITP.}
  To demonstrate the advantages of our compositional approach, we present what we
    believe to be the first case study that verifies infinite-horizon safety of
    a continuous NCPS, a medical device, in a general-purpose ITP (Rocq + MathComp Analysis\cite{affeldt_mathcomp-analysis_2026}).
    The techniques used in the case study
    can be generalised to any neural-cyber-physical system in which the property
    $\Phi$ guarantees the behavior known as \emph{exponential
      decay}~\cite{simmonsDifferentialEquationsApplications2016}. We show how
    recently released MathComp Analysis libraries (e.g. \texttt{derive.v} for multivariate
    derivatives and differentiation, \texttt{classical/filter.v} for continuity) can
    be used to prove properties of NCPS described by ordinary differential equations.
\end{enumerate}
Together, these contributions provide a concrete demonstration of how functional programming can be used to bridge the symbolic/sub-symbolic gap in NCPS verification in a fully compositional manner \correction{and obtain new infinite-horizon safety guarantees that were not possible before}.
In particular, this work opens the way for existing~\cite{MuniveFGSLH24,MuniveS22,FosterMGS21,BohrerRVVP17,GeuversKSW10,OuchaniKH20,Ricketts17} and future CPS formalisations in ITPs to integrate with neural network verification and training with minimal cost.
We hope that this will significantly increase the attractiveness of general-purpose ITPs as a mainstream tool for the verification of neural-cyber-physical systems.

\section{Recap: \vehicle{} Surface Language}
\label{sec:vehicle-surface}

We begin by providing a short Agda formalisation of the cyber-physical components of the car example shown in Figure~\ref{fig:wind-controller} and a brief recap of the \vehicle{} surface language.

\subsection{Wind Controller Example}
\label{sec:wind-controller-example}

In order to prove the safety of the NCPS introduced in
\autoref{section:introduction}, we start by defining a model in Agda of the cyber-physical components $\world(\spec(\cdot))$ of the system as described by \citet{boyer1990use}.
The \texttt{State} datatype represents the state of the world, and the noisy sensor reading of its position obtained by the car. 
The \texttt{Observation} \correction{datatype represents} updates to the world provided by an oracle, and finally the NCPS \texttt{controller} acts upon the sensor reading from the current and previous state, and instructs the car to change its velocity.

\begin{minipage}[t]{0.3\textwidth}\vspace{0pt}
  \begin{minted}[fontsize=\listingsize]{agda}
record State : Set
  where field
    windSpeed : ℝ
    position  : ℝ
    velocity  : ℝ
    sensor    : ℝ
  \end{minted}
\end{minipage}%
\begin{minipage}[t]{0.3\textwidth}\vspace{0pt}
  \begin{minted}[fontsize=\listingsize]{agda}
record Observation : Set
  where field
    windShift   : ℝ
    sensorError : ℝ
\end{minted}
\end{minipage}%
\begin{minipage}[t]{0.3\textwidth}\vspace{0pt}
  \begin{minted}[fontsize=\listingsize]{agda}
  controller : Tensor ℝ [2] → ℝ
  controller = ... -- network
  \end{minted}
\end{minipage}
\vspace{1em}

\noindent Although the dynamics of the system is better modelled continuously (e.g. using ODEs), for simplicity we will model the cyber and physical components discretely using the following functions:

\vspace{0.1em}
\begin{minted}[fontsize=\listingsize, xleftmargin=\parindent]{agda}
nextState : Observation → State → State
nextState o s = State newWindSpeed newPosition newVelocity newSensor
  where
  newWindSpeed = windSpeed s + windShift o
  newPosition  = position s + velocity s + newWindSpeed
  newSensor    = newPosition + sensorError o
  newVelocity  = velocity s + controller [ newSensor , sensor s ]

finalState : List Observation → State
finalState xs = foldr nextState initialState xs
\end{minted}
Given this model of $\world(S(\cdot))$, we can formally state the system property $\Phi$ described by~\autoref{quote:theorem-1}.

\noindent \begin{minipage}[t]{0.4\textwidth}
  \vspace{0pt}
  \begin{minted}[fontsize=\listingsize]{agda}
  OnRoad : State → Set
  OnRoad s = | s.position | ≤ 3
  \end{minted}
\end{minipage}
\hfill
\begin{minipage}[t]{0.6\textwidth}
  \vspace{0pt}
  \begin{minted}[fontsize=\listingsize]{agda}
ValidObservation : Observation → Set
ValidObservation o = | o.sensorError | ≤ 0.25 ⋀
                     | o.windShift | ≤ 1
  \end{minted}
\end{minipage}
\\
\begin{minted}[fontsize=\listingsize, xleftmargin=\parindent]{agda}
systemSafety : ∀ xs → All ValidObservation xs → OnRoad (finalState xs)
systemSafety = ... -- symbolic proof
\end{minted}
\noindent See \correction{the} supplementary material of~\cite{daggitt_et_al:LIPIcs.FSCD.2025.2} for the full inductive symbolic proof of this infinite time-horizon guarantee. For the purposes of this paper, the crucial step is that the symbolic proof requires the neural component of the system \texttt{controller} to satisfy the following property:
\begin{minted}[fontsize=\listingsize, xleftmargin=\parindent]{agda}
safe : ∀ x → | x ! 0 | ≤ 3.25 ⋀ | x ! 1 | ≤ 3.25
           → | controller x + 2 (x ! 0) - x ! 1 | < 1.25
\end{minted}
This predicate refers \emph{only} to the neural network controller, and thus is our neural specification~$\psi$.

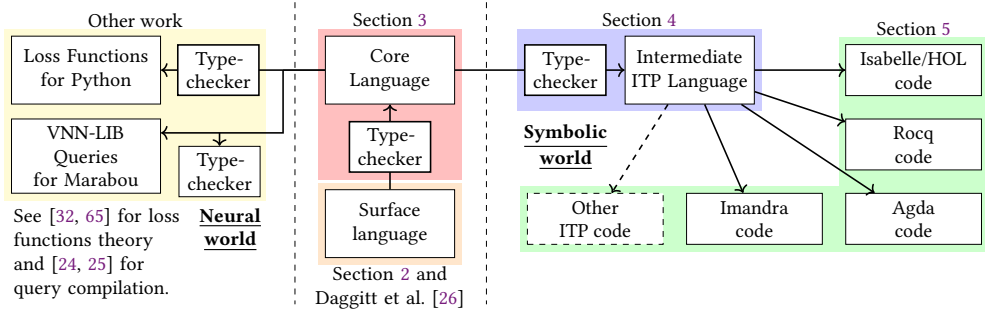
\begin{figure}[t]
    \centering
    \resizebox{0.95\linewidth}{!}{
        \begin{tikzpicture}[
  scale=0.6,
  node distance=2.2cm and 2.2cm,
  box/.style={
    draw,
    rectangle,
    minimum width=2.2cm,
    text width=2.05cm,
    minimum height=1.2cm,
    align=center,
    fill=white
  },
  tool/.style={
    draw,
    rectangle,
    minimum width=2.4cm,
    minimum height=0.9cm,
    align=center,
    fill=white
  },
  arrow/.style={->, thick}
]

\node[box] (surface) {Surface \\ language};
\node[box, above=1.5cm of surface] (core) {Core \\ Language};
\node[box, right=3.0cm of core] (intermediate) {Intermediate \\ ITP Language};

\coordinate (branch) at ($(core) + (-3.15, 0)$);
\draw[] (core) -- (branch);

\node[box, text width=2.4cm, left=2.9cm of core] (training) {Loss Functions \\ for Python};
\draw[arrow] (core) -- node(tti) [text width=1.2cm,align=center, fill=white, draw, pos=0.65]{Type-checker} (training);


\node[box, text width=2.4cm, below=0.25cm of training] (queries) {VNN-LIB Queries \\ for Marabou};

\draw[arrow] (branch) |- node(branch2) [pos=0.76] {} ([yshift=1.8cm]queries);
\coordinate (branch3) at (branch2);
\node [text width=1.2cm, align=center, fill=white, draw, below=0.1cm of branch2] (ttq) {Type-checker};
\draw[arrow] (branch3) -- (ttq);

\node[tool, right=1.6cm of intermediate] (isabelle) {Isabelle/HOL \\ code};
\node[tool, below=0.4cm of isabelle] (rocq) {Rocq \\ code};
\node[tool, below=0.4cm of rocq] (agda) {Agda \\ code};
\node[tool, left=0.4cm of agda] (imandra) {Imandra \\ code};
\node[tool, dashed, left=0.4cm of imandra] (other) {Other \\ ITP code};

\draw[arrow] (surface) -- node(ttc) [text width=1.2cm,align=center, fill=white, draw]{Type-checker} (core);
\draw[arrow] (core) -- node(cti)[text width=1.2cm,align=center,pos=0.65,fill=white, draw]{Type-checker} (intermediate);

\draw[arrow,dashed] (intermediate) -- (other);
\draw[arrow] (intermediate) -- (agda);
\draw[arrow] (intermediate) -- (rocq);
\draw[arrow] (intermediate) -- (isabelle);
\draw[arrow] (intermediate) -- (imandra);

\draw[dashed] ($(core)!0.55!(cti) + (0,2)$) -- ($(core)!0.55!(cti) + (0,-6.9)$);

\draw[dashed] ($(core) + (-2.8,2)$) -- ($(core) + (-2.8,-6.9)$);

\begin{pgfonlayer}{background}
\node[draw=none, fit=(surface), fill=orange!20, label={[anchor=south, text width=3cm, align=center, yshift=-0.9cm]south:Section~\ref{sec:vehicle-surface} and \\ \citet{daggitt_et_al:LIPIcs.FSCD.2025.2}}] (surface_outer) {};

\node[draw=none, fit=(core)(ttc), fill=pink, label={[anchor=north, yshift=4mm]north:Section~\ref{sec:vehicle-core}}] (core_outer) {};

\node[draw=none, fit=(intermediate)(cti), fill=blue!20, label={[anchor=north, yshift=4mm]north:Section~\ref{section:intermediate-language}}] (intermediate_outer) {};

\node[draw=none, text width=5cm, fit=(isabelle) (agda), fill=green!20, label={[anchor=north, yshift=4mm]north:Section~\ref{section:itp-implementations}}] (itp_outer1) {};
\node[draw=none, fit=(agda)(other), fill=green!20] (itp_outer2) {};

\node[draw=none, fit=(queries)(training)(tti), fill=yellow!20, label={[anchor=north, yshift=4mm]north:Other work}] (itp_outer1) {};

\node[right=1.5cm of ttc, text width=1.5cm, align=center] (symbolic) {\textbf{\underline{Symbolic} \\ \underline{world}}};

\node[left=0.8cm of surface, text width=1.5cm, align=center, yshift=-0.1cm] (neural) {\textbf{\underline{Neural} \\ \underline{world}}};

\node[left=1.45cm of neural, text width=3cm, anchor=north, yshift=0.5cm] { See \cite{fischer2019dl2,slusarz2022differentiable} for loss functions theory and \cite{DaggittAKKA23,daggitt2024efficient} for query compilation.};
\end{pgfonlayer}
\end{tikzpicture}
    }
    \caption{The internal architecture of \vehicle and the structure of this paper.}
    \label{fig:vehicle-structure}
\end{figure}

\subsection{\vehicle{} Surface Language}


As shown in Figure~\ref{fig:vehicle-structure}, the core purpose of \vehicle{} is to provide an external language in which users can write the type of the neural network $\tau$ and the neural specification $\psi$, use it to train and verify a network, and then export $\tau$ and $\psi$ to the ITP of their choice.
The syntax for the \vehicle surface language is shown in Figure~\ref{fig:surface-language} and the \vehicle specification $\psi$ for the car example can be written as shown in~\autoref{listing:car-controller-spec}.

\begin{figure}[h]
    \centering
    \input{figures/surface-syntax} 
    \caption{Grammar for \correction{surface} language of the \vehicle{} DSL. \correction{This grammar is a restricted fragment of both the grammar in the full implementation and that presented in~\citet{daggitt_et_al:LIPIcs.FSCD.2025.2}, sufficient to illustrate the challenges of compilation discussed in this paper.}}
    \label{fig:surface-language}
\end{figure}

\begin{listing}[h]
\begin{minted}[fontsize=\listingsize]{vehicle}
type InputVector  = Tensor Real [2];   sensor1 = 0;    sensor2 = 1
type OutputVector = Tensor Real [1];   velocity = 0

@network controller : InputVector -> OutputVector

safeInput : InputVector -> Bool
safeInput x = -3.25 <= x ! sensor1 <= 3.25 and -3.25 <= x ! sensor2 <= 3.25

safeOutput : InputVector -> Bool
safeOutput x = -1.25 < controller x ! velocity + 2*(x ! sensor1) - (x ! sensor2) < 1.25

@property safe = forall x . safeInput x => safeOutput x
\end{minted} 
\caption{The \vehicle{} specification for the car’s neural network controller.}
\label{listing:car-controller-spec}
\end{listing}

\noindent The property \mintinline{vehicle}{safe} in the \vehicle{} specification is semantically identical to the Agda lemma \mintinline{agda}{safe} in Section~\ref{sec:wind-controller-example}. Given this specification \vehicle{} can be used to train a neural network, via PyTorch, and then verify that the resulting file ``model.onnx'' satisfies the specification, via the neural solver Marabou~\cite{wu2024marabou}. Once this has been achieved, \vehicle{} aims to generate Agda interface code for the network declaration and specification of the form:

\vspace{0.5em}
\noindent \begin{minipage}[t]{0.49\textwidth}
\begin{minted}[fontsize=\listingsize]{agda}
controller : Tensor ℝ [2] → Tensor ℝ [1]
controller = callNetwork "model.onnx"
\end{minted}
\end{minipage}%
\hspace{-0.2em}
\vrule
\hspace{0.2em}
\begin{minipage}[t]{0.5\textwidth}
\begin{minted}[fontsize=\listingsize]{agda}
safe : forall x. SafeInput x → SafeOutput x
safe = checkVehicleProperty "spec.vclp"
\end{minted}
\end{minipage}
\vspace{0.5em}

\noindent where \mintinline{agda}{callNetwork} is a macro that invokes the network and \mintinline{agda}{checkVehicleProperty} is a macro that calls back to the \vehicle compiler to check the status of the verification result. In this case, the \vehicle compiler consults a cache that contains the location and a hash of the neural network model file used during verification, which it then uses to check the integrity of the proof avoiding reperforming verification of the network and disrupting the interactivity of the ITP.

\looseness=-1
\correction{
\vehicle{} 
goes against the prevailing trend in the verification community, where the vast majority of work integrating automated solvers into verification workflows in recent years has focused on implementing functional specification languages directly in the host language. Specifications are then compiled down to queries either by native macros or by the compiler itself. Examples are SMT or ASP solvers working underneath the hood in Lean~\cite{QianCBA25}, Isabelle/HOL~\cite{LachnittFBJA0SB25}, SMT-Coq~\cite{BarbosaK0VTB23}, F*~\cite{swamy2016dependent}, Liquid Haskell~\cite{vazou2016liquid}, Dafny~\cite{leino2010dafny} etc.}
The obvious question is why not follow this path and implement \vehicle{} itself as an embedded compiler in the ITP, thereby allowing us to represent the neural network $\network$ and specification $\psi$ in a single trusted logical framework? We argue that unfortunately this is infeasible for many reasons:
\begin{enumerate}
\item Even representing realistic neural networks in an ITP appears impractical. For example, the PyTorch~\cite{Ansel_PyTorch_2_Faster_2024} format has 3500 different operators and a trained network may have millions of parameters. Prior attempts at representing and verifying a concrete trained neural network in
Rocq~\cite{bagnall2019certifying}, Isabelle/HOL~\cite{brucker2023verifying}, and Imandra~\cite{desmartin2022checkinn} have all reported issues with scalability of type checking and verification, and only support 3~operators.

\item Training the network to obey the specification $\psi$ requires interfacing with state-of-the-art 
\correction{machine learning libraries (e.g.\ PyTorch, JAX, Tensorflow)}. We are unaware of any ITP that currently has support for this.

\item If the network's input tensors are large, even compiling $\psi$ down to neural network queries may be prohibitively computationally expensive to perform in an ITP~\cite{daggitt2024efficient}.

\item In NCPS systems, the final verified neural network often needs to be deployed on custom hardware such as ASICs or FPGAs. Again, there is no support for this workflow in any ITP that we are aware of.

\item Even if the above limitations were overcome, implementing the training and verification infrastructure in a single ITP would prevent it from being reused in other ITPs.

\item We hope that machine learning practioners will be interested in using the training and verification backends of \vehicle{} even if they do not wish to export their results to an ITP. Unfortunately, we believe they are unlikely to do so if they are required to work in an ITP instead of with the Python bindings provided by \vehicle{}.
\end{enumerate} 
Therefore we argue that \vehicle{}'s approach provides the most practical architecture for NCPS verification given the above limitations of current ITP technology.
However, this approach does introduce two unavoidable trade-offs.
First, \correction{it enlarges the trusted computing base of the ITP to include the Vehicle compiler, and therefore there is a possibility of introducing} a semantic mismatch between \vehicle{} and the target ITP. While this concern is valid, the \vehicle{} language is deliberately minimal compared to a full ITP logic and, as described in Section~\ref{sec:vehicle-core}, \correction{and specifications} are monomorphised prior to export to further simplify its semantics. We therefore argue that the risk of semantic divergence in practice is low.
Second, our approach treats the neural network as an abstract component and does not permit reasoning about its internal structure within the ITP.
However, we are unaware of practical NCPS verification scenarios that require such internal reasoning.
Moreover, if a NCPS verification task depends on semantically meaningful internal structure, we hypothesise that the network itself is decomposable into explicitly specified sub-components.



\section{Implementation of \vehicle{} Core Language}
\label{sec:vehicle-core}

This section contains the first technical description of the Haskell implementation of the core \vehicle language and its type-checker.
One key challenge is that the three backends of the \vehicle compiler, shown in Figure~\ref{fig:vehicle-structure}, each require unique semantic analyses of the neural specification~$\psi$:
\begin{enumerate}
    \item to compile the specification to a loss function, the training backend must determine which parts are differentiable.
    \item \correction{to provide actionable error messages to the user}, the solver backend must determine which parts of the specification are non-linear or contain alternating quantifiers; see~\cite{DaggittAKKA23}.
    \item to compile the specification to code for dependently typed ITPs, the ITP backend must determine which Boolean expressions in the specification are decidable internally in the ITP and which are not; see Section~\ref{section:intermediate-language}.
\end{enumerate}
In \vehicle, all of these analyses are implemented as auxiliary type-checking passes over the specification.
With this in mind, the core Haskell AST datatypes in \vehicle are defined as follows:
\vspace{0.5em}

\noindent \begin{minipage}[t]{0.53\linewidth}
\begin{minted}[fontsize=\listingsize]{haskell}
data Spec builtin = Spec [Decl builtin]

data Decl builtin
 = Function  Name (Expr builtin) (Expr builtin)
 | Network   Name (Expr builtin)
 | Parameter Name (Expr builtin)
 | Property  Name (Expr builtin)
\end{minted}
\end{minipage}%
\begin{minipage}[t]{0.47\linewidth}
\begin{minted}[fontsize=\listingsize]{haskell}
data Expr builtin
 = Type
 | Pi Name (Expr builtin) (Expr builtin)
 | Lam Name (Expr builtin) (Expr builtin)
 | App (Expr builtin) (Expr builtin)
 | Var Ix
 | Builtin builtin
\end{minted}
\end{minipage}
\vspace{1.5em}

\noindent By instantiating the \mintinline{haskell}{builtin} type parameter with different datatypes, we allow the language and type-system to be extended or restricted as appropriate in each backend. Note that we do not require the full power of extensible datatypes described in ``\correction{Data types à la carte}''~\cite{swierstra2008data} as we only need to extend the AST statically in a finite set of known ways.

Whilst parsing the specification, the AST is instantiated with the \mintinline{haskell}{StandardBuiltin} datatype:

\vspace{0.7em}
\begin{minipage}{1\textwidth}
\begin{minted}[fontsize=\listingsize]{haskell}
data StandardBuiltin
\end{minted}
\vspace{0.2em}
  \begin{minipage}{0.535\textwidth}
\begin{minted}[fontsize=\listingsize]{haskell}
  = Nat         -- Type
  | NLit Int    -- Nat
  | List        -- Type -> Type
  | Nil         -- List t
  | Cons        -- t -> List t -> List t
  | Real        -- Type
\end{minted}
  \end{minipage}%
  \begin{minipage}{0.465\textwidth}
\begin{minted}[fontsize=\listingsize]{haskell}
  | Bool      -- Type
  | BLit Bool -- Bool
  | Not       -- Bool -> Bool
  | And       -- Bool -> Bool -> Bool
  | Index     -- Nat -> Type 
  | If        -- Bool -> t -> t -> t
\end{minted}
  \end{minipage}
\begin{minted}[fontsize=\listingsize]{haskell}
  | Tensor      -- Type -> List Nat -> Type
  | TLit Tensor -- Tensor Real (size t)
  | Forall      -- (Tensor Real ds -> Bool) -> Bool
  | Stack       -- Tensor Real ds -> ... -> Tensor Real ds -> Tensor Real (cons d ds)
  | Lookup      -- Tensor Real (cons d ds) -> Index Real d -> Tensor Real ds
  | Leq         -- Tensor Real ds -> Tensor Real ds -> Bool 
\end{minted}
\end{minipage}
\vspace{0.3em}

\noindent Note that real literals are implemented internally as zero-dimensional tensors. 

We define the following Haskell type-class to express that an abstract \mintinline{haskell}{builtin} type can be type-checked:
\begin{minted}[fontsize=\listingsize, xleftmargin=\parindent]{haskell}
class Typable builtin where
  convertBuiltin   :: StandardBuiltin -> Expr builtin
  typeBuiltin      :: builtin -> Expr builtin  
\end{minted}
where \mintinline{haskell}{convertBuiltin} is a function for mapping \mintinline{haskell}{StandardBuiltin}s to the current builtin and \mintinline{haskell}{typeBuiltin} returns the type for each \mintinline{haskell}{builtin}.
An instance of \mintinline{haskell}{Typable} for \mintinline{haskell}{StandardBuiltin} can be defined as follows:
\begin{minted}[fontsize=\listingsize, xleftmargin=\parindent]{haskell}
instance Typable StandardBuiltin where
  convertBuiltin :: StandardBuiltin -> Expr StandardBuiltin
  convertBuiltin b = Builtin b

  typeBuiltin :: StandardBuiltin -> Expr StandardBuiltin
  typeBuiltin b = case b of
    ...  -- See comments on StandardBuiltin datatype declaration above
\end{minted}

\noindent We then define a function for type-checking declarations:
\begin{minted}[fontsize=\listingsize, xleftmargin=\parindent]{haskell}
checkDecl :: (TCM m, Typable builtin) => 
             Decl builtin -> m (Decl builtin, [Constraint builtin]) 
checkDecl = ... -- standard bidirectional type-checker
\end{minted}
This implements the type-system shown in Figure~\ref{fig:core-types}, where \mintinline{haskell}{TCM m} is a standard type-checking monad. 
We omit describing the details of the implementation, as it is a standard modern bidirectional type-checker for dependent-types (see \citet{loh2010tutorial} for details), where unknown types are represented using meta-variables and expressions are either checked against an expected type or a new type is synthesised.
It also supports implicit and instance arguments, and the bidirectional type-checking pass generates a list of unification and instance/type-class constraints that need to be resolved by dedicated unification and instance solvers.

\begin{figure}[t]
    \centering
    \input{figures/core-types}
    \caption{The standard dependent-type system for the core \vehicle{} lambda calculus.}
    \label{fig:core-types}
\end{figure}

The full procedure for type-checking a declaration in \vehicle{} is then defined as follows:
\begin{minted}[fontsize=\listingsize, xleftmargin=\parindent, linenos]{haskell}
typeDecl :: (TCM m, TypableBuiltin builtin) => Decl StandardBuiltin -> Decl builtin
typeDecl decl = do
  convertedDecl <- traverse convertBuiltin decl
  (checkedDecl, constraints) <- checkDecl convertedDecl
  (solutions, unsolvedConstraints) <- solveConstraints constraints
  substDecl <- substitute solutions checkedDecl
  generalisedDecl <- generaliseOver unsolvedConstraints substDecl
  return generalisedDecl
\end{minted}
First, the declaration is converted to use the current set of builtins before being bidirectionally type-checked, which generates a set of constraints and meta-variables to solve.
After attempting to solve the constraints, a standard constraint generalisation
procedure is
applied to pi-abstract over
unsolved meta-variables~\cite{jones2003qualified}. We present a concrete example of these steps
in Section~\ref{section:intermediate-language}.

The full procedure for type-checking a specification is as follows:
\begin{minted}[fontsize=\listingsize, xleftmargin=\parindent, linenos]{haskell}
typeSpec :: (TCM m, TypableBuiltin builtin) => Spec StandardBuiltin -> Spec builtin
typeSpec (Spec decls) = do
  typedSpec <- traverseWhileUpdatingCtx typeDecl decls
  monomorphisedSpec <- monomorphise typedSpec
  return monomorphisedSpec
\end{minted}
First, each declaration is type-checked in turn and then subsequently the whole specification undergoes a monomorphisation pass, which (i) specialises polymorphic definitions at each concrete type at which they are used and (ii) eliminates any type-class operations.
The resulting program contains only concrete types and explicit operations, simplifying later compilation stages and facilitating translation to the training and verification backends where parametric polymorphism and type classes are unsupported.

The standard type-checking pass that every specification undergoes regardless of whether it is being used for training or verification or being exported to an ITP can then be defined simply as:
\begin{minted}[fontsize=\listingsize, xleftmargin=\parindent]{haskell}
typeCoreSpec :: (TCM m) => Spec StandardBuiltin -> Spec StandardBuiltin
typeCoreSpec = typeSpec
\end{minted}

\section{ITP Intermediate Language}
\label{section:intermediate-language}

When translating a specification to a dependently-typed ITP such as Agda, Rocq or Lean, for each Boolean subexpression we must decide whether to lift it to a proposition or whether it can remain at the Boolean level as there exists a decision procedure internal to the ITP capable of resolving it.
\correction{To illustrate this, consider the example shown in Figure~\ref{fig:decidability-example}.
In particular, notice that there can be no decision procedure for \mintinline{vehicle}{forall x. calc x <= f x} in the definition of \mintinline{vehicle}{safe} internally in the ITP, as there is no representation of \mintinline{vehicle}{f} in the ITP and the ITP does not have access to the neural solver. Therefore it needs to be compiled to a proposition, and so the surrounding application of \mintinline{vehicle}{neg} needs to be compiled to proposition-level negation. In contrast, the expression \mintinline{vehicle}{neg (x <= 0)} in the definition of \mintinline{vehicle}{calc} is in the body of an \mintinline{vehicle}{if} and therefore \mintinline{vehicle}{neg} at that location must be compiled to Boolean negation. Finally, in the unshown case where the user writes a quantifier inside an \mintinline{vehicle}{if} condition then \vehicle{} should be able to error gracefully rather than generating invalid ITP code.}

\begin{figure}[t]

\begin{subfigure}{0.49\textwidth}
  \begin{minted}[fontsize=\listingsize]{vehicle}
@network
f : Real -> Real

neg : Bool -> Bool
neg b = not b




calc : Real -> Real
calc x = if neg (x <= 0) then 0 else 1

@property
safe = neg (forall x . calc x <= f x)
\end{minted}
\caption{Original \vehicle{} specification}
\label{fig:decidability-example-vehicle}
\end{subfigure}
\hfill
\begin{subfigure}{0.49\textwidth}
    \begin{minted}[fontsize=\listingsize]{agda}
f : Real → Real
f = callNetwork "path/to/model.onnx"

negBool : Bool → Bool
negBool b = not b

negProp : Set → Set
negProp b = ¬ b

calc : Real → Real
calc x = if negBool (x ≤ᵇ 0) then 0 else 1

safe : negProp (∀ x . calc x ≤ f x)
safe = checkNetworkSpecification
    \end{minted}
    \caption{Compiled Agda specification}
    \label{fig:decidability-example-agda}
\end{subfigure}

\caption{A \vehicle{} specification where Boolean expressions must be specialised before compiling to the ITP. }
\label{fig:decidability-example}
\end{figure}

\correction{
One solution would be to make the distinction between booleans and propositions explicit in \vehicle{}'s surface language.
However we want to avoid this, as the distinction may be unfamiliar to machine learning users. 
Instead, this section describes a compiler pass capable of deciding if a given Boolean expression in a surface specification language must be resolved to a Boolean or Proposition in the extracted ITP code. In particular, we show how we can reuse the standard bidirectional type-checker and design a method for automatic generalisation to do so. To the best of our knowledge, a compiler pass implemented here is the first of its kind in the functional programming literature.
}

\subsection{Implementation}

To do so, we define a new datatype that extends \mintinline{haskell}{StandardBuiltin} with the corresponding type-level counter-parts of the supported Boolean operations.

\vspace{0.25em}
\noindent \begin{minted}[fontsize=\listingsize, xleftmargin=\parindent]{haskell}
data DecidabilityBuiltin
  = Standard StandardBuiltin
  | TrueType    -- Type
  | FalseType   -- Type
  | AndType     -- Type -> Type -> Type
  | ImpliesType -- Type -> Type -> Type
  | LeqType     -- Tensor Real ds -> Tensor Real ds -> Type
\end{minted}
\vspace{0.25em}

\noindent The type \mintinline{haskell}{Spec DecidabilityBuiltin} will become the intermediate ITP language shown in Figure~\ref{fig:vehicle-structure}.
The question is how we can define an instance of \mintinline{haskell}{Typable DecidabilityBuiltin} that allows us to reuse \mintinline{haskell}{typeSpec} to move from the \vehicle{} specification in Figure~\ref{fig:decidability-example-vehicle} to a representation closer to the desired Agda code in Figure~\ref{fig:decidability-example-agda}.

The key trick is to define a new type-class \emph{internally in \vehicle} that has fields for the operations that can exist at either the type-level or the Boolean-level. There is deliberately no surface syntax for users to write instances in \vehicle{}, so we write it Haskell-style as follows:

\newcommand{\boolImpl}{BI}
\newcommand{\typeImpl}{TI}

\begin{minted}[fontsize=\listingsize, xleftmargin=\parindent]{haskell}
class Booleans where
  BoolTC    : Type
  TrueTC    : BoolTC
  FalseTC   : BoolTC
  AndTC     : BoolTC -> BoolTC -> BoolTC
  ImpliesTC : BoolTC -> BoolTC -> BoolTC
  LeqTC     : Tensor Real ds -> Tensor Real ds -> BoolTC
\end{minted}

\noindent and we then define the two instances of this type-class as:

\vspace{0.25em}
\begin{minipage}[t]{0.35\textwidth}
  \begin{minted}[fontsize=\listingsize]{haskell}
instance Booleans where
  BoolTC    = Bool
  TrueTC    = True
  FalseTC   = False
  AndTC     = And
  ImpliesTC = Or
  LeqTC     = Leq
  \end{minted}
\end{minipage}
\begin{minipage}[t]{0.3\textwidth}
\begin{minted}[fontsize=\listingsize]{haskell}
instance Booleans where
  BoolTC    = Type
  TrueTC    = TrueType
  FalseTC   = FalseType
  AndTC     = AndType
  ImpliesTC = ImpliesType
  LeqTC     = LeqType
\end{minted}
\end{minipage}
\vspace{0.55em}

\noindent Note that these instances have no type parameters and therefore are overlapping, which is disallowed by default in languages such as Haskell and Agda (see \citet{morris2010instance} for problems that can arise). However, as users cannot define instances directly in the \vehicle{} surface language, we can support overlapping instances in the \vehicle{} compiler without degrading the user experience. 


Next, we define a conversion function from \mintinline{haskell}{StandardBuiltin} to \mintinline{haskell}{DecidabilityBuiltin} that uses the above definition to replace  Boolean operations with their overloaded type-class equivalent:
\begin{minted}[fontsize=\listingsize, xleftmargin=\parindent]{haskell}
convertToDecidabilityBuiltin :: StandardBuiltin => Expr DecidabilityBuiltin
convertToDecidabilityBuiltin b = case b of
  Bool       -> BoolTC
  BLit True  -> TrueTC
  BLit False -> FalseTC
  And        -> AndTC
  Implies    -> ImpliesTC
  Leq        -> LeqTC
  _          -> Builtin $ Standard b
\end{minted}
and define the types for \mintinline{haskell}{DecidabilityBuiltin} as follows:   
\begin{minted}[fontsize=\listingsize, xleftmargin=\parindent]{haskell}
typeDecidabilityBuiltin :: DecidabilityBuiltin -> Expr DecidabilityBuiltin
typeDecidabilityBuiltin = \case
  StandardBuiltin Forall -> ... -- (t -> Type) -> Type
  StandardBuiltin b      -> typeBuiltin b
  _                      -> ... -- See comments on declaration of `DecidabilityBuiltin'
\end{minted}
where we overwrite the types of \mintinline{haskell}{Forall} to ensure it is always compiled to the type-level.

Given these two functions, the \mintinline{haskell}{Typable} instance for \mintinline{haskell}{DecidabilityBuiltin} can be defined as:
\begin{minted}[fontsize=\listingsize, xleftmargin=\parindent]{haskell}
data Typable DecidabilityBuiltin
  convertBuiltin = convertToDecidabilityBuiltin
  typeBuiltin    = typeDecidabilityBuiltin
\end{minted}
\noindent This \mintinline{haskell}{Typable} instance allow the existing instance constraint solver in the type-checker to determine the correct instance of \mintinline{vehicle}{Booleans} to use at each location in the specification, and therefore the translation of the core language to the intermediate representation suitable for ITPs can be implemented simply by delegating to \mintinline{haskell}{typeSpec} which was defined in Section~\ref{sec:vehicle-core}:
\begin{minted}[fontsize=\listingsize, xleftmargin=\parindent]{haskell}
toIntermediateITPLang :: TCM m => Spec StandardBuiltin -> m (Spec DecidabilityBuiltin)
toIntermediateITPLang = typeSpec
\end{minted}

\subsection{Worked Example}

To illustrate this claim, we step through the operation of \mintinline{haskell}{toIntermediateITPLang} in the \vehicle{} specification in Figure~\ref{fig:decidability-example-vehicle}. Note that for readability, we use the surface \vehicle{} syntax rather than the core syntax. 
First, we consider the operation of \mintinline{haskell}{typeDecl} on the \vehicle{} function \mintinline{vehicle}{neg}.  Initially, the following code is obtained after the builtins have been converted in Line 3 of \mintinline{haskell}{typeDecl}:
\begin{minted}[fontsize=\listingsize, xleftmargin=\parindent]{vehicle}
neg : BoolTC -> BoolTC
neg b = NotTC b
\end{minted}
At the end of constraint solving in Line 6 of \mintinline{haskell}{typeDecl} there is insufficient information to work out which implementation of \mintinline{vehicle}{Booleans} to use and so we end up with the following:
\begin{minted}[fontsize=\listingsize, xleftmargin=\parindent]{vehicle}
neg : BoolTC {{?0}} -> BoolTC {{?0}}
neg {{?0}} b = NotTC {{?0}} b
\end{minted}
where \texttt{\{\{..\}\}} is the standard double braces notation for instance arguments, and meta-variable \mintinline{haskell}{?0} represents the unsolved \mintinline{haskell}{Booleans} instance. The application of constraint generalisation in Line~7 of \mintinline{haskell}{typeDecl} will cause this unsolved meta-variable to be universally quantified and prependend to the type to obtain the final generalised version of the function that works over either at the level of Booleans or at the level of types:
\begin{minted}[fontsize=\listingsize, xleftmargin=\parindent]{vehicle}
neg : {{i : Booleans}} -> BoolTC {{i}} -> BoolTC {{i}}
neg {{i}} b = NotTC {{i}} b
\end{minted}
However when processing the declaration of \mintinline{haskell}{calc}, the type-checker \emph{will} be able to resolve which implementation to use, as according to the type system the condition of an ``if'' must be of type Bool and so the neg here must use the plain Boolean instance (denoted here by BI):
\begin{minted}[fontsize=\listingsize, xleftmargin=\parindent]{vehicle}
calc : Real -> Real
calc x = if neg {{BI}} (x >= 0) then 0 else 1
\end{minted}
Similarly, the type-checker will be able to deduce that the call to \mintinline{haskell}{neg} must use the type instance (denoted here by TI) due to the type of \mintinline{haskell}{forall}: 
\begin{minted}[fontsize=\listingsize, xleftmargin=\parindent]{vehicle}
safe : Type
safe = neg {{TI}} (forall x . calc x <= f x)
\end{minted}
The monomorphisation pass in \mintinline{haskell}{checkSpec} will then specialise the newly polymoprhic \mintinline{vehicle}{neg} declaration into two separate definitions, giving us the following final code:

\begin{minipage}{\textwidth}
  \vspace{0.9em}
\begin{minipage}[t]{0.30\textwidth}
\begin{minted}[fontsize=\listingsize]{vehicle}
negBool : Bool -> Bool
negBool b = not b

negType : Type -> Type
negType b = notType b
\end{minted}
\end{minipage}
\begin{minipage}[t]{0.45\textwidth}
\begin{minted}[fontsize=\listingsize]{vehicle}
calc : Real -> Real
calc x = if negBool (x >= 0) then 0 else 1

safe : Type
safe = negType (forall x . calc x <= f x)
\end{minted}
\end{minipage}
  \vspace{0.9em}
\end{minipage}
This code is now a simple conversion of the builtins away from the final desired Agda code shown in Figure~\ref{fig:decidability-example-agda} (or indeed any other ITP). We now present a proof of soundness of this translation.
\begin{theorem}[Soundness of ITP Translation]
\label{thm:soundness}
Let $s$ be a well-typed specification, then if \mintinline{haskell}{toIntermediateITPLang} terminates on $s$ successfully yielding $s'$, then $s'$ is also well-typed and $s = s'$ everywhere except for boolean expressions, where every boolean expression in $s$ corresponds to a set of one or more semantically equivalent expressions in $s'$.
\end{theorem}

\begin{proof}[Proof]
Assuming that the implementation of the bidirectional type-checker, constraint solvers and monomorphisation are sound, then we observe that the translation is the composition of three distinct semantic-preserving stages. First, the initial conversion pass overloads only the Boolean builtins via the \mintinline{vehicle}{Booleans} type-class. Next, the success of the constraint solver implies that there exists a valid partition of the declaration's boolean logic into Boolean and propositional sorts.
Lastly, monomorphisation guarantees that all polymorphic definitions (e.g.,
\mintinline{vehicle}{neg} in Figure~\ref{fig:decidability-example}) are specialised, thereby ensuring
that all references to the \mintinline{vehicle}{Booleans} type-class are eliminated and resulting
in a target term $e'$ containing only concrete, disjoint \mintinline{vehicle}{Bool} or
\mintinline{vehicle}{Type}-level operations from \mintinline{haskell}{DecidabilityBuiltin}.
Thus, by construction, and, as $s$ is already well-typed, if \mintinline{haskell}{toIntermediateITPLang} succeeds, $e'$ is a well-typed instance of $e$ where the ambiguity of boolean operations has been resolved into explicit decidable or propositional forms.
\end{proof}

\section{Integration with an ITP of Your Choice}
\label{section:itp-implementations}

Given the implementation pipeline described in Section~\ref{section:intermediate-language}, unlocking compositional verification of neural-cyber-physical systems in a new ITP \emph{X} consists of the following two steps:
\begin{enumerate}
  \item creating a \vehicle{} companion library in X that contains missing definitions in X (e.g. tensors) and, optionally, a macro that calls back to the \vehicle{} compiler when checking the proof;
  \item adding a Haskell function \mintinline{haskell}{compile : Spec DecidabilityBuiltin -> String} to the \vehicle{} compiler that translates \vehicle{} specifications to code targeting the ITP \emph{X} and its companion library.
\end{enumerate}

We have performed these steps for 4 widely-used interactive theorem provers:
Agda, Rocq, Isabelle/HOL and Imandra.
Table~\ref{tab:itp-comparison} provides a comparison of each implementation. For ITPs that already contain a formalisation of tensors (i.e. Isabelle/HOL), the size of the ITP companion library is extremely small (178 LOC). Similarly, the implementation of compilation to the ITP within \vehicle{} requires less than 1000 LOC in all systems.
In contrast, \vehicle{}'s integration with neural network training and verification ecosystems amounts to ~60,000 LOC of Haskell and ~4,000 LOC of Python.
This underlines the value of integrating ITPs via \vehicle{} as opposed to developing custom infrastructure for NCPS verification in each ITP individually.

\begin{table}[t]  
\caption{Comparison of ITP backends. LOC = line of code. No* = introduced as part of this development.}
\begin{tabular}{lccccc}
  \toprule
  Interactive Theorem Prover & Agda & Rocq & Isabelle/HOL & Imandra \\
  \midrule
  Dependent types & Yes & Yes & No & No \\
  Pre-existing tensor libraries & No* & No* & Yes & No* \\
    \correction{ITP - tensor library LOC} & \correction{192} & \correction{939} & \correction{783} & \correction{1055} \\
    ITP - other LOC & 204 & \correction{219} & 178 & 684 \\
  Haskell LOC in \vehicle{} & 772 & 720 & 981 & 862 \\
  Proof cache integration & Yes & \correction{Yes} & No & No \\
  \bottomrule
\end{tabular}
\label{tab:itp-comparison}
\vspace{-1em}
\end{table}

We will now discuss the implementation in each of the 4 interactive
theorem provers, comparing the suitability of each system and its libraries for cyber-physical
verification.

\subsection{Agda}

As Agda is dependently typed, \vehicle{} specifications can be compiled directly one-to-one to Agda code. Section~\ref{sec:vehicle-surface} has already given an overview of the Agda code that is generated by the Agda backend. Therefore we will only briefly discuss the \vehicle{}-Agda companion library.
Agda was chosen as the first ITP we added to support for primarily due to familiarity with the system, rather than a belief that it is the best ITP for performing verification of cyber-physical systems.
Indeed, it has proven severely lacking in library support, as not only does the Agda Standard Library~\cite{daggittAgdaStandardLibrary2025} not contain a formalisation of ODEs or tensors, it does not even contain a standardised representation for real numbers.
Therefore, the \vehicle-Agda companion library contains a very basic formalisation of tensors, and we are forced to compile real numbers to Agda rationals. The latter is sound as long as the network does not contain any transcendental operations (e.g. tanh).
The advantage of \vehicle{} is that rather than investing serious effort in fixing these deficiencies in the Agda libraries, we can instead move to an ITP with better library support.

\subsection{Rocq}

Rocq is the second ITP we implemented support for in \vehicle.
We chose to target the MathComp~\cite{mathcomp} and MathComp Analysis~\cite{affeldt_mathcomp-analysis_2026} libraries for this backend, as they contain an extensive formalization of mathematics that goes well beyond what both the Rocq and Agda standard libraries provide.
Targeting Rocq and MathComp Analysis allows us to leverage its formalization of real numbers and analytical arguments.
As part of this paper's contribution, we complement Rocq/MathComp with a 
novel general implementation of tensors\correction{.}\footnote{Our tensor
  implementation \correction{has been} merged into MathComp~\cite{mathcomp_tensor_pr}}

\subsubsection{MathComp Integration}

The MathComp~\cite{mathcomp} and Analysis~\cite{affeldt_mathcomp-analysis_2026} libraries contain a comprehensive collection of Rocq formalization of mathematical structures.
MathComp focuses on the constructive formalization of algebraic structures.
In particular, we exploit the hierarchy of algebraic structures (i.e., monoids, semigroups, Z-modules, semirings and rings) to construct appropriate instances for our tensor library.
As we will see later, we define tensors by wrapping matrices and appropriately translating indices, so most of these interface implementations are one-liners.

\looseness=-1
Many \vehicle constructs map directly onto MathComp definitions. For example, \vehicle's \texttt{Index n} representing the finite set $\{m \in \mathbb{N} \mid m < n\}$, corresponds to MathComp ordinals (\mintinline{coq}{'I_n}).
Other \vehicle constructs lifted to MathComp include lists with operations such as folding and mapping, ordering relations, equality, and boolean operations.
Our translation allows users to leverage mature, well-tested formalizations and benefit from MathComp's extensive lemma libraries.
We translate \vehicle{} reals to the formalization of real numbers in MathComp Analysis, which are accessed via an interface and thus compatible with the constructive definition present in the Rocq standard library, which we exploit whenever we need to compute concrete inequalities.

\subsubsection{Tensor Implementation}

Tensors are fundamental to \vehicle's design: neural network declarations require tensor types for both inputs and outputs, and, as discussed in \autoref{sec:vehicle-core}, internally all scalar values are represented as zero-dimensional tensors. However, MathComp lacks a tensor formalization, necessitating a custom implementation.

We designed the tensor library to satisfy two goals: (1) support the subset of tensor operations required for \vehicle specifications, and (2) provide a general formalization suitable for eventual inclusion in MathComp itself.
The key insight is to represent tensors as matrices from MathComp's existing \texttt{algebra/matrix.v} module~\cite{mathcomp}.
A tensor $T^{i_1,\dots,i_n}_{j_1,\dots,j_m}$ with contravariant indices $i_1 \in I_1, \dots, i_n \in I_n$ and covariant indices $j_1 \in J_1, \dots, j_m \in J_m$ is encoded as a matrix $M = (r_{ij})_{1\le i \le M, 1\le j \le N}$ where $M = |J_1| \cdot \ldots \cdot |J_m|$ and $N = |I_1| \cdot \ldots \cdot |I_n|$. 
In the general tensor notation, covariant indices represent basis vectors and transform inversely under coordinate changes, while contravariant indices represent coordinate displacements and transform directly (see e.g.,~\cite{dullemond1991introduction}).
The Rocq definition is shown in \autoref{listing:rocq-tensor-def}, where \mintinline{coq}{nat ^ n} represents a finite function from \mintinline{coq}{'I_n} to \mintinline{coq}{nat} and \mintinline{coq}{'M[K]_(n, m)} is a matrix of size \mintinline{coq}{n} times \mintinline{coq}{m}.
Note that this encoding supports both covariant and contravariant indices, though \vehicle itself uses only contravariant tensors. This generality enhances the library's potential utility beyond \vehicle.

\begin{listing}
\begin{minted}[fontsize=\listingsize]{coq}
Context {n m : nat} (u_ : nat ^ n) (d_ : nat ^ m) (K : Type) (R : pzRingType).
Variant tensor_of : Type := Tensor of 'M[K]_(\prod_(i < n) u_ i, \prod_(i < m) u_ j).
Notation "''T[' R ]_ ( u_ , d_ )" := (tensor u_ d_ R).
\end{minted}
    \caption{Rocq tensor definition: tensors as matrices.}
    \label{listing:rocq-tensor-def}
\end{listing}

The matrix-based representation allows us to inherit algebraic properties from MathComp's matrix library. For example, the Hadamard (pointwise) tensor product is defined as:
\[ (T \odot U)^{i_1,\dots,i_n}_{j_1,\dots,j_m} = T^{i_1,\dots,i_n}_{j_1,\dots,j_m} \cdot U^{i_1,\dots,i_n}_{j_1,\dots,j_m} \]
In Rocq we define this as:
\begin{minted}[fontsize=\listingsize, xleftmargin=\parindent]{coq}
Definition hmult (t u : 'T[R]_(u_, d_)) := 
  @Tensor _ _ u_ d_ R (map2_mx *%R (\val t) (\val u)).
\end{minted}
\looseness=-1
This definition uses the \mintinline{coq}{map2t} function that applies an operator (e.g., ring multiplication \mintinline{coq}{*
We then inherit all algebraic structure from the Hadamard matrix product, e.g., to instantiate the ring interface we prove distributivity:
\begin{minted}[fontsize=\listingsize, xleftmargin=\parindent]{coq}
Lemma hmultDl : left_distributive hmult +%R.
Proof. by move=> x y z; rewrite /hmult map2_mxDl. Qed.
[...]
HB.instance Definition _ := GRing.Nmodule_isPzSemiRing.Build                  
  'T[R] hmultA hmul1t hmult1 hmultDl hmultDr hmul0t hmult0.
\end{minted}

\noindent The resulting tensor library implements standard interfaces including equality, addition, subtraction, pointwise partial ordering, and the Hadamard product, all matching \vehicle's tensor semantics, as well as other definitions of general utility, including the tensor product and its bilinear properties.

\subsubsection{Neural Network Representation and Axiomatization}

Compiled \vehicle specifications represent neural networks as abstract functions with axiomatized properties. When \vehicle compiles a specification containing a \texttt{@network} declaration, the generated Rocq code includes:

\begin{enumerate}
\item An \emph{opaque function definition} for the network, hiding its internal structure. The network weights and architecture are not exposed in the generated proof script.

\item \emph{Axioms} asserting that the network satisfies the verified properties. These axioms are justified by the external verification step performed by the neural network verifier (e.g., Marabou).

\item \emph{Type signatures} specifying the network's input and output tensor dimensions.
\end{enumerate}
These components look very similar to analogous components in Agda, exemplified
in Section~\ref{sec:vehicle-surface}; further examples in Rocq are given in
Section~\ref{section:case-studies}  and Appendices A.2 and B.



\subsection{Isabelle/HOL}
\label{sec:isabelle}

The third ITP we added support for was Isabelle/HOL~\cite{DBLP:books/sp/NipkowPW02}, an important addition seeing that majority of ITP CPS formalisations were given in this prover~\cite{MuniveFGSLH24,MuniveS22,FosterMGS21,BohrerRVVP17}.
We begin this section by describing the \vehicle companion library and how we bridge \vehicle's and Isabelle/HOL's type system despite Isabelle/HOL's lack of dependent types.
Subsequently, we describe the compilation of \vehicle properties to Isabelle/HOL.

\subsubsection{\vehicle-Isabelle/HOL Companion Library}
In contrast to Rocq MathComp, Isabelle/HOL already has a formalization of Tensors available,
which is part of the AFP's \textit{Deep\_Learning} package~\cite{Deep_Learning-AFP}.
Hence, we reuse this formalization for our companion library.

\paragraph{Type System.}
The tensor type of \vehicle{} is mapped to the tensor type \texttt{\textcolor{isaTypeVarColor}{'a} tensor}.
In contrast to Agda and Rocq, Isabelle/HOL does not support dependent types.
Hence, while Isabelle/HOL's tensor type is parametric in its element type (\texttt{\textcolor{isaTypeVarColor}{'a}}) it does not support constraining tensors to a particular shape.
Consequently, Isabelle/HOL's tensor addition implementation checks shape compatibility during evaluation and returns the type's default member \mintinline{isabelle}{undefined} in case of incompatible shapes.
Fortunately, \vehicle's type checker statically analyzes shape compatibility thus ensuring that no ill-shaped tensor operations appear in the generated assumptions.

Nonetheless, we must ensure that we preserve tensor shape information for the networks represented inside Isabelle/HOL.
If a neural network's return type is only given as \texttt{\textcolor{isaTypeVarColor}{'a} tensor},
extracting a component may return \mintinline{isabelle}{undefined} unless the tensor is of suitable shape which we would not know from the type alone.
To this end, we leverage Isabelle/HOL's semantic subtyping infrastructure~\cite[Sec. 8.5.2]{DBLP:books/sp/NipkowPW02} to preserve the typing information in Isabelle/HOL despite its lack of dependent types.
This allows us to construct a type of all tensors of a particular shape. For example, we can define a real valued input vector of shape \mintinline{isabelle}{[2]} as follows:

\begin{minted}[fontsize=\listingsize, xleftmargin=\parindent]{isabelle}
typedef InputTensor = "{ a :: R tensor. (dims a) = [2] }"
  using dims_tensor_from_lookup by blast
\end{minted}


\paragraph{Type Coercion.}
In the context of larger expressions, elements of custom subtypes like \mintinline{isabelle}{InputTensor} inevitably have to interact with other tensor types.
By default, this would lead to type errors.
While this can be resolved by explicitly wrapping all variables in appropriate conversion functions, this would make expressions unwieldy.







Instead, we rely on \emph{automatic coercion} functions which allow us to instruct Isabelle/HOL to automatically convert, e.g., from \texttt{InputTensor} to \texttt{R tensor}.
Ideally, we would allow seamless bidirectional coercions between \texttt{InputTensor} and \texttt{R tensor}.
However, Isabelle/HOL does not allow for cyclic coercions.
To circumvent this, we define an (isomorphic) tensor type \texttt{\textcolor{isaTypeVarColor}{'a} FlexTensor}.
We leverage this additional type to break cycles by ensuring that all functions take \emph{inputs} of type \texttt{\textcolor{isaTypeVarColor}{'a} tensor}, and \emph{outputs} of type \texttt{\textcolor{isaTypeVarColor}{'a} FlexTensor}.
We then define automatic coercions from custom subtypes to \texttt{\textcolor{isaTypeVarColor}{'a} tensor} as well as coercions from \texttt{\textcolor{isaTypeVarColor}{'a} FlexTensor} to \texttt{\textcolor{isaTypeVarColor}{'a} tensor}.
This enables arbitrary composition of different tensor types (with matching shapes guaranteed by \vehicle).

Our companion library comprises the definition of \texttt{\textcolor{isaTypeVarColor}{'a} FlexTensor}, its automatic coercion functions, and wrapper functions for the tensor operations provided by \textit{Deep\_Learning} which satisfy the typing constraints outlined above.

\subsubsection{Compilation of \vehicle{} Specifications to Isabelle/HOL}
The high-level approach of \vehicle's Isabelle/HOL backend is similar to the strategy pursued by the Agda and Rocq backend, i.e., a given \vehicle specification file is translated into an Isabelle/HOL theory by mapping \vehicle expressions to Isabelle/HOL expressions.
Unlike Agda, \vehicle's Isabelle/HOL backend currently has no support for a proof cache. We now discuss compilation of custom tensor types, neural network properties and function definitions, in turn.

\paragraph{Compilation of Custom Tensor Types.}
Tensor types with a specific shape, like \mintinline{isabelle}{InputTensor}, are compiled to \mintinline{isabelle}{typedef} declarations as discussed above.
In addition to the definition itself, \vehicle{} also generates a proof \correction{of} the non-emptiness of the generated type which is required by Isabelle/HOL.
\vehicle{} also defines and registers the automatic coercion from \texttt{\textcolor{isaTypeVarColor}{'a} FlexTensor} to \mintinline{isabelle}{InputTensor} as well as from \mintinline{isabelle}{InputTensor} to \texttt{\textcolor{isaTypeVarColor}{'a} tensor}.
Finally, each custom type is equipped with a commonly needed simplification lemma (and its proof) enabling easy access to a tensor's shape and elements:
\begin{minted}[fontsize=\listingsize, xleftmargin=\parindent]{isabelle}
lemma InputTensor_tensor_rewrite0[simp]:
  assumes "prod_list shape = length elems"
      and "shape = (2 :: nat) # []"
    shows "(Rep_tensor (Rep_InputTensor (Abs_InputTensor (Abs_tensor (shape,elems)))))
           = (shape,elems)"
\end{minted}

\paragraph{Networks and their Properties.}
To modularize the reasoning about neural networks and their properties, we translate a given \vehicle specification file into an Isabelle/HOL \mintinline{isabelle}{locale}.
A locale combines a set of fixed parameters with assumptions about these fixed parameters which can be used in subsequent definitions, lemmas and proofs.
Locales can also be extended or combined with other locales in downstream theories to allow composition of different sets of assumptions.
In our case, we define the neural networks as fixed parameters of a locale and formalize the properties of the network provided by \vehicle as assumptions of the locale.
By relying on Isabelle/HOL's locale infrastructure, we avoid introducing the neural network properties directly as axioms in the global scope (which would be unproven from the perspective of Isabelle/HOL's core).
The resulting locale is of the following form:
\begin{minted}[fontsize=\listingsize, xleftmargin=\parindent]{isabelle}
locale WindCtrl  = 
  fixes controller :: " InputVector => OutputVector"
  assumes safe:"(\<forall> x. (safeInput controller x) --> (safeOutput controller x))"
  begin end
\end{minted}
To prove system-level properties for a given neural-cyber-physical system, Isabelle/HOL users can import the \vehicle generated file into their theory.
Subsequently, they can leverage the fixed parameters and assumptions of the locale in their proofs.
For example, we can define a scalar version of the DNN's execution within the locale's context as follows:
\begin{minted}[fontsize=\listingsize, xleftmargin=\parindent]{isabelle}
context WindCtrl
begin
    definition controller_scalar :: "real => real => real"
    where "controller_scalar p1 p2 = (
        lookup (Rep_OutputVector (controller 
            (normalise controller (Abs_InputVector (tensor_from_vec [2] [p1, p2])))))
            [WindControllerSpec.velocity])"
    lemma finalState_safe: (* ... *)  end
\end{minted}

\paragraph{Function Definitions.}
Function definitions are constructed as regular Isabelle/HOL definitions.
As seen in the example above, the definition of the neural network only exists relative to the locale \texttt{WindCtrl} while the definitions of \texttt{safeInput} and \texttt{safeOutput} happen outside the locale.
Consequently, we extend all definitions by additional parameters for the
referenced neural network (e.g. \texttt{safeInput} now has as first input the neural
network \texttt{controller}). See A.3.


\subsection{Imandra}\label{sec:imandra}

Imandra is an industrial automated reasoning system whose modelling language, IML, is based on a pure subset of OCaml~\cite{passmore2020imandra}. It provides support for infinite-precision reals~\cite{10.1007/978-3-642-38574-2_12}, which we make use of in this paper\correction{.}
Imandra is used in production for formal verification in finance by firms such as Goldman Sachs, Broadridge and Citi.
There is a history of neural network verification in Imandra: fully executable neural network representations in IML have been subjected to proofs and counterexample analysis~\cite{desmartin2022checkinn,desmartin2025certified}. In addition, verification of autonomous (cyber-physical) systems is part of Imandra's portfolio\correction{.}\footnote{\url{https://www.imandra.ai/autonomous-systems}} Thus, combining the two themes bears promise. 
Imandra's logic is classical and all propositions are expressed as executable boolean-valued predicates, so \vehicle{}'s decidability distinction maps straightforwardly.
A key feature of Imandra's core logic is a semi-decision procedure that is complete for counterexamples in a precise sense, even in the presence of recursion, nonlinearity and higher-order functions~\cite{passmore2020imandra}.

\subsubsection{Type System}
Like Isabelle, Imandra is not dependently typed.
As part of this work, we developed a tensor library for IML (1,739 LOC across 5 IML files) providing tensor operations with algebraic properties proven within Imandra.
Tensors are represented as records with well-formedness maintained by smart constructors and validity predicates:
\begin{minted}[fontsize=\listingsize, xleftmargin=\parindent]{ocaml}
type 'a tensor = { dims : int list; vec : 'a list }
let tensor_from_vec (d : int list) (v : 'a list) : 'a tensor = { dims = d; vec = v }
\end{minted}

\subsubsection{Translation Strategy}
Neural networks are declared as opaque IML functions using the \mintinline{ocaml}{[@@opaque]} attribute, 
and externally verified properties become axioms (typically expressed using tensor arithmetic) about the opaque neural network.
For example, the code below shows some of the components of the compiled wind controller specification:
\begin{minted}[fontsize=\listingsize, xleftmargin=\parindent]{ocaml}
let controller : real tensor -> real tensor = () [@@opaque]
(* ... *)
axiom safe x = safe_input x ==> safe_output x
\end{minted}

\noindent
The full compiler-generated Imandra file is shown in Appendix A.4. In future work, we envision importing \vehicle{}'s externally found proofs alongside more extensive executable representations of the 
neural network controllers expressed in Imandra's logic, so as to take advantage of Imandra's verified neural network proof checker~\cite{desmartin2025certified} 
which can import proofs from external neural network verifiers such as Marabou.

 \correction{
 To scale to more complex examples of NCPS, we will need to upgrade some of the standard Imandra libraries. 
 Imandra's real type is the full field of real numbers: input-language literals (such as 
 \mintinline{ocaml}{ 3.14 : real}) are written as rationals, but are interpreted as reals with reasoning taking place using the theory of real closed fields (and computational support for computing with these values taking place via the minimal real closed field, the real algebraic numbers). What is absent from the standard libraries is the more general apparatus of real analysis: transcendental functions such as  \mintinline{ocaml}{exp} and  \mintinline{ocaml}{ln}, limits, derivatives, and certified interval arithmetic over transcendental expressions.}

 \correction{
Because Imandra is a computational logic in which every definition is executable, the natural path is to develop analysis computationally rather than classically. E.g., \mintinline{ocaml}{exp} and  \mintinline{ocaml}{ln} would be defined as fast-converging Cauchy sequences of rationals equipped with a computable modulus of convergence; continuity and derivatives would be stated and proved in the computable-analysis style~\cite{alma991030766629705251,C09}; and range constraints over reals would be discharged by a verified interval-arithmetic tactic built on those definitions. The resulting library would be genuinely different from Rocq's MathComp Analysis, as every theorem would carry a computation that can be evaluated on concrete inputs. Though nontrivial, this is mostly an engineering programme that is very much in the spirit of Imandra's design. We expect the Imandra community to take it on over time.}


\section{Case Study: Infinite-Horizon Safety of a Continuous Neural-Cyber-Physical System}
\label{section:case-studies}



Neural networks promise
to give a viable alternative to traditional Bayesian
estimation methods in the medical domain~\cite{poweleit_artificial_2023}.
This case study concerns a medical device that administers an antibiotic Vancomycin, taking into account the patient's current drug concentration,
temperature, white blood cell count, age, and weight, using simulated data for
training.
We are assuming that the neural network was trained to
administer an optimal dose given observations based on these five parameters
and has a type
$\mathbb{R}^5\rightarrow \mathbb{R}$.
As drug dosing is a safety-critical application, we also wish to obtain
worst-case guarantees about the safety of dosing.
Using the same methodology as in Section~\ref{sec:wind-controller-example}, we will first discuss the overall system safety, and then discuss the
role of neural network verification in it.



\begin{wraptable}{r}{0.5\linewidth}
  \listingsize
  \centering
  \vspace*{-1.75em}
  \caption{\correction{Mathematical libraries used in case study. The Rocq
    proof script is approximately 1{,}200 lines and with roughly 860
    \texttt{apply}/\texttt{rewrite}/\texttt{exact} tactic invocations drawing
    on about 200 distinct library identifiers.}}
  \vspace*{-0.5em}
  \begin{tabular}{@{}ll@{}}
    \toprule
    \textbf{Library / module} & \textbf{Approx.\ uses} \\ \midrule
    \multicolumn{2}{@{}l}{MathComp}\\\midrule
    \texttt{ssralg}, \texttt{ssrnum}, \texttt{order}                                & $\sim$100 lemmas \\
    \texttt{bigop}, \texttt{seq}, \texttt{tuple}, \texttt{ssrnat}                   & $\sim$15 lemmas \\\midrule
    \multicolumn{2}{@{}l}{MathComp Analysis}\\\midrule
    Derivatives~\cite{affeldt_mathcomp-analysis_2026}                    & $\sim$25 lemmas \\
    Topology, filters, limits                                             & $\sim$20 lemmas \\
    Real exponentials and logarithms                                      & $\sim$20 lemmas \\\midrule\midrule
    \vehicle{} (tensor primitives)                                                   & $\sim$15 lemmas \\

    CoqInterval~\cite{martin-dorel_proving_2016}                                             &
    \makecell[l]{10 \texttt{interval}\\tactic invocations} \\
    \texttt{lra}                                                                              & 25 invocations \\ \bottomrule
  \end{tabular}
  \label{tab:medical-libraries}
\end{wraptable}

\looseness=-1
\correction{This case serves to discuss the situation when the larger
system verification may involve reasoning about the system dynamics explicitly:
in particular, reasoning about non-linear functions and their derivatives. Such
verification is best served by an ITP whose real-analysis libraries are mature
enough to handle the elementary calculus directly and whose numerical-bound
machinery can discharge the transcendental comparisons that arise. Our case
study draws on MathComp-Analysis~\cite{affeldt_mathcomp-analysis_2026} for
derivatives, continuity, and integration
, and on
CoqInterval~\cite{martin-dorel_proving_2016} for verified real-valued
inequalities. Beyond this experiment, the wider Rocq ecosystem already supports
richer dynamic-systems verification through LaSalle's invariance principle for
stability arguments~\cite{cohen_rouhling_itp_2017} with demonstrated case
studies (e.g., an inverted pendulum~\cite{rouhling_cpp_2018}) and recent
analytic-ODE solutions~\cite{thies_cpp_2026}. The further additions needed to
reach the state of the art in full continuous-dynamics verification are modest:
principally, a general-purpose Picard--Lindel\"of theorem and the resulting
flow theory.}


\begin{figure}[h]
  \noindent \hspace{-1.5em} \begin{subfigure}[]{0.69\textwidth}
    \begin{minted}[fontsize=\listingsize, xleftmargin=1.5em]{coq}
Definition derive (f : V -> W) a v := lim (
 (fun h => h^-1 *: ((f \o shift a) (h *: v) - f a)) @ 0^').

Notation "''D_' v f c" := (derive f c v).

Definition derive1 V (f : R -> V) (a : R) := lim (
 (fun h => h^-1 *: (f (h + a) - f a)) @ 0^').

Notation "f ^` ()" := (derive1 f).

Definition continuous_at (T U : nbhsType) (x : T) 
 (f : T -> U) := (f%function @ x --> f%function x).

Notation continuous f := (forall x, continuous_at x f).

\end{minted}
  \end{subfigure}
  \hfill
  \vrule
  \hfill
  \begin{subfigure}[]{0.3\textwidth}
    \listingsize
    \begin{align*}
      D_vf(x)=\lim_{h\rightarrow 0}\frac{f(x+hv)-f(x)}{h}\\\\\\
      f'(x)=\lim_{h\rightarrow 0}\frac{f(x+h)-f(x)}{h}\\\\\\
      \lim_{a\rightarrow x}f(a)=f(x)
    \end{align*}
  \end{subfigure}
  \vspace{-0.75em}
  \caption{MathComp Analysis definitions of multivariate derivative, differentiation and
    continuity and their mathematical counterparts.}
  \label{fig:analysis-defs}
  \Description[Introduction to MathComp notation]{
    This figure serves to show the MathComp notation and their respective
    definitions, as well as the mathematical definitions and notations. If
    covers multivariate derivatives and differentiation as defined in the
    derive.v file of MathComp analysis,  and continuity, as defined in
    classical/filter.v of MathComp analysis.}
\end{figure}

\correction{\autoref{tab:medical-libraries} summarises the Rocq libraries the medical
case study draws on, with approximate usage counts. The bulk of the proof
relies on MathComp algebra/order and the MathComp Analysis derivative layer
(including lemmas reasoning about maxima of monotone functions); bridging to
Stdlib's reals via \texttt{Rstruct}/\texttt{Rstruct\_topology} is required to
invoke CoqInterval, which in turn discharges the approximation arguments.}

\subsection{System Safety in Rocq}

\begin{wrapfigure}[15]{r}{0.5\textwidth}
  \centering
  \vspace*{-1.75em}
	\includegraphics[width=0.45\textwidth]{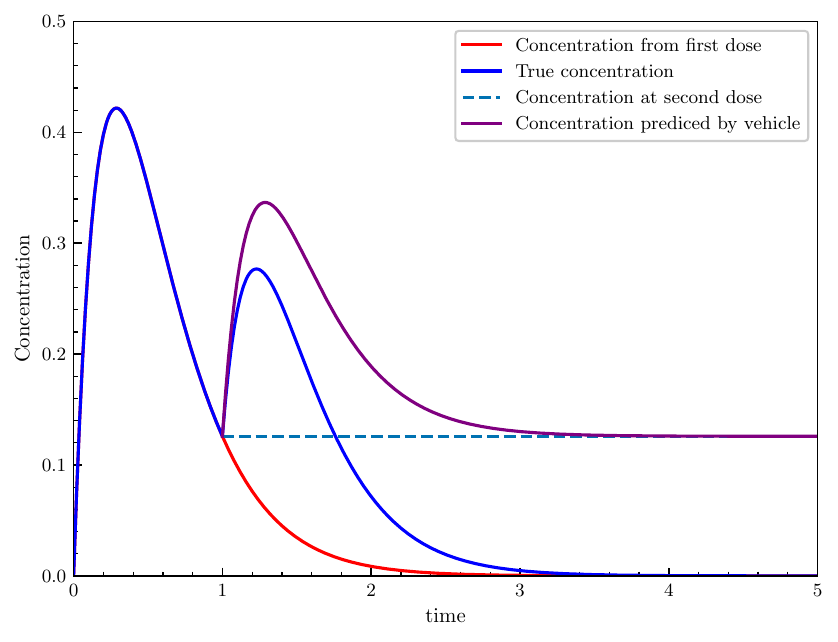}
	\vspace{-1.3em}
    \caption{A graph showing how the concentrations are interpreted by different
    parts of the verification process. The red line is equal to the blue line
    until 1, representing~\autoref{fig:single-dose} and~\autoref{fig:sum-doses}, respectively.}
  \label{fig:medical-conc}
  \Description[Graph of the different concentrations as interpreted by different
  parts of the proof]{This graph shows a blue line representing the true
    concentration of the patient, a red line which is the same as the blue line
    until the second dose is given, which it ignores and continues to decay.
    Finally, it shows a blue dashed line at the point where the second dose is
    given and a purple line that is the value of concentration predicted by \vehicle{}.}
\end{wrapfigure}
As shown in~\autoref{fig:rocq-defs}, the state consists of the patient
information. Furthermore, \texttt{n\_doses} provides
a function that returns the sequence of doses chosen by the neural network.
This is because, unlike the Wind Controller
example, the concentration at any point in time is dependent on all previous
doses, and this information cannot be obtained from \textit{only} the previous
state.
Note that \texttt{n\_doses} uses the neural network called
\texttt{network} to perform the computation. The \texttt{network} used here is in a reduced
form, where all information except the concentration has been partially applied,
making it of type $\mathbb{R}\rightarrow\mathbb{R}$. In line with the other example, it
is this occurrence of the neural network that raises concerns about the overall
system safety.

\begin{figure}[b]
  \begin{subfigure}[t]{0.53\linewidth}

\begin{minted}[fontsize=\listingsize, xleftmargin=1.5em]{coq}
Fixpoint n_doses (initial : R) (n : nat)
  : n.+1-tuple :=
  match n with
  | 0 => [:: network initial]
  | n'.+1 =>
      let Doses := n_doses initial n' in
      rcons Doses (network
      (total_conc D'.+1))%R)) end.
\end{minted}
\end{subfigure}
  \hfill
  \vrule
  \hfill
\begin{subfigure}[t]{0.46\linewidth}
    \begin{minted}[fontsize=\listingsize, xleftmargin=1.5em]{coq}
Record state := State
         { C : R
         ; T : R
         ; wbc : R
         ; age : R
         ; weight : R}.
\end{minted}
\end{subfigure}
  \Description[table functions from Rocq]{Here we define some notable functions
  from the Rocq proof, including the State record consisting of 5 reals, and
  n doses, a function that iteratively applies the network n times, and returns
  an n + 1 tuple of all the doses given}
\caption{The Rocq definition for the state and the iterative function (similar
  to \texttt{next\_state} in \autoref{sec:wind-controller-example}). Note that \texttt{network} is a
  reduced form that only takes the concentration as its
  argument.} 
\label{fig:rocq-defs}
\end{figure}




To model drug concentration after a single dose, we follow the standard
exponential decay \emph{pharmacokinetic equation}~\cite{talevi_one-compartment_2021}
shown in~\autoref{fig:single-dose}. Pharmacokinetic equations are a system of
ODEs that describe how the body processes doses of drugs. In our case, we focus
on the \emph{one-compartment} pharmacokinetic equations, which is a more simplistic
model that considers the body as a single conglomeration of blood. More intricate models
can be found in~\citet{taleviTwoCompartmentPharmacokineticModel2021}.
The
corresponding function \texttt{Concentration} takes as input the dose $D$ and time
$t$. As parameters, it has the volume of blood $V_d$ and an absorption and
elimination coefficients $k_a$ and $k_e$ that characterise, respectively, the
rates of absorption and elimination of the drug in the patient's blood \correction{as seen in Figure \ref{fig:single-dose}.}


\begin{figure}[t]
  \begin{subfigure}[]{0.53\textwidth}
    \begin{minted}[fontsize=\listingsize, xleftmargin=1.5em]{coq}
Definition Concentration (D t : R) : R :=
  ((D * Ka) / (Vd * (Ka - Ke))) *
  (expR ((-Ke) * t) - expR ((-Ka) * t)).
\end{minted}
  \end{subfigure}
  \hfill
  \vrule
  \hfill
  \begin{subfigure}[]{0.46\textwidth}
    \listingsize
    \begin{equation*}
      C(D,t)=\frac{D\cdot k_a}{V_d\cdot (k_a - k_e)}\cdot \left(e^{-k_e\cdot t}-e^{-k_a\cdot t}\right)
    \end{equation*}
  \end{subfigure}
  \caption{Pharmacokinetic equation as a concentration function relying on the absorption and elimination coefficients $k_a, k_e$, blood volume $V_d$ and dose $D$ at
    time $t$. Adapted from Equation 15 in ``One-Compartment Pharmacokinetic Model``,~\citet{talevi_one-compartment_2021}.}
  \label{fig:single-dose}
  \Description[Definition of concentration in Rocq and respective formula]{
    The mathematical definition of the one compartment pharmacokinetic equation
    and its definition in Rocq.}
\end{figure}

\subsection*{} To model drug concentration after a series of doses, we define the function \texttt{total\_conc} in~\autoref{fig:sum-doses}
that takes as a parameter \texttt{Ds}, a tuple of all doses given to the patient
at regular dosing intervals (\texttt{ttd}) and takes the sum of all concentrations from
all doses given a point in time. Note that, in the context of MathComp, an
\mintinline{coq}{n.-tuple R} is a list of length \mintinline{coq}{n} of type \mintinline{coq}{R}:

\begin{figure}[h]
  \begin{subfigure}[]{0.51\textwidth}
    \begin{minted}[fontsize=\listingsize, xleftmargin=\parindent]{coq}
Definition total_conc {n} (Ds : n.-tuple R) 
  : R -> R
  := \sum_(i < n) ((cst 0)
        \max (Concentration (tnth Ds i))
           \o (shift (-ttd * i%:R))).
\end{minted}
  \end{subfigure}
  \hfill
  \vrule
  \hfill
  \begin{subfigure}[]{0.47\textwidth}
    \vspace{-1em}
    \begin{multline*}
      C(\bar{D}, t)=\sum_{i=0}^n\text{max}\left(0, \frac{\bar{D}_i\cdot k_a}{V_d\cdot (k_a - k_e)}\right. \\
        \left. \cdot \left(e^{-k_e\cdot (t - i \cdot \mathit{ttd})}-e^{-k_a\cdot (t - i \cdot \mathit{ttd})}\right)\right)
      \end{multline*}
    \end{subfigure}
  \caption{The function for the total concentration at time $t$, where \texttt{Ds} is
    an $n$ tuple of all the doses given.}
  \label{fig:sum-doses}
  \Description[Definition of total concentration in mathematical notation and in
  Rocq]{The definition of the total concentration, defined as the sum of the max
  of 0 and the concentration in both Rocq and mathematical notation}
\end{figure}

This is an over-approximation of the actual
value of the concentration, because the exponential decay is only accounting for
the concentration from each
individual dose, and not the total concentration at any point in time. Because
our safety property is expressed as \textit{less than or equal to} a critical value, this
approximation does not impede the safety of the system.

Finally, the desired safety-critical property can be expressed in Rocq, as shown
in~\autoref{thm:doses_safe}:

\begin{theorem}[Main Safety Theorem]\label{thm:doses_safe}
For any neural network $\mathcal{N}$, the concentration of the drug in a patient's
body never reaches a critical value.
\vspace{0.5em}

\hspace{-0.5em}\begin{minipage}{0.48\textwidth}
\begin{minted}[fontsize=\listingsize]{coq}
Theorem doses_safe 
  (n : nat) (initial t : R)
  (initial_itv : 0 <= initial <= C_safe) :
    0 
    <= total_conc (n_doses initial n) t 
    <= C_safe.
\end{minted}
\end{minipage}
\hfill
  \vrule
\begin{minipage}{0.5\textwidth}
  \vspace{-1em}
  \begin{multline*}
    \forall \mathbf{x}\in\mathbb{R}^5, \forall t : \mathbb{R},
    \\
    0 \leq \sum_{i=0}^n\text{max} \left(0,\frac{\text{n\_doses}(\mathbf{x}_0, i)\cdot
        k_a}{V_d\cdot(k_a-k_e)}\cdot\right.\\
      \left.\left(e^{-k_e\cdot(t-i\cdot ttd)}-e^{-k_a\cdot(t-i\cdot ttd)}\right)\right)\leq C_{\mathit{safe}}
  \end{multline*}
\end{minipage}
  \hfill
\end{theorem}

\noindent This theorem represents $\Phi$
in~\autoref{eq:prop-impl}, and the remaining question is what should $\psi$ be?


\subsection{Neural Network Specification in \vehicle}

In order to derive the specification $\psi$ that \vehicle can verify, we restrict \texttt{ttd} to
be greater than or equal to the time of the peak concentration. This restriction
means that the concentration must be non-increasing at the point in time
when a dose is given\correction{.}\footnote{In real life
application, it is plausible that the application of neural networks can be
constrained in this way. This restriction allows the assumption that the
concentration at a dosing point is higher than any future point.}
Then, we find a specification of the network that is sufficient to prove~\autoref{thm:doses_safe}.
Because, when the neural network is applied, \correction{the} concentration is non-increasing, we can consider
only the concentration from all previous doses at that point.
Then, we consider
the output of the network as a new dose being given and compute its peak. If the
sum of the current concentration and calculated max concentration are less than
or equal to \texttt{C\_safe}, then it follows that at no point could the concentration
ever be greater than \texttt{C\_safe}. This constitutes our property $\psi$ defined as
\begin{equation}\label{eq:vcl-spec-maths}
  \forall \mathbf{x}\in\mathbb{R}^5, C(t)+C\left(\mathcal{N}(\mathbf{x}), \frac{\text{ln}\left(\frac{k_a}{k_e}\right)}{k_a-k_e}\right)\leq C_{\mathit{safe}},
\end{equation}
where $C(t)$ is the initial total concentration when the dose is given.
The \vehicle spec for $\psi$ is shown in~\autoref{lst:medical-vcl-spec} where property \texttt{safe}
encodes \autoref{eq:vcl-spec-maths}. As peak time is a potentially-irrational number, \texttt{Ke\_over},
\texttt{Ke\_under}, \texttt{Ka\_over} and \texttt{Ka\_under} are approximations of the
exponentials in~\autoref{fig:single-dose}. 

With these, \vehicle can be used to
verify~$\psi$\footnote{Due to
limitations in Marabou wrt. strict inequalities, adaptions to the specification
have had to be made. The
properties shown in the specification file are sufficient to prove the main
theorem.} by invoking the following command:

\begin{center}
\small{\texttt{vehicle verify -v Marabou -c cache -s pk.vcl -n
  pk:pk.onnx -p Ka:3.5 -p Ke:4.5 $\cdots$},} 
\end{center} 

\noindent where the \texttt{-p} arguments set the values for the parameters in the specification and the \texttt{-c} argument, which creates a cache for storing the parameter values and verification result for later retrevial.
\begin{listing}[t]
\vspace{-1.5em}
{
  \setlength{\columnsep}{-3.4cm}
\begin{multicols}{2}
\begin{minted}[numbers=left, fontsize=\listingsize, xleftmargin=1.5em]{vehicle}
type Input = Tensor Real [5]

conc   = 0
temp   = 1
wbc    = 2
age    = 3
weight = 4

type Output = Tensor Real [1]

@network
pk : Input -> Output

@parameter
Ka, Ke, Vd, C_safe, ttd, 
  Ka_over, Ka_under, 
  Ke_over, Ke_under : Real

@property
Ka_pos : Bool
Ka_pos = 0 < Ka

@property
Ke_pos : Bool
Ke_pos = 0 < Ke

@property
Ke_n_Ka : Bool
Ke_n_Ka = Ka != Ke

@property
Vd_pos : Bool
Vd_pos = 0 < Vd
\end{minted}
\columnbreak

\begin{minted}[numbers=left, firstnumber=34, fontsize=\listingsize, xleftmargin=1.5em]{vehicle}
@property
C_safe_pos : Bool
C_safe_pos = 0 < C_safe

@property
ttd_pos : Bool
ttd_pos = 0 < ttd

safeInput : Input -> Bool
safeInput x =
  0    <= x ! conc   <= C_safe and
  36.5 <= x ! temp   <= 40     and
  7.5  <= x ! wbc    <= 20     and
  18   <= x ! age    <= 89     and
  50   <= x ! weight <= 100

safeOutput : Input -> Bool
safeOutput x =
  let y = ((((normpk x) ! 0)*Ka) / (Vd*(Ka - Ke))) 
  in if Ka < Ke
   then (x!conc) + y*(Ke_under - Ka_over) <= C_safe
   else (x!conc) + y*(Ke_over - Ka_under) <= C_safe

@property
safe : Bool
safe = forall x . safeInput x => safeOutput x

nonNegOutput : Input -> Bool
nonNegOutput x = 0 <= (normpk x) ! 0

@property
nonNeg : Bool
nonNeg = forall x . safeInput x => nonNegOutput x
\end{minted}
\end{multicols}
}
\vspace{-0.8em}
\caption{The \vehicle specification file, expressing the neural network safety property. }
\label{lst:medical-vcl-spec}
\end{listing}

\looseness=-1
Our specification is parametric in the patient-specific
variables, such as $k_a$ and
$k_e$. 
In fact, the proof in Rocq only relies
on the parameters being positive and $k_a\neq k_e$. These restrictions are
expressed as properties in \vehicle{}\correction{.}\footnote{Note that \vehicle cannot
verify that \texttt{ttd} is greater than or equal to the time of peak dose,
because it is potentially irrational. Instead, the job of dealing with the
irrational parameters is delegated to Rocq.} This is a generalisation of the approach taken in the WindController example presented earlier, where we assumed that system parameters such as the sensor error would be the same for any car. Now our proofs are parametric on the patient's data. Crucially, the execution of the above command gives us a "Yes" answer: the neural network is proven safe by the neural solver Marabou.
%
The specification then is extracted to \rocq using 
\texttt{vehicle export -t Rocq -c
  cache -o Spec.v -r}.
\noindent  This sets the target language to \rocq, the cache path, the output file and a \rocq specific flag \texttt{-r}, that causes
\vehicle to extract to the constructive reals as defined in the \rocq standard
library, as opposed to the default MathComp \texttt{realType} interface. We obtain a
file translating the neural network safety property of
Listing~\ref{lst:medical-vcl-spec} in Rocq, see Appendix B.


\subsection{System Correctness}

Next, we prove in Rocq that
the properties verified by \vehicle are sufficient to prove the main safety \autoref{thm:doses_safe}:

\begin{theorem}
  \label{thm:doses-safe-restated}
  For any neural network $\mathcal{N}$, for any valid initial state $\mathbf{x}$ (omitted from
  the equation for brevity), $\psi$ (\autoref{eq:vcl-spec-maths}) implies an over-dose could never occur.

\hspace{-0.5em}\begin{minipage}{0.50\textwidth}
  \begin{minted}[fontsize=\listingsize]{coq}
Theorem doses_safe 
  (n : nat) (initial t : R)
  (initial_itv : 0 <= initial <= C_safe) :
  (forall C : R, 0 <= C <= C_safe -> 
    C + (Concentration (network C) dCdt_root) 
      <= C_safe)
  (0 <= total_conc (n_doses initial n) t 
    <= C_safe).
  \end{minted}
\end{minipage}
\hfill
  \vrule
  \hfill
  \begin{minipage}{0.45\textwidth}
  \begin{multline*}
    \mkern-18mu \forall \mathbf{x}\in\mathbb{R}^5, \forall t : \mathbb{R}, \\
    C(t)+C\left(\mathcal{N}(\mathbf{x}), \frac{\text{ln}\left(\frac{k_a}{k_e}\right)}{k_a-k_e}\right)\leq C_{\mathit{safe}} \implies\\
    0 \leq \sum_{i=0}^n\text{max} \left(0,\frac{\text{n\_doses}(\mathbf{x}_0, i)\cdot
        k_a}{V_d\cdot(k_a-k_e)}\cdot\right.\\
      \left.\left(e^{-k_e\cdot(t-i\cdot ttd)}-e^{-k_a\cdot(t-i\cdot ttd)}\right)\right)\leq C_{\mathit{safe}}
  \end{multline*}
\end{minipage}
\end{theorem}

\begin{proof}
We perform induction on $n$, the size of the tuple $Ds$ from~\autoref{fig:sum-doses}. For $n=0$, it
is trivial, as concentration is always $0$. For the inductive case,
we proceed by cases on $t$ and $S(n)\cdot \mathit{ttd}$.

\textbf{Case 1:}
    $t \leq S(n)\cdot \mathit{ttd}$. The $S(n)$\textsuperscript{th} dose has not been
    given or has been given at this instant, so it follows from the inductive
    hypothesis. This is shwon by the blue line from 0 to 1
    in~\autoref{fig:medical-conc}.

\textbf{Case 2:}  $S(n) \cdot \mathit{ttd} < t$ implies that the $S(n)$\textsuperscript{th} dose has been given and had some
time to be absorbed. We first note that the equation in  \autoref{fig:sum-doses} is
differentiable for all $t$ s.t. $S(n) \cdot \mathit{ttd} < t$, because max$(x, y)$ is differentiable
everywhere except $x=y$, and the $n$\textsuperscript{th} concentration
function is strictly
positive for all $t > S(n)\cdot \mathit{ttd}$. The sum can be split into two
components, the concentration from all previous doses (the red line
in~\autoref{fig:medical-conc}) and the newly added dose:
\begin{equation}
  \sum_{i=0}^n C_i(t)+C_{n+1}(t)\leq C_{\mathit{safe}}
\end{equation}
\textit{Note that $C_i(t)$ is equivalent to $C(D, t)$ where $D$ is the $i^{\textit{th}}$ dose given to
the patient.}

Because $C_i(t)$ is decreasing for all $S(n)\cdot \mathit{ttd} < t$,
$C_i(t)< C_i(S(n)\cdot \mathit{ttd})$, we can weaken the statement. This proves that the blue dashed line is always
greater than or equal to the red line in~\autoref{fig:medical-conc}.
\begin{equation}
  \sum_{i=0}^n C_i(S(n)\cdot \mathit{ttd})+C_{n+1}(t)\leq C_{\mathit{safe}}
\end{equation}
At $S(n)\cdot ttd$, the new concentration has no effect, so the
summation is simply the total concentration, $C(S(n)\cdot ttd)$:
\begin{equation}
  C(S(n)\cdot ttd)+C_{n+1}(t)\leq C_{\mathit{safe}}
\end{equation}
As $\frac{\text{ln}\left(\frac{k_a}{k_e}\right)}{k_a-k_e}$, the root of the derivative, is the
maximum of the function, we can similarly weaken this:
\begin{equation}
  C(S(n)\cdot ttd)+C_{n+1}\left(\frac{\text{ln}\left(\frac{k_a}{k_e}\right)}{k_a-k_e}\right)\leq C_{\mathit{safe}}
\end{equation}
Notice how the left side of the $+$ is exactly the concentration at the point
the dose is given, and the right side is exactly the peak concentration of our
new dose. This is exactly~\autoref{eq:vcl-spec-maths}! We have successfully both
derived a property $\psi$ that \vehicle can verify, and shown that
$\psi(\network)\implies \Phi(\spec(\network))$\correction{.}\footnote{Because the time of peak concentration is potentially irrational, an over-approximation is taken by \vehicle{}.}
\end{proof}

The resulting Rocq proof \correction{(submitted in supplementary materials~\cite{daggitt_2026_20529849})} is $\sim$1200 LOC, and many parts of both MathComp and Analysis are
leveraged to complete it. Derivatives are used to reason about the peak
concentration, continuity is used to reason about the main functions.
%
As part of this proof, we contributed an original addition to MathComp Analysis,
the proof that \texttt{f \textbackslash max g} is Gateaux derivable at $x$ if \texttt{f} and \texttt{g} are
Gateaux derivable and continuous at $x$, and $f(x) \neq g(x)$. The following fact
implies the derivability of \texttt{max} and \texttt{min}, as well as lemmas like
\texttt{derive\_maxl}\correction{.}\footnote{Our proof of \texttt{der\_max} and the respective
  lemmas \correction{have been} merged into MathComp Analysis.}

\begin{minted}[fontsize=\listingsize]{coq}
Fact der_max f g x v :
  f x <> g x -> derivable f x v -> derivable g x v ->
  {for x, continuous f} -> {for x, continuous g} ->
  (fun h => h^-1 *: (((f \max g) \o shift x) (h *: v) - (f \max g) x)) @
    0^' --> if f x < g x then 'D_v g x else 'D_v f x.

Lemma derive_maxl f g x v : f x > g x ->
  {for x, continuous f} -> {for x, continuous g} -> 'D_v (f \max g) x = 'D_v f x.
\end{minted}
In addition, \texttt{interval} was used
to reason about the parameters in the extract Rocq file,
to prove:
\begin{align*}
  \texttt{K\_e\_under} \leq e^{-k_e\cdot \frac{ln\left(\frac{k_a}{k_e}\right)}{k_a-k_e}} \leq \texttt{K\_e\_under}, \qquad \qquad 
  \texttt{K\_a\_under} \leq e^{-k_a\cdot \frac{ln\left(\frac{k_a}{k_e}\right)}{k_a-k_e}} \leq \texttt{K\_a\_under}
\end{align*}

\correction{ To show that the specification is satisfiable, we constructed
  a simple ReLU network with 5 inputs, 2 hidden layers of 128
  and 64 nodes respectively, and 1 output. We use the \vehicle
  Machine Learning backend to train the network using the specification. The resulting network is verified to satisfy the safety properties laid out in~\autoref{lst:medical-vcl-spec} using the \vehicle{} solver backend and is available alongside the specification in the supplementary material. Accordingly, all backends of \vehicle were used in this case study.}

In summary, this case study serves as an example of how the proposed  compositional approach enables 
infinite-horizon safety verification in domains where the system dynamics is continuous.
\correction{Table~\ref{tab:2} gives a comparison between this case study and the \emph{Wind Controller} example we considered in the previous sections. Both case studies showed the power of compositional approach to NCPS verification and inductive proofs of infinite-horizon safety properties. However, the \emph{Medical Controller}
example addresses a continuous system by leveraging existing mathematical libraries inside the ITPs.}
We
believe that the presented method is applicable to exponential
decay problems more generally and showcases how a neural network
acting in a discrete capacity can be verified in a continuous system.

\correction{
A natural question is whether this case study could be replicated in other ITPs. While we have implemented Vehicle integration for Agda, Isabelle/HOL and Imandra, the existence of the lemmas from CoqInterval, MathComp, and MathComp analysis used in this proof within each ITP is an orthogonal problem. Not all ITPs are equal in this regard: they do not exist within Agda, but are already largely present in Isabelle/HOL.
Section~\ref{sec:imandra} discusses necessary steps to formalize such libraries in an ITP from scratch using Imandra as an example.}

\correction{
\vehicle{}'s design provides a uniform interface for neural network verification that is agnostic to the chosen ITP and symbolic CPS formalisation. 
Consequently, \vehicle{} can integrate with all existing work on formalising CPS systems in ITPs, regardless of whether the symbolic system is formalized via (temporal) differential dynamic logic in Isabelle/HOL and Rocq~\cite{BohrerRVVP17} or PVS~\cite{WhiteTSM24}, via hybrid predicate transformers in Isabelle/HOL~\cite{FosterMGS21,MuniveS22,MuniveFGSLH24}, via predicate abstraction in Rocq~\cite{GeuversKSW10}, or in another framework~\cite{OuchaniKH20,Ricketts17,boyer1990use}.
We envisage users building on this substantial body of research in a manner that is best suited to their case study.
}

\begin{table}[t]
\caption{\correction{Comparison of the \emph{Wind Controller} case study and the \emph{Medical Controller} case study.}}\label{tab:2}
\setlength{\tabcolsep}{3pt}
\begin{tabular}{lcc} 
	\toprule
    NCPS Model $\backslash$ Safety Proofs 
    & Wind Controller 
    & Medical Controller
    \\
	\midrule
    NCPS composes system and neural levels 
    & \cmark 
    & \cmark 
    \\     
    Controller trained via \vehicle ML backend 
    & \xmark 
    & \cmark 
    \\     
    Entire system description given in ITP 
    & \cmark  
    & \cmark  
    \\
    Infinite horizon modelling via inductive proofs 
    & \cmark 
    & \cmark 
    \\
    Proofs concern complex continuous system dynamics 
    & \xmark 
    & \cmark 
    \\
    Proofs rely on larger mathematical library infrastructure 
    & \xmark 
    & \cmark  
    \\
    \bottomrule
\end{tabular} 
\vspace{-1em}
\end{table}

\section{Conclusions}

Verification of neural-cyber-physical systems is a promising application area for ITPs, but requires integrated reasoning about the neural component and the symbolic cyber-physical components.
\correction{
Our approach addresses the current limitations of the 
state-of-the-art tools~\cite{huang2022polar,wang2023polar,Althoff2015ARCH,10.1145/3302504.3311804,ivanov2021verisig} participating in ARCH-COMP~\cite{ARCH25}.
Since these tools are based on reachability methods, they can typically only handle safety proofs for finite time horizons, rarely exceeding a few hundred time steps.
In contrast, in this paper we have made a fundamental shift from finite to infinite-time horizon guarantees, (i.e.\ guarantees that keep the system safe at \emph{any} future point in time).
}
Concretely we have shown how our Haskell DSL \vehicle provides the first ever practical, fully compositional approach to integrating proofs about the neural components from neural solvers into symbolic proofs about the cyber-physical components in ITPs. 
This paper complements the existing solver~\cite{DaggittAKKA23,daggitt2024efficient} and machine-learning backends of \vehicle with an original technical description of the core type-checker and its ITP backend. 

We have proven that our approach allows lightweight integration with a diverse range of ITPs with and without dependent types, 
by integrating \vehicle with Agda, Rocq, Isabelle/HOL and Imandra. 
On average this integration only required ~775 lines of ITP code per ITP, although in Isabelle/HOL where formalisations of real numbers and tensors already exist, integration was achieved in as low as 178 lines.  
The biggest per-ITP effort was the development of a new tensor library for MathComp, -- also an original contribution of this paper. 
To show the benefits of \vehicle's compositional approach, we then used this library in the verification of the infinite time-horizon safety of a continuous NCPS in Rocq's MathComp, -- to our knowledge this is the first result of this kind in any ITP, and a testament to the power of compositionality in the new emerging application domain of neural-cyber-physical system verification.  



\section{Other Related and Future Work}


\noindent \textbf{Other functional DSLs for neural network verification.} The tool CAISAR~\cite{AlbertiBGGVC25}, written in Why3 and OCAML, offers a DSL for high-level neural network specifications; and compiles those specifications to a range of neural solvers. However, it has no capacity for building ITP backends, and thus does not facilitate NCPS verification tasks.

\noindent \textbf{Certificates for Neural Networks.} 
Currently the validity of the neural specifications generated by
Vehicle are not not checked by the ITP core, and therefore the ITP must trust Vehicle and the neural solvers.
We believe that recent efforts to extend neural solvers with proof production~\cite{DBLP:conf/fmcad/IsacBZK22} and to check such proofs in theorem provers~\cite{desmartin2025certified} can be seamlessly integrated with \vehicle in future.
For example, in the Isabelle/HOL backend described in Section~\ref{sec:isabelle}, the locales generated by Vehicle can be instantiated by a function and a proof that this concrete function satisfies the locale's assumptions.
Given an imported proof from a neural solver, it would be possible to obtain all results proven within the locale, now \emph{instantiated} for the concrete definition and checkable by Isabelle itself.


\noindent \textbf{Neural Networks as Certificates.}
An alternative to verifying reachability properties is to learn neural network-based certificates~\cite{DBLP:conf/nips/ChangRG19,DBLP:conf/corl/DawsonQGF21,DBLP:journals/csysl/AbateAGP21,DBLP:conf/ijcai/BacciG021}, such as barrier functions or Lyapunov functions, satisfaction of which provides formal guarantees about the NCPS.
While these approaches provide infinite-time horizon guarantees,
the central component of their safety argument, e.g. encoded in the neural network-based barrier functions, is \emph{entirely sub-symbolic}, and itself unverified. We believe that \vehicle and ITPs may be of use here. 

%

\noindent \textbf{Neural networks and ODEs.} Many works have looked at applying neural networks to cyber-physical system verification.
Theoretically, reachability analysis can be combined with ODE solvers by expressing ODEs as neural networks~\cite{ZeqiriM0V23} or neural networks as ODEs~\cite{ivanov2021verisig}.
In practice, the resulting approaches suffer from the limitations of automated reachability-based techniques discussed in \autoref{section:introduction} (in particular state space complexity and time horizons).


\noindent \textbf{New Applications.} Finally, we plan to apply \vehicle to further real-life application in the NCPS domain.  \correction{For example, this paper does not re-implement prior ARCH-COMP benchmark translations as these require non-linear neural network specifications.
While such specifications are already supported by \vehicle,
existing non-linear solvers~\cite{TeuberVerSAILLE2024} are not yet compatible with the official VNN-LIB 2.0 standard~\cite{RAGD26} which is a prerequisite for integration with \vehicle.}

\makeatletter
\if@ACM@anonymous
\else
\subsection*{Contribution Statement}
\textbf{Daggitt}: Conceptualisation, Type system, Agda backend.
\textbf{Komendantskaya}: Conceptualisation, Medical case study.
\textbf{Sirman}: Medical case study, Rocq backend.
\textbf{Teuber}: Isabelle backend, CPS discussion.
\textbf{Bruni}: Rocq backend, Rocq tensor library, Medical case study.
\textbf{Smart}: Rocq backend, Rocq tensor library.
\textbf{Passmore}: Imandra backend.
\fi
\makeatother

\subsection*{Data-Availability Statement}
The \vehicle source code, the medical case study source code, and a virtual
machine containing Vehicle and the case study, along side all
necessary packages, can be found in \citet{daggitt_2026_20529849}.

\begin{acks}
E.Komendantskaya acknowledges the support of the grant \grantsponsor{EP/Y030834/1}{UKRI AI Centre for Doctoral Training in Dependable and Deployable Artificial Intelligence for Robotics (CDT-D2AIR)}{https://google.com}.
The PhD scholarship of A.Sirman was funded by EPSRC DLA Collaborative Studentship.
\end{acks}

\begin{subappendices}
\renewcommand{\thesection}{\Alph{section}}%

\section{Generated Car Controller ITP Interfaces}\label{appendix:car-controller}

In this Appendix we present the full interface generated by each ITP from the
Vehicle specification in Listing 1.

\subsection{Agda}\label{appendix:car-controller:agda}

\begin{minted}{agda}
-- WARNING: This file was generated automatically by Vehicle
-- and should not be modified manually!
-- Metadata:
--  - Vehicle version: 0.23.0
--  - Agda version: 2.6.2

{-# OPTIONS --allow-exec #-}

open import Vehicle
open import Vehicle.Utils
open import Vehicle.Data.Tensor
open import Data.Product
open import Data.Integer as ℤ using (ℤ)
open import Data.Rational as ℚ using (ℚ)
open import Data.Fin as Fin using (Fin; #_)
open import Data.Vec.Functional.Relation.Unary.All as Fin
open import Data.List.Base

module WindControllerSpec where

InputVector : Set
InputVector = Tensor ℚ (2 ∷ [])

currentSensor : Fin 2
currentSensor = # 0

previousSensor : Fin 2
previousSensor = # 1

OutputVector : Set
OutputVector = Tensor ℚ (1 ∷ [])

velocity : Fin 1
velocity = # 0

postulate controller : InputVector → OutputVector

normalise : InputVector → InputVector
normalise x = foreach (λ i → ((x ! i) 𝕋.⊕ 4) 𝕋.÷ 8)

SafeInput : InputVector → Set
SafeInput x = Fin.All (λ i → 𝕋.- (ℤ.+ 13 / 4) ≤ (x ! i) × (x ! i) ≤ ℤ.+ 13 / 4)

SafeOutput : InputVector → Set
SafeOutput x = let y = controller (normalise x) ! velocity in 𝕋.- (ℤ.+ 5 / 4) < (y 𝕋.⊕ 2 𝕋.* (x ! currentSensor)) 𝕋.⊖ (x ! previousSensor) × (y 𝕋.⊕ 2 𝕋.* (x ! currentSensor)) 𝕋.⊖ (x ! previousSensor) < ℤ.+ 5 / 4

abstract
  safe : ∀ x → SafeInput x → SafeOutput x
  safe = checkSpecification record
    { cache = "examples/windController/verificationResult"
    }
\end{minted}


\vspace{2em}
\subsection{Rocq}\label{appendix:car-controller:rocq}

\begin{minted}{coq}
(* WARNING: This file was generated automatically by Vehicle *)
(* and should not be modified manually! *)
(* Metadata: *)
(*  - Vehicle version: 0.23.0 *)
(*  - Rocq version: 9.0.0 *)

From mathcomp Require Import all_boot.
From mathcomp Require Import all_algebra.
From mathcomp Require Import all_reals.
Require Import vehicle.tensor.
Require Import vehicle.utils.
Open Scope ring_scope.
Open Scope order_scope.

Parameter R : realType.

Definition InputVector : Type := 'nT[R]_(2%N :: nil).

Definition currentSensor : 'I_2%N := 0.

Definition previousSensor : 'I_2%N := 1.

Definition OutputVector : Type := 'nT[R]_(1%N :: nil).

Definition velocity : 'I_1%N := 0.

Parameter controller : InputVector -> OutputVector.

Definition normalise (x : InputVector) : InputVector := nstack (fun i => x^^i + const_t (4 : R) / const_t (8 : R)).

Definition safeInput (x : InputVector) : Prop := forallIndex (fun i => (- const_t (13 / 4 : R) <= x^^i) /\ (x^^i <= const_t (13 / 4 : R))).

Definition safeOutput (x : InputVector) : Prop := let y := (controller (normalise x))^^velocity in (- const_t (5 / 4 : R) < ((y + (const_t (2 : R) * x^^currentSensor)) - x^^previousSensor)) /\ (((y + (const_t (2 : R) * x^^currentSensor)) - x^^previousSensor) < const_t (5 / 4 : R)).

Axiom safe : forall x, safeInput x -> safeOutput x.
\end{minted}



\vspace{2em}
\subsection{Isabelle/HOL}\label{appendix:car-controller:isabelle}

\begin{minted}{isabelle}
(* WARNING: This file was generated automatically by Vehicle *)
(* and should not be modified manually! *)
(* Metadata: *)
(*  - Vehicle version: 0.23.0 *)
(*  - Isabelle version: 2024 *)

theory  WindControllerSpec

  imports

    "Complex_Main"

    "Deep_Learning.Tensor"
    "Deep_Learning.Tensor_Subtensor"
    "Deep_Learning.Tensor_Scalar_Mult"
    "Vehicle.Vehicle"

begin

  type_synonym R = "real"

  

  typedef InputVector  =  "{ a :: R tensor. (dims a) = ((2 :: nat) # []) }"
    using dims_tensor_from_lookup by blast

  (* Type Coercions *)
  declare [[coercion Rep_InputVector]]
  definition to_InputVector :: "R FlexTensor \<Rightarrow> InputVector"
    where[simp]: "to_InputVector a = (
      let t = Rep_FlexTensor a
      in (if dims t = ((2 :: nat) # []) then Abs_InputVector t else undefined))"
  declare [[coercion to_InputVector]]
  (* Type Rewrite Rules *)
  lemma InputVector_tensor_rewrite0[simp]:
    assumes "prod_list shape = length elems"
        and "shape = (2 :: nat) # []"
    shows "(Rep_tensor (Rep_InputVector (Abs_InputVector (Abs_tensor (shape,elems))))) =  (shape,elems)"
  proof -
    have "Rep_InputVector (Abs_InputVector (Abs_tensor (shape,elems)))
      = Abs_tensor (shape,elems)"
        using Abs_InputVector_inverse[of "Abs_tensor (shape,elems)"]
        using Abs_tensor_inverse[of "(shape, elems)"]
        unfolding dims_def
        using assms
        by (simp)
    moreover have "Rep_tensor (Abs_tensor (shape,elems)) = (shape,elems)"
        using assms
        by (simp add: Abs_tensor_inverse)
    ultimately show ?thesis by simp
  qed

  definition currentSensor  :: "  FlexIndex "
    where " currentSensor   = ( (Abs_FlexIndex 0) ) "

  definition previousSensor  :: "  FlexIndex "
    where " previousSensor   = ( (Abs_FlexIndex 1) ) "

  typedef OutputVector  =  "{ a :: R tensor. (dims a) = ((1 :: nat) # []) }"
    using dims_tensor_from_lookup by blast

  (* Type Coercions *)
  declare [[coercion Rep_OutputVector]]
  definition to_OutputVector :: "R FlexTensor \<Rightarrow> OutputVector"
    where[simp]: "to_OutputVector a = (
      let t = Rep_FlexTensor a
      in (if dims t = ((1 :: nat) # []) then Abs_OutputVector t else undefined))"
  declare [[coercion to_OutputVector]]
  (* Type Rewrite Rules *)
  lemma OutputVector_tensor_rewrite0[simp]:
    assumes "prod_list shape = length elems"
        and "shape = (1 :: nat) # []"
    shows "(Rep_tensor (Rep_OutputVector (Abs_OutputVector (Abs_tensor (shape,elems))))) =  (shape,elems)"
  proof -
    have "Rep_OutputVector (Abs_OutputVector (Abs_tensor (shape,elems)))
      = Abs_tensor (shape,elems)"
        using Abs_OutputVector_inverse[of "Abs_tensor (shape,elems)"]
        using Abs_tensor_inverse[of "(shape, elems)"]
        unfolding dims_def
        using assms
        by (simp)
    moreover have "Rep_tensor (Abs_tensor (shape,elems)) = (shape,elems)"
        using assms
        by (simp add: Abs_tensor_inverse)
    ultimately show ?thesis by simp
  qed

  definition velocity  :: "  FlexIndex "
    where " velocity   = ( (Abs_FlexIndex 0) ) "

  definition normalise  :: " (InputVector \<Rightarrow> OutputVector) \<Rightarrow> (InputVector) \<Rightarrow>  InputVector "
    where " normalise   controller x = ( (foreach (2 :: nat)  (\<lambda>  i . ((pointwise_div (tensor_plus (flex_subtensor x i) (flextensor_from_vec [] [ (((4 :: R) )) ])) (flextensor_from_vec [] [ (((8 :: R) )) ]))))) ) "

  definition safeInput  :: " (InputVector \<Rightarrow> OutputVector) \<Rightarrow> (InputVector) \<Rightarrow>  bool "
    where " safeInput   controller x = ( (forallIndex (2 :: nat)  (\<lambda>  i . (((leqTensorReduced (tensor_cdot (-1 :: R) (flextensor_from_vec [] [ (((13 :: R) / 4)) ])) (flex_subtensor x i))) \<and> ((leqTensorReduced (flex_subtensor x i) (flextensor_from_vec [] [ (((13 :: R) / 4)) ])))))) ) "

  definition safeOutput  :: " (InputVector \<Rightarrow> OutputVector) \<Rightarrow> (InputVector) \<Rightarrow>  bool "
    where " safeOutput   controller x = ( let y = (flex_subtensor ((controller ((normalise controller x)))) velocity) in ((ltTensorReduced (tensor_cdot (-1 :: R) (flextensor_from_vec [] [ (((5 :: R) / 4)) ])) ((tensor_plus ((tensor_plus y ((hadamard_prod (flextensor_from_vec [] [ (((2 :: R) )) ]) (flex_subtensor x currentSensor))))) (tensor_cdot (-1 :: R) (flex_subtensor x previousSensor)))))) \<and> ((ltTensorReduced ((tensor_plus ((tensor_plus y ((hadamard_prod (flextensor_from_vec [] [ (((2 :: R) )) ]) (flex_subtensor x currentSensor))))) (tensor_cdot (-1 :: R) (flex_subtensor x previousSensor)))) (flextensor_from_vec [] [ (((5 :: R) / 4)) ]))) ) "

  locale  WindControllerSpec  = 

    fixes  controller  :: " InputVector \<Rightarrow> OutputVector "

    assumes  safe :  " (\<forall> x. (safeInput controller x) \<longrightarrow> (safeOutput controller x)) "

    begin

    end

end
\end{minted}

\vspace{2em}
\subsection{Imandra}\label{appendix:car-controller:imandra}

\begin{minted}{ocaml}
(* WARNING: This file was generated automatically by Vehicle *)
(* and should not be modified manually! *)
(* Metadata: *)
(*  - Vehicle version: 0.23.0 *)

[@@@import "tensor.iml"]
[@@@import "subtensor.iml"]
[@@@import "add.iml"]
[@@@import "scalar_mult.iml"]
[@@@import "vehicle.iml"]

open Vehicle

module WindControllerSpec = struct

  type flex_index = int

  type input_vector = real Tensor.tensor

  type output_vector = real Tensor.tensor

  let controller : input_vector -> output_vector = () [@@opaque]

  let current_sensor : int = 0

  let previous_sensor : int = 1

  let velocity : int = 0

  let normalise (x : input_vector) : input_vector =
    (foreach 2
      (fun (i : int) ->
        (pointwise_div_real
          (tensor_plus_real (flex_subtensor x i) (flextensor_from_vec [] [ (Real.(4.0)) ]))
          (flextensor_from_vec [] [ (Real.(8.0)) ]))))

  let safe_input (x : input_vector) : bool =
    (forall_index 2
      (fun (i : int) ->
        ((leq_tensor_reduced_real
            (tensor_cdot (-1.0) (flextensor_from_vec [] [ (Real.(13.0 /. 4.0)) ]))
            (flex_subtensor x i))
          && (leq_tensor_reduced_real (flex_subtensor x i) (flextensor_from_vec [] [ (Real.(13.0 /. 4.0)) ])))))

  let safe_output (x : input_vector) : bool =
    let y = (flex_subtensor (controller (normalise x)) velocity) in
    ((lt_tensor_reduced_real
        (tensor_cdot (-1.0) (flextensor_from_vec [] [ (Real.(5.0 /. 4.0)) ]))
        ((tensor_plus_real (tensor_plus_real
          y
          (hadamard_prod_real
            (flextensor_from_vec [] [ (Real.(2.0)) ])
            (flex_subtensor x current_sensor))) (tensor_cdot (-1.0) (flex_subtensor x previous_sensor)))))
      && (lt_tensor_reduced_real
        ((tensor_plus_real (tensor_plus_real
          y
          (hadamard_prod_real
            (flextensor_from_vec [] [ (Real.(2.0)) ])
            (flex_subtensor x current_sensor))) (tensor_cdot (-1.0) (flex_subtensor x previous_sensor))))
        (flextensor_from_vec [] [ (Real.(5.0 /. 4.0)) ])))

  axiom safe x =
    ((safe_input x) ==> (safe_output x))

end
\end{minted}

\section{Medical Case Study Rocq}\label{appendix:medical:rocq}
\begin{minted}{coq}
(* WARNING: This file was generated automatically by Vehicle *)
(* and should not be modified manually! *)
(* Metadata: *)
(*  - Vehicle version: 0.23.0+dev *)
(*  - Rocq version: 9.0.0 *)

From mathcomp Require Import all_boot.
From mathcomp Require Import all_algebra.
From mathcomp Require Import all_reals.
From mathcomp Require Import Rstruct.
Require Import vehicle.tensor.
Require Import Stdlib.Reals.Reals.
Open Scope ring_scope.
Open Scope order_scope.

Notation R := Rdefinitions.R.

Definition UnnormalisedInputVector : Type := 'nT[R]_(5%N :: nil).

Definition InputVector : Type := 'nT[R]_(5%N :: nil).

Definition conc : 'I_5%N := 0.

Definition temp : 'I_5%N := 1.

Definition wbc : 'I_5%N := 2.

Definition age : 'I_5%N := 3.

Definition weight : 'I_5%N := 4.

Definition OutputVector : Type := 'nT[R]_(1%N :: nil).

Parameter pk : InputVector -> OutputVector.

Definition normpk (x : UnnormalisedInputVector) : OutputVector := pk x.

Definition Ka : 'nT[R]_(nil) := const_t (9 / 2 : R).

Axiom Ka_pos : const_t (0 : R) < Ka.

Definition Ke : 'nT[R]_(nil) := const_t (7 / 2 : R).

Axiom Ke_pos : const_t (0 : R) < Ke.

Axiom Ke_n_Ka : Ka != Ke.

Definition Vd : 'nT[R]_(nil) := const_t (10 : R).

Axiom Vd_pos : const_t (0 : R) < Vd.

Definition C_safe : 'nT[R]_(nil) := const_t (30 : R).

Axiom C_safe_pos : const_t (0 : R) < C_safe.

Definition ttd : 'nT[R]_(nil) := const_t (2 : R).

Axiom ttd_pos : const_t (0 : R) < ttd.

Definition Ka_over : 'nT[R]_(nil) := const_t (807 / 2500 : R).

Definition Ka_under : 'nT[R]_(nil) := const_t (3227 / 10000 : R).

Definition Ke_over : 'nT[R]_(nil) := const_t (83 / 200 : R).

Definition Ke_under : 'nT[R]_(nil) := const_t (4149 / 10000 : R).

Definition eps : 'nT[R]_(nil) := const_t (1 / 1000 : R).

Definition safeFarInput (x : InputVector) : Prop := (const_t (0 : R) <= x ^^ conc /\ x ^^ conc <= C_safe * const_t (99 / 100 : R)) /\ ((const_t (73 / 2 : R) <= x ^^ temp /\ x ^^ temp <= const_t (40 : R)) /\ ((const_t (15 / 2 : R) <= x ^^ wbc /\ x ^^ wbc <= const_t (20 : R)) /\ ((const_t (18 : R) <= x ^^ age /\ x ^^ age <= const_t (89 : R)) /\ (const_t (50 : R) <= x ^^ weight /\ x ^^ weight <= const_t (100 : R))))).

Definition safeFarOutput (x : InputVector) : Prop := let y := ((normpk x)^^ 0 * Ka) / (Vd * (Ka - Ke)) in if (Ka < Ke) then (x ^^ conc + y * (Ke_under - Ka_over) < C_safe) else (x ^^ conc + y * (Ke_over - Ka_under) < C_safe).

Axiom safeFar : forall x, safeFarInput x -> safeFarOutput x.

Definition safeNearInput (x : InputVector) : Prop := (C_safe * const_t (99 / 100 : R) <= x ^^ conc /\ x ^^ conc <= C_safe) /\ ((const_t (73 / 2 : R) <= x ^^ temp /\ x ^^ temp <= const_t (40 : R)) /\ ((const_t (15 / 2 : R) <= x ^^ wbc /\ x ^^ wbc <= const_t (20 : R)) /\ ((const_t (18 : R) <= x ^^ age /\ x ^^ age <= const_t (89 : R)) /\ (const_t (50 : R) <= x ^^ weight /\ x ^^ weight <= const_t (100 : R))))).

Definition safeNearOutput (x : InputVector) : Prop := (normpk x) ^^ 0 < eps.

Axiom safeNear : forall x, safeNearInput x -> safeNearOutput x.

Definition safeInput (x : InputVector) : Prop := (const_t (0 : R) <= x ^^ conc /\ x ^^ conc <= C_safe) /\ ((const_t (73 / 2 : R) <= x ^^ temp /\ x ^^ temp <= const_t (40 : R)) /\ ((const_t (15 / 2 : R) <= x ^^ wbc /\ x ^^ wbc <= const_t (20 : R)) /\ ((const_t (18 : R) <= x ^^ age /\ x ^^ age <= const_t (89 : R)) /\ (const_t (50 : R) <= x ^^ weight /\ x ^^ weight <= const_t (100 : R))))).

Definition nonNegOutput (x : InputVector) : Prop := const_t (0 : R) < (normpk x) ^^ 0.

Axiom nonNeg : forall x, safeInput x -> nonNegOutput x.
\end{minted}

\end{subappendices}


\bibliographystyle{ACM-Reference-Format}
\bibliography{bibliography_cleaned}

@phdthesis{C09,
author = {Russell O'Connor}, 
title = {Incompleteness and Completeness: Formalizing Logic and Analysis in Type Theory},
school = {Radboud University Nijmegen},
year  = {2009} 
}

@book{alma991030766629705251,
year = {1985 - 1985},
title = {Constructive analysis / Errett Bishop, Douglas Bridges},
author = {Bishop, Errett and Bridges, D. S. and Bishop, Errett},
address = {Berlin ;},
booktitle = {Constructive analysis},
isbn = {0387150668},
language = {eng},
lccn = {85002828},
publisher = {Springer-Verlag},
series = {Grundlehren der mathematischen Wissenschaften ; 279},
}

@inproceedings{RAGD26,
author = "A. Roy and A. Antony and A. Gimelli and M.L. Daggitt",
title = "VNN-LIB 2.0: Rigorous Foundations for Neural Network Verification",
booktitle    = {Computer Aided Verification, CAV 2026},
year         = {2026}
}

@article{DBLP:journals/csysl/AbateAGP21,
  author       = {Alessandro Abate and
                  Daniele Ahmed and
                  Mirco Giacobbe and
                  Andrea Peruffo},
  title        = {Formal Synthesis of {Lyapunov} Neural Networks},
  journal      = {{IEEE} Control. Syst. Lett.},
  volume       = {5},
  number       = {3},
  pages        = {773--778},
  year         = {2021},
  opturl          = {https://doi.org/10.1109/LCSYS.2020.3005328},
  doi          = {10.1109/LCSYS.2020.3005328},
  bibsource    = {dblp computer science bibliography, https://dblp.org}
}

@software{affeldt_mathcomp-analysis_2026,
	title = {{MathComp}-Analysis: Mathematical Components compliant analysis library},
	url = {https://github.com/math-comp/analysis},
	publisher = {Mathematical Components},
	author = {Affeldt, Reynald and Bertot, Yves and Bruni, Alessandro and Cohen, Cyril and Kerjean, Marie and Mahboubi, Assia and Rouhling, Damien and Roux, Pierre and Sakaguchi, Kazuhiko and Stone, Zachary and Strub, Pierre-Yves and Théry, Laurent},
	urldate = {2026-02-02},
	date = {2026-02-02},
    year = {2026},
	note = {original-date: 2017-11-29T12:17:44Z},
}

@article{albarghouthi2021introductionneuralnetworkverification,
  author       = {Aws Albarghouthi},
  title        = {Introduction to Neural Network Verification},
  journal      = {Found. Trends Program. Lang.},
  volume       = {7},
  number       = {1-2},
  pages        = {1--157},
  year         = {2021},
  opturl          = {https://doi.org/10.1561/2500000051},
  doi          = {10.1561/2500000051},
  bibsource    = {dblp computer science bibliography, https://dblp.org}
}

@inproceedings{AlbertiBGGVC25,
  author       = {Michele Alberti and
                  Fran{\c{c}}ois Bobot and
                  Julien Girard{-}Satabin and
                  Alban Grastien and
                  Aymeric Varasse and
                  Zakaria Chihani},
  editor       = {Ferruccio Damiani and
                  Marie Farrell},
  title        = {The {CAISAR} Platform: Extending the Reach of Machine Learning Specification
                  and Verification},
  booktitle    = {Integrated Formal Methods - 20th International Conference, iFM 2025,
                  Paris, France, November 19-21, 2025, Proceedings},
  series       = {LNCS},
  volume       = {16194},
  pages        = {290--309},
  publisher    = {Springer},
  address = {Cham, Switzerland},
  year         = {2025},
  opturl          = {https://doi.org/10.1007/978-3-032-10794-7_15},
  doi          = {10.1007/978-3-032-10794-7_15},
  bibsource    = {dblp computer science bibliography, https://dblp.org}
}

@inproceedings{Althoff2015ARCH,
	author			= {Matthias Althoff},
	title			= {An Introduction to {CORA} 2015},
	year			= {2015},
	month 			= {December},
	booktitle		= {Proc. of the 1st and 2nd Workshop on Applied Verification for Continuous and Hybrid Systems},
	doi			= {10.29007/zbkv},
	url			= {https://easychair.org/publications/paper/xMm},
	publisher 		= {EasyChair},
	pages			= {120-151}
}

@inproceedings{Ansel_PyTorch_2_Faster_2024,
  author       = {Jason Ansel and
                  Edward Z. Yang and
                  Horace He and
                  Natalia Gimelshein and
                  Animesh Jain and
                  Michael Voznesensky and
                  Bin Bao and
                  Peter Bell and others},
  editor       = {Rajiv Gupta and
                  Nael B. Abu{-}Ghazaleh and
                  Madan Musuvathi and
                  Dan Tsafrir},
  title        = {PyTorch 2: Faster Machine Learning Through Dynamic Python Bytecode
                  Transformation and Graph Compilation},
  booktitle    = {Proceedings of the 29th {ACM} International Conference on Architectural
                  Support for Programming Languages and Operating Systems, Volume 2,
                  {ASPLOS} 2024, La Jolla, CA, USA, 27 April 2024- 1 May 2024},
  pages        = {929--947},
  publisher    = {{ACM}},
  address   = {New York, NY, USA},
  year         = {2024},
  url          = {https://doi.org/10.1145/3620665.3640366},
  doi          = {10.1145/3620665.3640366},
  bibsource    = {dblp computer science bibliography, https://dblp.org}
}

@misc{mathcomp_tensor_pr,
  author       = {Alessandro Bruni},
  title        = {Tensor formalization \#1535},
  howpublished = {GitHub Pull Request},
  url          = {https://github.com/math-comp/math-comp/pull/1535},
  note         = {Mathematical Components library},
  year         = {2026}
}

@inproceedings{DBLP:conf/ijcai/BacciG021,
  author       = {Edoardo Bacci and
                  Mirco Giacobbe and
                  David Parker},
  editor       = {Zhi{-}Hua Zhou},
  title        = {Verifying Reinforcement Learning up to Infinity},
  booktitle    = {Proceedings of the Thirtieth International Joint Conference on Artificial
                  Intelligence, {IJCAI} 2021, Virtual Event / Montreal, Canada, 19-27
                  August 2021},
  pages        = {2154--2160},
  publisher    = {ijcai.org},
  year         = {2021},
  opturl          = {https://doi.org/10.24963/ijcai.2021/297},
  doi          = {10.24963/IJCAI.2021/297},
  bibsource    = {dblp computer science bibliography, https://dblp.org}
}

@inproceedings{bagnall2019certifying,
  author       = {Alexander Bagnall and
                  Gordon Stewart},
  title        = {Certifying the True Error: Machine Learning in Coq with Verified Generalization
                  Guarantees},
  booktitle    = {The Thirty-Third {AAAI} Conference on Artificial Intelligence, {AAAI}
                  2019, Honolulu, Hawaii,
                  USA, January 27 - February 1, 2019},
  pages        = {2662--2669},
  publisher    = {{AAAI} Press},
  year         = {2019},
  opturl          = {https://doi.org/10.1609/aaai.v33i01.33012662},
  doi          = {10.1609/AAAI.V33I01.33012662},
  bibsource    = {dblp computer science bibliography, https://dblp.org}
}

@inproceedings{bak2020improved,
  author       = {Stanley Bak and
                  Hoang{-}Dung Tran and
                  Kerianne Hobbs and
                  Taylor T. Johnson},
  editor       = {Shuvendu K. Lahiri and
                  Chao Wang},
  title        = {Improved Geometric Path Enumeration for Verifying ReLU Neural Networks},
  booktitle    = {Computer Aided Verification - 32nd International Conference, {CAV}
                  2020, Los Angeles, CA, USA, July 21-24, 2020, Proceedings, Part {I}},
  series       = {LNCS},
  volume       = {12224},
  pages        = {66--96},
  publisher    = {Springer},
  address = {Cham, Switzerland},
  year         = {2020},
  opturl          = {https://doi.org/10.1007/978-3-030-53288-8_4},
  doi          = {10.1007/978-3-030-53288-8_4},
  bibsource    = {dblp computer science bibliography, https://dblp.org}
}

@inproceedings{BarbosaK0VTB23,
  author       = {Haniel Barbosa and
                  Chantal Keller and
                  Andrew Reynolds and
                  Arjun Viswanathan and
                  Cesare Tinelli and
                  Clark W. Barrett},
  editor       = {Ruzica Piskac and
                  Andrei Voronkov},
  title        = {An Interactive {SMT} Tactic in Coq using Abductive Reasoning},
  booktitle    = {{LPAR} 2023: Proceedings of 24th International Conference on Logic
                  for Programming, Artificial Intelligence and Reasoning, Manizales,
                  Colombia, 4-9th June 2023},
  series       = {EPiC Series in Computing},
  volume       = {94},
  pages        = {11--22},
  publisher    = {EasyChair},
  year         = {2023},
  opturl          = {https://doi.org/10.29007/432m},
  doi          = {10.29007/432M},
  bibsource    = {dblp computer science bibliography, https://dblp.org}
}

@article{Deep_Learning-AFP,
  author  = {Alexander Bentkamp},
  title   = {Expressiveness of Deep Learning},
  journal = {Archive of Formal Proofs},
  month   = {November},
  year    = {2016},
  note    = {\url{https://isa-afp.org/entries/Deep_Learning.html},
             Formal proof development},
  ISSN    = {2150-914x},
}

@inproceedings{10.1145/3302504.3311804,
  author       = {Sergiy Bogomolov and
                  Marcelo Forets and
                  Goran Frehse and
                  Kostiantyn Potomkin and
                  Christian Schilling},
  editor       = {Necmiye Ozay and
                  Pavithra Prabhakar},
  title        = {JuliaReach: a toolbox for set-based reachability},
  booktitle    = {Proceedings of the 22nd {ACM} International Conference on Hybrid Systems:
                  Computation and Control, {HSCC} 2019, Montreal, QC, Canada, April
                  16-18, 2019},
  pages        = {39--44},
  publisher    = {{ACM}},
  address   = {New York, NY, USA},
  year         = {2019},
  opturl          = {https://doi.org/10.1145/3302504.3311804},
  doi          = {10.1145/3302504.3311804},
  bibsource    = {dblp computer science bibliography, https://dblp.org}
}

@inproceedings{BohrerRVVP17,
  author       = {Rose Bohrer and
                  Vincent Rahli and
                  Ivana Vukotic and
                  Marcus V{\"{o}}lp and
                  Andr{\'{e}} Platzer},
  editor       = {Yves Bertot and
                  Viktor Vafeiadis},
  title        = {Formally verified differential dynamic logic},
  booktitle    = {Proceedings of the 6th {ACM} {SIGPLAN} Conference on Certified Programs
                  and Proofs, {CPP} 2017, Paris, France, January 16-17, 2017},
  pages        = {208--221},
  publisher    = {{ACM}},
  address   = {New York, NY, USA},
  year         = {2017},
  opturl          = {https://doi.org/10.1145/3018610.3018616},
  doi          = {10.1145/3018610.3018616},
  bibsource    = {dblp computer science bibliography, https://dblp.org}
}

@incollection{boyer1990use,
  title={The use of a formal simulator to verify a simple real time control program},
  author={Boyer, Robert S and Green, Milton W and Moore, J Strother},
  booktitle={Beauty is Our Business: A Birthday Salute to Edsger W. Dijkstra},
  pages={54--66},
  year={1990},
  publisher={Springer},
  address = {Berlin, Heidelberg}
}

@inproceedings{brucker2023verifying,
  author       = {Achim D. Brucker and
                  Amy Stell},
  editor       = {Marsha Chechik and
                  Joost{-}Pieter Katoen and
                  Martin Leucker},
  title        = {Verifying Feedforward Neural Networks for Classification in Isabelle/HOL},
  booktitle    = {Formal Methods - 25th International Symposium, {FM} 2023, L{\"{u}}beck,
                  Germany, March 6-10, 2023, Proceedings},
  series       = {LNCS},
  volume       = {14000},
  pages        = {427--444},
  publisher    = {Springer},
  address = {Cham, Switzerland},
  year         = {2023},
  url          = {https://doi.org/10.1007/978-3-031-27481-7_24},
  doi          = {10.1007/978-3-031-27481-7_24},
  bibsource    = {dblp computer science bibliography, https://dblp.org}
}

@inproceedings{DBLP:conf/nips/ChangRG19,
  author       = {Ya{-}Chien Chang and
                  Nima Roohi and
                  Sicun Gao},
  editor       = {Hanna M. Wallach and
                  Hugo Larochelle and
                  Alina Beygelzimer and
                  Florence d'Alch{\'{e}}{-}Buc and
                  Emily B. Fox and
                  Roman Garnett},
  title        = {Neural {Lyapunov} Control},
  booktitle    = {Advances in Neural Information Processing Systems 32: Annual Conference
                  on Neural Information Processing Systems 2019, NeurIPS 2019, December
                  8-14, 2019, Vancouver, BC, Canada},
  publisher    = {Curran Associates, Inc.},
  address   = {Red Hook, NY, USA},
  pages        = {3240--3249},
  year         = {2019},
  url          = {https://proceedings.neurips.cc/paper/2019/hash/2647c1dba23bc0e0f9cdf75339e120d2-Abstract.html},
  bibsource    = {dblp computer science bibliography, https://dblp.org}
}

@software{mathcomp,
  author = {Cyril Cohen and
                  Pierre Roux and
                  Enrico Tassi and
                  Kazuhiko Sakaguchi and
                  Reynald Affeldt and
                  Laurent Théry and
                  Erik Martin-Dorel and
                  Georges Gonthier and others},
  title = {math-comp/math-comp: The Mathematical Components Library 2.3.0},
  month = nov,
  year = 2024,
  publisher = {Zenodo},
  version = {mathcomp-2.3.0},
  doi = {10.5281/zenodo.14237175},
  url = {https://doi.org/10.5281/zenodo.14237175},
  swhid = {swh:1:dir:7edc79552bb39a26d1a8472526f171dc4791e986;origin=https://doi.org/10.5281/zenodo.10513585;visit=swh:1:snp:c14fec6c42fd73089a2bbcc977b09f2787c45c10;anchor=swh:1:rel:b8feae4cfbc27d8e972b11d03c48e153bd5fdcc5;path=/},
}

@inproceedings{CordeiroDGIJKKLMSW25,
  author       = {Lucas C. Cordeiro and
                  Matthew L. Daggitt and
                  Julien Girard{-}Satabin and
                  Omri Isac and
                  Taylor T. Johnson and
                  Guy Katz and
                  Ekaterina Komendantskaya and
                  Augustin Lemesle and
                  Edoardo Manino and
                  Artjoms Sinkarovs and
                  Haoze Wu},
  editor       = {Viktor Vafeiadis},
  title        = {Neural Network Verification is a Programming Language Challenge},
  booktitle    = {Programming Languages and Systems - 34th European Symposium on Programming,
                  {ESOP} 2025, Hamilton, ON, Canada,
                  May 3-8, 2025, Proceedings, Part {I}},
  series       = {LNCS},
  volume       = {15694},
  pages        = {206--235},
  publisher    = {Springer},
  address = {Cham, Switzerland},
  year         = {2025},
  opturl          = {https://doi.org/10.1007/978-3-031-91118-7_9},
  doi          = {10.1007/978-3-031-91118-7_9},
  bibsource    = {dblp computer science bibliography, https://dblp.org}
}

@article{daggittAgdaStandardLibrary2025,
  title = {The {Agda} Standard Library: Version 2.0},
  shorttitle = {The {Agda} Standard Library},
  author = {Daggitt, Matthew L.
            and Allais, Guillaume
            and McKinna, James
            and Abel, Andreas
            and Van Doorn, Nathan
            and Wood, James
            and Norell, Ulf
            and Kidney, Donnacha Ois{\'i}n and others},
  year = 2025,
  month = dec,
  journal = {Journal of Open Source Software},
  volume = {10},
  number = {116},
  pages = {9241},
  issn = {2475-9066},
  doi = {10.21105/joss.09241},
  urldate = {2026-02-16},
  copyright = {http://creativecommons.org/licenses/by/4.0/}
}

@inproceedings{DaggittAKKA23,
  author       = {Matthew L. Daggitt and
                  Robert Atkey and
                  Wen Kokke and
                  Ekaterina Komendantskaya and
                  Luca Arnaboldi},
  editor       = {Robbert Krebbers and
                  Dmitriy Traytel and
                  Brigitte Pientka and
                  Steve Zdancewic},
  title        = {Compiling Higher-Order Specifications to {SMT} Solvers: How to Deal
                  with Rejection Constructively},
  booktitle    = {Proceedings of the 12th {ACM} {SIGPLAN} International Conference on
                  Certified Programs and Proofs, {CPP} 2023, Boston, MA, USA, January
                  16-17, 2023},
  pages        = {102--120},
  publisher    = {{ACM}},
  address   = {New York, NY, USA},
  year         = {2023},
  url          = {https://doi.org/10.1145/3573105.3575674},
  doi          = {10.1145/3573105.3575674},
  bibsource    = {dblp computer science bibliography, https://dblp.org}
}

@article{daggitt2024efficient,
  author       = {Matthew L. Daggitt and
                  Wen Kokke and
                  Robert Atkey},
  title        = {Efficient compilation of expressive problem space specifications to
                  neural network solvers},
  journal      = {CoRR},
  volume       = {abs/2402.01353},
  year         = {2024},
  opturl          = {https://doi.org/10.48550/arXiv.2402.01353},
  doi          = {10.48550/ARXIV.2402.01353},
  eprinttype    = {arXiv},
  eprint       = {2402.01353},
  bibsource    = {dblp computer science bibliography, https://dblp.org}
}

@inproceedings{DBLP:conf/corl/DawsonQGF21,
  author       = {Charles Dawson and
                  Zengyi Qin and
                  Sicun Gao and
                  Chuchu Fan},
  editor       = {Aleksandra Faust and
                  David Hsu and
                  Gerhard Neumann},
  title        = {Safe Nonlinear Control Using Robust Neural {Lyapunov}-Barrier Functions},
  booktitle    = {Conference on Robot Learning, 8-11 November 2021, London, {UK}},
  series       = {Proceedings of Machine Learning Research},
  volume       = {164},
  pages        = {1724--1735},
  publisher    = {{PMLR}},
  year         = {2021},
  url          = {https://proceedings.mlr.press/v164/dawson22a.html},
  bibsource    = {dblp computer science bibliography, https://dblp.org}
}

@InProceedings{10.1007/978-3-642-38574-2_12,
  author       = {Leonardo Mendon{\c{c}}a de Moura and
                  Grant Olney Passmore},
  editor       = {Maria Paola Bonacina},
  title        = {Computation in Real Closed Infinitesimal and Transcendental Extensions
                  of the Rationals},
  booktitle    = {Automated Deduction - {CADE-24} - 24th International Conference on
                  Automated Deduction, Lake Placid, NY, USA, June 9-14, 2013. Proceedings},
  series       = {LNCS},
  volume       = {7898},
  pages        = {178--192},
  publisher    = {Springer},
  address = {Berlin, Heidelberg},
  year         = {2013},
  opturl          = {https://doi.org/10.1007/978-3-642-38574-2_12},
  doi          = {10.1007/978-3-642-38574-2_12},
  bibsource    = {dblp computer science bibliography, https://dblp.org}
}

@inproceedings{desmartin2022checkinn,
  author       = {Remi Desmartin and
                  Grant O. Passmore and
                  Ekaterina Komendantskaya and
                  Matthew L. Daggitt},
  title        = {CheckINN: Wide Range Neural Network Verification in Imandra},
  booktitle    = {{PPDP} 2022: 24th International Symposium on Principles and Practice
                  of Declarative Programming, Tbilisi, Georgia, September 20 - 22, 2022},
  pages        = {3:1--3:14},
  publisher    = {{ACM}},
  address   = {New York, NY, USA},
  year         = {2022},
  opturl          = {https://doi.org/10.1145/3551357.3551372},
  doi          = {10.1145/3551357.3551372},
  bibsource    = {dblp computer science bibliography, https://dblp.org}
}

@inproceedings{desmartin2025certified,
  author       = {Remi Desmartin and
                  Omri Isac and
                  Grant O. Passmore and
                  Ekaterina Komendantskaya and
                  Kathrin Stark and
                  Guy Katz},
  editor       = {Yannick Forster and
                  Chantal Keller},
  title        = {A Certified Proof Checker for Deep Neural Network Verification in
                  Imandra},
  booktitle    = {16th International Conference on Interactive Theorem Proving, {ITP}
                  2025, Reykjavik, Iceland, September 28 - October 1, 2025},
  series       = {LIPIcs},
  address   = {Dagstuhl, Germany},
  volume       = {352},
  pages        = {1:1--1:21},
  publisher    = {Schloss Dagstuhl - Leibniz-Zentrum f{\"{u}}r Informatik},
  year         = {2025},
  opturl          = {https://doi.org/10.4230/LIPIcs.ITP.2025.1},
  doi          = {10.4230/LIPICS.ITP.2025.1},
  bibsource    = {dblp computer science bibliography, https://dblp.org}
}

@article{dullemond1991introduction,
  title={Introduction to tensor calculus},
  author={Dullemond, Kees and Peeters, Kasper},
  journal={Kees Dullemond and Kasper Peeters},
  pages={42--44},
  year={1991}
}

@inproceedings{fischer2019dl2,
  author       = {Marc Fischer and
                  Mislav Balunovic and
                  Dana Drachsler{-}Cohen and
                  Timon Gehr and
                  Ce Zhang and
                  Martin T. Vechev},
  editor       = {Kamalika Chaudhuri and
                  Ruslan Salakhutdinov},
  title        = {{DL2:} Training and Querying Neural Networks with Logic},
  booktitle    = {Proceedings of the 36th International Conference on Machine Learning,
                  {ICML} 2019, 9-15 June 2019, Long Beach, California, {USA}},
  series       = {Proceedings of Machine Learning Research},
  volume       = {97},
  pages        = {1931--1941},
  publisher    = {{PMLR}},
  year         = {2019},
  opturl          = {http://proceedings.mlr.press/v97/fischer19a.html},
  bibsource    = {dblp computer science bibliography, https://dblp.org}
}

@inproceedings{FosterMGS21,
  author       = {Simon Foster and
                  Jonathan Juli{\'{a}}n Huerta y Munive and
                  Mario Gleirscher and
                  Georg Struth},
  editor       = {Marieke Huisman and
                  Corina S. Pasareanu and
                  Naijun Zhan},
  title        = {Hybrid Systems Verification with Isabelle/HOL: Simpler Syntax, Better
                  Models, Faster Proofs},
  booktitle    = {Formal Methods - 24th International Symposium, {FM} 2021, Virtual
                  Event, November 20-26, 2021, Proceedings},
  series       = {LNCS},
  volume       = {13047},
  pages        = {367--386},
  publisher    = {Springer},
  address = {Cham, Switzerland},
  year         = {2021},
  url          = {https://doi.org/10.1007/978-3-030-90870-6_20},
  doi          = {10.1007/978-3-030-90870-6_20},
  bibsource    = {dblp computer science bibliography, https://dblp.org}
}

@proceedings{ARCH25,
  title     = {Proceedings of 12th Int. Workshop on Applied Verification for Continuous and Hybrid Systems},
  editor    = {Goran Frehse and Matthias Althoff},
  series    = {EPiC Series in Computing},
  volume    = {108},
  publisher = {EasyChair},
  bibsource = {EasyChair, https://easychair.org},
  issn      = {2398-7340},
  year      = {2025}
}

@INPROCEEDINGS{FultonMQVP15,
  author       = {Nathan Fulton and
                  Stefan Mitsch and
                  Jan{-}David Quesel and
                  Marcus V{\"{o}}lp and
                  Andr{\'{e}} Platzer},
  editor       = {Amy P. Felty and
                  Aart Middeldorp},
  title        = {KeYmaera {X:} An Axiomatic Tactical Theorem Prover for Hybrid Systems},
  booktitle    = {Automated Deduction - {CADE-25} - 25th International Conference on
                  Automated Deduction, Berlin, Germany, August 1-7, 2015, Proceedings},
  series       = {LNCS},
  volume       = {9195},
  pages        = {527--538},
  publisher    = {Springer},
  address = {Cham, Switzerland},
  year         = {2015},
  url          = {https://doi.org/10.1007/978-3-319-21401-6_36},
  doi          = {10.1007/978-3-319-21401-6_36},
  bibsource    = {dblp computer science bibliography, https://dblp.org}
}

@inproceedings{GeuversKSW10,
  author       = {Herman Geuvers and
                  Adam Koprowski and
                  Dan Synek and
                  Eelis van der Weegen},
  editor       = {Matt Kaufmann and
                  Lawrence C. Paulson},
  title        = {Automated Machine-Checked Hybrid System Safety Proofs},
  booktitle    = {Interactive Theorem Proving, First International Conference, {ITP}
                  2010, Edinburgh, UK, July 11-14, 2010. Proceedings},
  series       = {LNCS},
  volume       = {6172},
  pages        = {259--274},
  publisher    = {Springer},
  address = {Berlin, Heidelberg},
  year         = {2010},
  opturl          = {https://doi.org/10.1007/978-3-642-14052-5_19},
  doi          = {10.1007/978-3-642-14052-5_19},
  bibsource    = {dblp computer science bibliography, https://dblp.org}
}

@inproceedings{huang2022polar,
  author       = {Chao Huang and
                  Jiameng Fan and
                  Xin Chen and
                  Wenchao Li and
                  Qi Zhu},
  editor       = {Ahmed Bouajjani and
                  Luk{\'{a}}s Hol{\'{\i}}k and
                  Zhilin Wu},
  title        = {{POLAR:} {A} Polynomial Arithmetic Framework for Verifying Neural-Network
                  Controlled Systems},
  booktitle    = {Automated Technology for Verification and Analysis - 20th International
                  Symposium, {ATVA} 2022, Virtual Event, October 25-28, 2022, Proceedings},
  series       = {LNCS},
  volume       = {13505},
  pages        = {414--430},
  publisher    = {Springer},
  address = {Cham, Switzerland},
  year         = {2022},
  opturl          = {https://doi.org/10.1007/978-3-031-19992-9_27},
  doi          = {10.1007/978-3-031-19992-9_27},
  bibsource    = {dblp computer science bibliography, https://dblp.org}
}

@inproceedings{DBLP:conf/fmcad/IsacBZK22,
  author       = {Omri Isac and
                  Clark W. Barrett and
                  Min Zhang and
                  Guy Katz},
  editor       = {Alberto Griggio and
                  Neha Rungta},
  title        = {Neural Network Verification with Proof Production},
  booktitle    = {22nd Formal Methods in Computer-Aided Design, {FMCAD} 2022, Trento,
                  Italy, October 17-21, 2022},
  pages        = {38--48},
  publisher    = {{IEEE}},
  address   = {Piscataway, NJ, USA},
  year         = {2022},
  opturl          = {https://doi.org/10.34727/2022/isbn.978-3-85448-053-2_9},
  doi          = {10.34727/2022/ISBN.978-3-85448-053-2_9},
  bibsource    = {dblp computer science bibliography, https://dblp.org}
}

@inproceedings{ivanov2021verisig,
  author       = {Radoslav Ivanov and
                  Taylor J. Carpenter and
                  James Weimer and
                  Rajeev Alur and
                  George J. Pappas and
                  Insup Lee},
  editor       = {Alexandra Silva and
                  K. Rustan M. Leino},
  title        = {Verisig 2.0: Verification of Neural Network Controllers Using Taylor
                  Model Preconditioning},
  booktitle    = {Computer Aided Verification - 33rd International Conference, {CAV}
                  2021, Virtual Event, July 20-23, 2021, Proceedings, Part {I}},
  series       = {LNCS},
  volume       = {12759},
  pages        = {249--262},
  publisher    = {Springer},
  address = {Cham, Switzerland},
  year         = {2021},
  opturl          = {https://doi.org/10.1007/978-3-030-81685-8_11},
  doi          = {10.1007/978-3-030-81685-8_11},
  bibsource    = {dblp computer science bibliography, https://dblp.org}
}

@book{jones2003qualified,
  title={Qualified types: theory and practice},
  author={Jones, Mark P},
  number={9},
  year={2003},
  publisher={Cambridge University Press},
  address   = {Cambridge, UK}
}

@article{JulianACASXuDNN,
author = {Julian, Kyle D. and Kochenderfer, Mykel J. and Owen, Michael P.},
title = {Deep Neural Network Compression for Aircraft Collision Avoidance Systems},
journal = {Journal of Guidance, Control, and Dynamics},
volume = {42},
number = {3},
pages = {598-608},
year = {2019},
doi = {10.2514/1.G003724},
optURL = {https://doi.org/10.2514/1.G003724}
}

@inproceedings{katz_reluplex_2017,
  author       = {Guy Katz and
                  Clark W. Barrett and
                  David L. Dill and
                  Kyle Julian and
                  Mykel J. Kochenderfer},
  editor       = {Rupak Majumdar and
                  Viktor Kuncak},
  title        = {Reluplex: An Efficient {SMT} Solver for Verifying Deep Neural Networks},
  booktitle    = {Computer Aided Verification - 29th International Conference, {CAV}
                  2017, Heidelberg, Germany, July 24-28, 2017, Proceedings, Part {I}},
  series       = {LNCS},
  volume       = {10426},
  pages        = {97--117},
  publisher    = {Springer},
  address = {Cham, Switzerland},
  year         = {2017},
  opturl          = {https://doi.org/10.1007/978-3-319-63387-9_5},
  doi          = {10.1007/978-3-319-63387-9_5},
  bibsource    = {dblp computer science bibliography, https://dblp.org}
}

@inproceedings{KesslerKCVFOMM25,
  author       = {Colin Kessler and
                  Ekaterina Komendantskaya and
                  Marco Casadio and
                  Ignazio Maria Viola and
                  Thomas Flinkow and
                  Albaraa Ammar Othman and
                  Alistair Malhotra and
                  Robbie McPherson},
  editor       = {Mirco Giacobbe and
                  Anna Lukina},
  title        = {Neural Network Verification for Gliding Drone Control: {A} Case Study},
  booktitle    = {{AI} Verification - Second International Symposium, {SAIV} 2025, Zagreb,
                  Croatia, July 21-22, 2025, Proceedings},
  series       = {LNCS},
  volume       = {15947},
  pages        = {180--199},
  publisher    = {Springer},
  address = {Cham, Switzerland},
  year         = {2025},
  opturl          = {https://doi.org/10.1007/978-3-031-99991-8_9},
  doi          = {10.1007/978-3-031-99991-8_9},
  bibsource    = {dblp computer science bibliography, https://dblp.org}
}

@inproceedings{LachnittFBJA0SB25,
  author       = {Hanna Lachnitt and
                  Mathias Fleury and
                  Haniel Barbosa and
                  Jibiana Jakpor and
                  Bruno Andreotti and
                  Andrew Reynolds and
                  Hans{-}J{\"{o}}rg Schurr and
                  Clark W. Barrett and
                  Cesare Tinelli},
  editor       = {Yannick Forster and
                  Chantal Keller},
  title        = {Improving the {SMT} Proof Reconstruction Pipeline in Isabelle/HOL},
  booktitle    = {16th International Conference on Interactive Theorem Proving, {ITP}
                  2025, Reykjavik, Iceland, September 28 - October 1, 2025},
  series       = {LIPIcs},
  address   = {Dagstuhl, Germany},
  volume       = {352},
  pages        = {26:1--26:22},
  publisher    = {Schloss Dagstuhl - Leibniz-Zentrum f{\"{u}}r Informatik},
  year         = {2025},
  url          = {https://doi.org/10.4230/LIPIcs.ITP.2025.26},
  doi          = {10.4230/LIPICS.ITP.2025.26},
  bibsource    = {dblp computer science bibliography, https://dblp.org}
}

@inproceedings{leino2010dafny,
  author       = {K. Rustan M. Leino},
  editor       = {Edmund M. Clarke and
                  Andrei Voronkov},
  title        = {Dafny: An Automatic Program Verifier for Functional Correctness},
  booktitle    = {Logic for Programming, Artificial Intelligence, and Reasoning - 16th
                  International Conference, LPAR-16, Dakar, Senegal, April 25-May 1,
                  2010, Revised Selected Papers},
  series       = {LNCS},
  volume       = {6355},
  pages        = {348--370},
  publisher    = {Springer},
  address = {Berlin, Heidelberg},
  year         = {2010},
  opturl          = {https://doi.org/10.1007/978-3-642-17511-4_20},
  doi          = {10.1007/978-3-642-17511-4_20},
  bibsource    = {dblp computer science bibliography, https://dblp.org}
}

@article{loh2010tutorial,
  author       = {Andres L{\"{o}}h and
                  Conor McBride and
                  Wouter Swierstra},
  title        = {A Tutorial Implementation of a Dependently Typed Lambda Calculus},
  journal      = {Fundam. Informaticae},
  volume       = {102},
  number       = {2},
  pages        = {177--207},
  year         = {2010},
  opturl          = {https://doi.org/10.3233/FI-2010-304},
  doi          = {10.3233/FI-2010-304},
  bibsource    = {dblp computer science bibliography, https://dblp.org}
}

@inproceedings{nnv2_cav2023,
  author       = {Diego Manzanas Lopez and
                  Sung Woo Choi and
                  Hoang{-}Dung Tran and
                  Taylor T. Johnson},
  editor       = {Constantin Enea and
                  Akash Lal},
  title        = {{NNV} 2.0: The Neural Network Verification Tool},
  booktitle    = {Computer Aided Verification - 35th International Conference, {CAV}
                  2023, Paris, France, July 17-22, 2023, Proceedings, Part {II}},
  series       = {LNCS},
  volume       = {13965},
  pages        = {397--412},
  publisher    = {Springer},
  address = {Cham, Switzerland},
  year         = {2023},
  url          = {https://doi.org/10.1007/978-3-031-37703-7_19},
  doi          = {10.1007/978-3-031-37703-7_19},
  bibsource    = {dblp computer science bibliography, https://dblp.org}
}

@inproceedings{ARCH_COMP19_Category_Report,
  author    = {Diego Manzanas Lopez and Patrick Musau and Hoang-Dung Tran and Souradeep Dutta and Taylor J. Carpenter and Radoslav Ivanov and Taylor T. Johnson},
  title     = {ARCH-COMP19 Category Report: Artificial Intelligence and Neural Network Control Systems (AINNCS) for Continuous and Hybrid Systems Plants},
  booktitle = {ARCH19. 6th International Workshop on Applied Verification of Continuous and Hybrid Systems},
  editor    = {Goran Frehse and Matthias Althoff},
  series    = {EPiC Series in Computing},
  volume    = {61},
  publisher = {EasyChair},
  bibsource = {EasyChair, https://easychair.org},
  issn      = {2398-7340},
  url       = {/publications/paper/BFKs},
  doi       = {10.29007/rgv8},
  pages     = {103-119},
  year      = {2019}}

@inproceedings{DBLP:conf/eann/Lopez-MiguelAGV23,
  author       = {Ignacio D. Lopez{-}Miguel and
                  Borja Fern{\'{a}}ndez Adiego and
                  Faiq Ghawash and
                  Enrique Blanco Vi{\~{n}}uela},
  editor       = {Lazaros Iliadis and
                  Ilias Maglogiannis and
                  Seraf{\'{\i}}n Alonso and
                  Chrisina Jayne and
                  Elias Pimenidis},
  title        = {Verification of Neural Networks Meets {PLC} Code: An {LHC} Cooling
                  Tower Control System at {CERN}},
  booktitle    = {Engineering Applications of Neural Networks - 24th International Conference,
                  {EAAAI/EANN} 2023, Le{\'{o}}n, Spain, June 14-17, 2023, Proceedings},
  series       = {Communications in Computer and Information Science},
  volume       = {1826},
  pages        = {420--432},
  publisher    = {Springer},
  address = {Cham, Switzerland},
  year         = {2023},
  opturl          = {https://doi.org/10.1007/978-3-031-34204-2_35},
  doi          = {10.1007/978-3-031-34204-2_35},
  bibsource    = {dblp computer science bibliography, https://dblp.org}
}

@article{martin-dorel_proving_2016,
  author       = {{\'{E}}rik Martin{-}Dorel and
                  Guillaume Melquiond},
  title        = {Proving Tight Bounds on Univariate Expressions with Elementary Functions
                  in Coq},
  journal      = {J. Autom. Reason.},
  volume       = {57},
  number       = {3},
  pages        = {187--217},
  year         = {2016},
  url          = {https://doi.org/10.1007/s10817-015-9350-4},
  doi          = {10.1007/S10817-015-9350-4},
  bibsource    = {dblp computer science bibliography, https://dblp.org}
}

@inproceedings{morris2010instance,
  author       = {J. Garrett Morris and
                  Mark P. Jones},
  editor       = {Paul Hudak and
                  Stephanie Weirich},
  title        = {Instance chains: type class programming without overlapping instances},
  booktitle    = {Proceedings of the 15th {ACM} {SIGPLAN} International Conference on
                  Functional Programming, {ICFP} 2010, Baltimore, Maryland, USA, September
                  27-29, 2010},
  pages        = {375--386},
  publisher    = {{ACM}},
  address   = {New York, NY, USA},
  year         = {2010},
  opturl          = {https://doi.org/10.1145/1863543.1863596},
  doi          = {10.1145/1863543.1863596},
  bibsource    = {dblp computer science bibliography, https://dblp.org}
}

@book{DBLP:books/sp/NipkowPW02,
  author       = {Tobias Nipkow and
                  Lawrence C. Paulson and
                  Markus Wenzel},
  title        = {Isabelle/HOL - {A} Proof Assistant for Higher-Order Logic},
  series       = {LNCS},
  volume       = {2283},
  publisher    = {Springer},
  address = {Berlin, Heidelberg},
  year         = {2002},
  url          = {https://doi.org/10.1007/3-540-45949-9},
  doi          = {10.1007/3-540-45949-9},
  isbn         = {3-540-43376-7},
  bibsource    = {dblp computer science bibliography, https://dblp.org}
}

@inproceedings{norell2009dependently,
  author       = {Ulf Norell},
  editor       = {Andrew Kennedy and
                  Amal Ahmed},
  title        = {Dependently typed programming in Agda},
  booktitle    = {Proceedings of TLDI'09: 2009 {ACM} {SIGPLAN} International Workshop
                  on Types in Languages Design and Implementation, Savannah, GA, USA,
                  January 24, 2009},
  pages        = {1--2},
  publisher    = {{ACM}},
  address   = {New York, NY, USA},
  year         = {2009},
  url          = {https://doi.org/10.1145/1481861.1481862},
  doi          = {10.1145/1481861.1481862},
  bibsource    = {dblp computer science bibliography, https://dblp.org}
}

@inproceedings{OuchaniKH20,
  author       = {Samir Ouchani and
                  Khaled Khebbeb and
                  Meriem Hafsi},
  title        = {Towards Enhancing Security and Resilience in {CPS:} {A} Coq-Maude
                  based Approach},
  booktitle    = {17th {IEEE/ACS} International Conference on Computer Systems and Applications,
                  {AICCSA} 2020, Antalya, Turkey, November 2-5, 2020},
  pages        = {1--6},
  publisher    = {{IEEE}},
  address   = {Piscataway, NJ, USA},
  year         = {2020},
  opturl          = {https://doi.org/10.1109/AICCSA50499.2020.9316535},
  doi          = {10.1109/AICCSA50499.2020.9316535},
  bibsource    = {dblp computer science bibliography, https://dblp.org}
}

@inproceedings{passmore2020imandra,
  author       = {Grant O. Passmore and
                  Simon Cruanes and
                  Denis Ignatovich and
                  Dave Aitken and
                  Matt Bray and
                  Elijah Kagan and
                  Kostya Kanishev and
                  Ewen Maclean and
                  Nicola Mometto},
  editor       = {Nicolas Peltier and
                  Viorica Sofronie{-}Stokkermans},
  title        = {The Imandra Automated Reasoning System (System Description)},
  booktitle    = {Automated Reasoning - 10th International Joint Conference, {IJCAR}
                  2020, Paris, France, July 1-4, 2020, Proceedings, Part {II}},
  series       = {LNCS},
  volume       = {12167},
  pages        = {464--471},
  publisher    = {Springer},
  address = {Cham, Switzerland},
  year         = {2020},
  opturl          = {https://doi.org/10.1007/978-3-030-51054-1_30},
  doi          = {10.1007/978-3-030-51054-1_30},
  bibsource    = {dblp computer science bibliography, https://dblp.org}
}

@article{poweleit_artificial_2023,
  title={Artificial intelligence and machine learning approaches to facilitate therapeutic drug management and model-informed precision dosing},
  author={Poweleit, Ethan A and Vinks, Alexander A and Mizuno, Tomoyuki},
  journal={Therapeutic drug monitoring},
  volume={45},
  number={2},
  pages={143--150},
  year={2023},
  publisher={LWW}
}

@inproceedings{QianCBA25,
  author       = {Yicheng Qian and
                  Joshua Clune and
                  Clark W. Barrett and
                  Jeremy Avigad},
  editor       = {Ruzica Piskac and
                  Zvonimir Rakamaric},
  title        = {{Lean-Auto}: An Interface Between Lean 4 and Automated Theorem Provers},
  booktitle    = {Computer Aided Verification - 37th International Conference, {CAV}
                  2025, Zagreb, Croatia, July 23-25, 2025, Proceedings, Part {III}},
  series       = {LNCS},
  volume       = {15933},
  pages        = {175--196},
  publisher    = {Springer},
  address = {Cham, Switzerland},
  year         = {2025},
  url          = {https://doi.org/10.1007/978-3-031-98682-6_10},
  doi          = {10.1007/978-3-031-98682-6_10},
  bibsource    = {dblp computer science bibliography, https://dblp.org}
}

@phdthesis{Ricketts17,
  author       = {Daniel Ricketts},
  title        = {Verification of Sampled-Data Systems using Coq},
  school       = {University of California, San Diego, {USA}},
  year         = {2017},
  url          = {http://www.escholarship.org/uc/item/5n1899s2},
  bibsource    = {dblp computer science bibliography, https://dblp.org}
}

@inproceedings{DBLP:conf/tacas/ShiJKJHZ25,
  author       = {Zhouxing Shi and
                  Qirui Jin and
                  Zico Kolter and
                  Suman Jana and
                  Cho{-}Jui Hsieh and
                  Huan Zhang},
  editor       = {Arie Gurfinkel and
                  Marijn Heule},
  title        = {Neural Network Verification with Branch-and-Bound for General Nonlinearities},
  booktitle    = {Tools and Algorithms for the Construction and Analysis of Systems
                  - 31st International Conference, {TACAS} 2025, Hamilton, ON, Canada, May 3-8, 2025, Proceedings, Part
                  {I}},
  series       = {LNCS},
  volume       = {15696},
  pages        = {315--335},
  publisher    = {Springer},
  address = {Cham, Switzerland},
  year         = {2025},
  opturl          = {https://doi.org/10.1007/978-3-031-90643-5_17},
  doi          = {10.1007/978-3-031-90643-5_17},
  bibsource    = {dblp computer science bibliography, https://dblp.org}
}

@book{simmonsDifferentialEquationsApplications2016,
	edition = {3},
	title = {Differential Equations with Applications and Historical Notes},
	isbn = {978-1-4987-0262-1},
	pagetotal = {510},
	publisher = {{CRC} Press},
	author = {Simmons, George F.},
	year = {1972}
}

@article{singh_abstract_2019,
  author       = {Gagandeep Singh and
                  Timon Gehr and
                  Markus P{\"{u}}schel and
                  Martin T. Vechev},
  title        = {An abstract domain for certifying neural networks},
  journal      = {Proc. {ACM} Program. Lang.},
  volume       = {3},
  number       = {{POPL}},
  pages        = {41:1--41:30},
  year         = {2019},
  opturl          = {https://doi.org/10.1145/3290354},
  doi          = {10.1145/3290354},
  bibsource    = {dblp computer science bibliography, https://dblp.org}
}

@inproceedings{slusarz2022differentiable,
  author       = {Natalia Slusarz and
                  Ekaterina Komendantskaya and
                  Matthew L. Daggitt and
                  Robert J. Stewart},
  editor       = {Omri Isac and
                  Radoslav Ivanov and
                  Guy Katz and
                  Nina Narodytska and
                  Laura Nenzi},
  title        = {Differentiable Logics for Neural Network Training and Verification},
  booktitle    = {15th International
                  Workshop, Numerical Software Verification 2022, Haifa, Israel, July 31 - August 1, and August 11, 2022, Proceedings},
  series       = {LNCS},
  volume       = {13466},
  pages        = {67--77},
  publisher    = {Springer},
  address = {Cham, Switzerland},
  year         = {2022},
  url          = {https://doi.org/10.1007/978-3-031-21222-2_5},
  doi          = {10.1007/978-3-031-21222-2_5},
  bibsource    = {dblp computer science bibliography, https://dblp.org}
}

@inproceedings{swamy2016dependent,
  author       = {Nikhil Swamy and
                  Catalin Hritcu and
                  Chantal Keller and
                  Aseem Rastogi and
                  Antoine Delignat{-}Lavaud and
                  Simon Forest and
                  Karthikeyan Bhargavan and
                  C{\'{e}}dric Fournet and
                  Pierre{-}Yves Strub and
                  Markulf Kohlweiss and
                  Jean Karim Zinzindohoue and
                  Santiago Zanella{-}B{\'{e}}guelin},
  editor       = {Rastislav Bod{\'{\i}}k and
                  Rupak Majumdar},
  title        = {Dependent types and multi-monadic effects in {F}},
  booktitle    = {Proceedings of the 43rd Annual {ACM} {SIGPLAN-SIGACT} Symposium on
                  Principles of Programming Languages, {POPL} 2016, St. Petersburg,
                  FL, USA, January 20 - 22, 2016},
  pages        = {256--270},
  publisher    = {{ACM}},
  address   = {New York, NY, USA},
  year         = {2016},
  opturl          = {https://doi.org/10.1145/2837614.2837655},
  doi          = {10.1145/2837614.2837655},
  bibsource    = {dblp computer science bibliography, https://dblp.org}
}

@article{swierstra2008data,
  author       = {Wouter Swierstra},
  title        = {Data types {\`{a}} la carte},
  journal      = {J. Funct. Program.},
  volume       = {18},
  number       = {4},
  pages        = {423--436},
  year         = {2008},
  opturl          = {https://doi.org/10.1017/S0956796808006758},
  doi          = {10.1017/S0956796808006758},
  bibsource    = {dblp computer science bibliography, https://dblp.org}
}

@incollection{talevi_one-compartment_2021,
    address = {Cham, Switzerland},
	title = {One-Compartment Pharmacokinetic Model},
	isbn = {978-3-030-51519-5},
	url = {https://doi.org/10.1007/978-3-030-51519-5_58-1},
	doi = {10.1007/978-3-030-51519-5_58-1},
	pages = {1--8},
	booktitle = {The {ADME} Encyclopedia: A Comprehensive Guide on Biopharmacy and Pharmacokinetics},
	publisher = {Springer International Publishing},
	author = {Talevi, Alan and Bellera, Carolina L.},
	year = {2021},
}

@incollection{taleviTwoCompartmentPharmacokineticModel2021,
	title = {Two-Compartment Pharmacokinetic Model},
	isbn = {978-3-030-51519-5},
	url = {https://link.springer.com/rwe/10.1007/978-3-030-51519-5_59-1},
	doi = {10.1007/978-3-030-84860-6_59},
	pages = {1--8},
	booktitle = {The {ADME} Encyclopedia},
	publisher = {Springer},
    address = {Cham, Switzerland},
	author = {Talevi, Alan and Bellera, Carolina L.},
	urldate = {2026-02-19},
	year = {2021},
	langid = {english},
	doi = {10.1007/978-3-030-51519-5_59-1},
}

@inproceedings{TeuberVerSAILLE2024,
  author       = {Samuel Teuber and
                  Stefan Mitsch and
                  Andr{\'{e}} Platzer},
  editor       = {Amir Globersons and
                  Lester Mackey and
                  Danielle Belgrave and
                  Angela Fan and
                  Ulrich Paquet and
                  Jakub M. Tomczak and
                  Cheng Zhang},
  title        = {Provably Safe Neural Network Controllers via Differential Dynamic
                  Logic},
  booktitle    = {Advances in Neural Information Processing Systems 38: Annual Conference
                  on Neural Information Processing Systems 2024, NeurIPS 2024, Vancouver,
                  BC, Canada, December 10 - 15, 2024},
  year         = {2024},
  publisher    = {Curran Associates, Inc.},
  address   = {Red Hook, NY, USA},
  url          = {http://papers.neurips.cc/paper_files/paper/2024/hash/031b5fd7d847f69ed33378a9a1117b4b-Abstract-Conference.html},
  pages = {1586--1624},
  bibsource    = {dblp computer science bibliography, https://dblp.org}
}

@INPROCEEDINGS{TeuberAngelsAndDemons2025,
  author       = {Samuel Teuber and
                  Debasmita Lohar and
                  Bernhard Beckert},
  editor       = {Ahmed Irfan and
                  Daniela Kaufmann},
  title        = {Of Good Demons and Bad Angels: Guaranteeing Safe Control under Finite
                  Precision},
  booktitle    = {Proceedings of the 25th Conference on Formal Methods in Computer-Aided
                  Design, {FMCAD} 2025, Menlo Park, CA, USA, October 6-10, 2025},
  publisher    = {{TU} Wien Academic Press},
  address = {Vienna, Austria},
  year         = {2025},
  opturl          = {https://doi.org/10.34727/2025/isbn.978-3-85448-084-6_12},
  doi          = {10.34727/2025/ISBN.978-3-85448-084-6_12},
  bibsource    = {dblp computer science bibliography, https://dblp.org}
}

@manual{RocqManual,
  author       = {{The Rocq Development Team}},
  title        = {{The Rocq Reference Manual, version 9.1.0}},
  year         = {2025},
  howpublished = {\url{https://rocq-prover.org/doc/V9.1.0/refman/index.html}}
}

@phdthesis{vazou2016liquid,
  author       = {Niki Vazou},
  title        = {Liquid Haskell: Haskell as a Theorem Prover},
  school       = {University of California, San Diego, {USA}},
  year         = {2016},
  opturl          = {http://www.escholarship.org/uc/item/8dm057ws},
  bibsource    = {dblp computer science bibliography, https://dblp.org}
}

@inproceedings{DBLP:conf/nips/WangZXLJHK21,
  author       = {Shiqi Wang and
                  Huan Zhang and
                  Kaidi Xu and
                  Xue Lin and
                  Suman Jana and
                  Cho{-}Jui Hsieh and
                  J. Zico Kolter},
  editor       = {Marc'Aurelio Ranzato and
                  Alina Beygelzimer and
                  Yann N. Dauphin and
                  Percy Liang and
                  Jennifer Wortman Vaughan},
  title        = {Beta-CROWN: Efficient Bound Propagation with Per-neuron Split Constraints
                  for Neural Network Robustness Verification},
  booktitle    = {Advances in Neural Information Processing Systems 34: Annual Conference
                  on Neural Information Processing Systems 2021, NeurIPS 2021, December
                  6-14, 2021, virtual},
  pages        = {29909--29921},
  year         = {2021},
  url          = {https://proceedings.neurips.cc/paper/2021/hash/fac7fead96dafceaf80c1daffeae82a4-Abstract.html},
  publisher    = {Curran Associates, Inc.},
  address   = {Red Hook, NY, USA},
  bibsource    = {dblp computer science bibliography, https://dblp.org}
}

@article{wang2023polar,
  author       = {Yixuan Wang and
                  Weichao Zhou and
                  Jiameng Fan and
                  Zhilu Wang and
                  Jiajun Li and
                  Xin Chen and
                  Chao Huang and
                  Wenchao Li and
                  Qi Zhu},
  title        = {POLAR-Express: Efficient and Precise Formal Reachability Analysis
                  of Neural-Network Controlled Systems},
  journal      = {{IEEE} Trans. Comput. Aided Des. Integr. Circuits Syst.},
  volume       = {43},
  number       = {3},
  pages        = {994--1007},
  year         = {2024},
  url          = {https://doi.org/10.1109/TCAD.2023.3331215},
  doi          = {10.1109/TCAD.2023.3331215},
  bibsource    = {dblp computer science bibliography, https://dblp.org}
}

@inproceedings{WhiteTSM24,
  author       = {Lauren M. White and
                  Laura Titolo and
                  J. Tanner Slagel and
                  C{\'{e}}sar A. Mu{\~{n}}oz},
  editor       = {Amin Timany and
                  Dmitriy Traytel and
                  Brigitte Pientka and
                  Sandrine Blazy},
  title        = {A Temporal Differential Dynamic Logic Formal Embedding},
  booktitle    = {Proceedings of the 13th {ACM} {SIGPLAN} International Conference on
                  Certified Programs and Proofs, {CPP} 2024, London, UK, January 15-16,
                  2024},
  pages        = {162--176},
  publisher    = {{ACM}},
  address   = {New York, NY, USA},
  year         = {2024},
  opturl          = {https://doi.org/10.1145/3636501.3636943},
  doi          = {10.1145/3636501.3636943},
  bibsource    = {dblp computer science bibliography, https://dblp.org}
}

@inproceedings{wu2024marabou,
  author       = {Haoze Wu and
                  Omri Isac and
                  Aleksandar Zeljic and
                  Teruhiro Tagomori and
                  Matthew L. Daggitt and
                  Wen Kokke and
                  Idan Refaeli and
                  Guy Amir and others},
  editor       = {Arie Gurfinkel and
                  Vijay Ganesh},
  title        = {Marabou 2.0: {A} Versatile Formal Analyzer of Neural Networks},
  booktitle    = {Computer Aided Verification - 36th International Conference, {CAV}
                  2024, Montreal, QC, Canada, July 24-27, 2024, Proceedings, Part {II}},
  series       = {LNCS},
  volume       = {14682},
  pages        = {249--264},
  publisher    = {Springer},
  address = {Cham, Switzerland},
  year         = {2024},
  url          = {https://doi.org/10.1007/978-3-031-65630-9_13},
  doi          = {10.1007/978-3-031-65630-9_13},
  bibsource    = {dblp computer science bibliography, https://dblp.org}
}

@inproceedings{xufast,
  author       = {Kaidi Xu and
                  Huan Zhang and
                  Shiqi Wang and
                  Yihan Wang and
                  Suman Jana and
                  Xue Lin and
                  Cho{-}Jui Hsieh},
  title        = {Fast and Complete: Enabling Complete Neural Network Verification with
                  Rapid and Massively Parallel Incomplete Verifiers},
  booktitle    = {9th International Conference on Learning Representations, {ICLR} 2021,
                  Virtual Event, Austria, May 3-7, 2021},
  publisher    = {OpenReview.net},
  year         = {2021},
  url          = {https://openreview.net/forum?id=nVZtXBI6LNn},
  bibsource    = {dblp computer science bibliography, https://dblp.org}
}

@article{MuniveFGSLH24,
  author       = {Jonathan Juli{\'{a}}n Huerta y Munive and
                  Simon Foster and
                  Mario Gleirscher and
                  Georg Struth and
                  Christian Pardillo Laursen and
                  Thomas Hickman},
  title        = {IsaVODEs: Interactive Verification of Cyber-Physical Systems at Scale},
  journal      = {J. Autom. Reason.},
  volume       = {68},
  number       = {4},
  pages        = {21},
  year         = {2024},
  opturl          = {https://doi.org/10.1007/s10817-024-09709-2},
  doi          = {10.1007/S10817-024-09709-2},
  bibsource    = {dblp computer science bibliography, https://dblp.org}
}

@article{MuniveS22,
  author       = {Jonathan Juli{\'{a}}n Huerta y Munive and
                  Georg Struth},
  title        = {Predicate Transformer Semantics for Hybrid Systems},
  journal      = {J. Autom. Reason.},
  volume       = {66},
  number       = {1},
  pages        = {93--139},
  year         = {2022},
  opturl          = {https://doi.org/10.1007/s10817-021-09607-x},
  doi          = {10.1007/S10817-021-09607-X},
  bibsource    = {dblp computer science bibliography, https://dblp.org}
}

@inproceedings{ZeqiriM0V23,
  author       = {Mustafa Zeqiri and
                  Mark Niklas M{\"{u}}ller and
                  Marc Fischer and
                  Martin T. Vechev},
  title        = {Efficient Certified Training and Robustness Verification of Neural
                  {ODEs}},
  booktitle    = {The Eleventh International Conference on Learning Representations,
                  {ICLR} 2023, Kigali, Rwanda, May 1-5, 2023},
  publisher    = {OpenReview.net},
  year         = {2023},
  url          = {https://openreview.net/forum?id=KyoVpYvWWnK},
  bibsource    = {dblp computer science bibliography, https://dblp.org}
}

@inproceedings{zhang_efficient_2018,
  author       = {Huan Zhang and
                  Tsui{-}Wei Weng and
                  Pin{-}Yu Chen and
                  Cho{-}Jui Hsieh and
                  Luca Daniel},
  editor       = {Samy Bengio and
                  Hanna M. Wallach and
                  Hugo Larochelle and
                  Kristen Grauman and
                  Nicol{\`{o}} Cesa{-}Bianchi and
                  Roman Garnett},
  title        = {Efficient Neural Network Robustness Certification with General Activation
                  Functions},
  booktitle    = {Advances in Neural Information Processing Systems 31: Annual Conference
                  on Neural Information Processing Systems 2018, NeurIPS 2018, December
                  3-8, 2018, Montr{\'{e}}al, Canada},
  publisher    = {Curran Associates, Inc.},
  address   = {Red Hook, NY, USA},
  pages        = {4944--4953},
  year         = {2018},
  url          = {https://proceedings.neurips.cc/paper/2018/hash/d04863f100d59b3eb688a11f95b0ae60-Abstract.html},
  bibsource    = {dblp computer science bibliography, https://dblp.org}
}

@InProceedings{daggitt_et_al:LIPIcs.FSCD.2025.2,
  author       = {Matthew L. Daggitt and
                  Wen Kokke and
                  Robert Atkey and
                  Ekaterina Komendantskaya and
                  Natalia Slusarz and
                  Luca Arnaboldi},
  editor       = {Maribel Fern{\'{a}}ndez},
  title        = {Vehicle: Bridging the Embedding Gap in the Verification of Neuro-Symbolic
                  Programs (Invited Talk)},
  booktitle    = {10th International Conference on Formal Structures for Computation
                  and Deduction, {FSCD} 2025, Birmingham, UK, July 14-20, 2025},
  series       = {LIPIcs},
  volume       = {337},
  pages        = {2:1--2:20},
  publisher    = {Schloss Dagstuhl - Leibniz-Zentrum f{\"{u}}r Informatik},
  address   = {Dagstuhl, Germany},
  year         = {2025},
  url          = {https://doi.org/10.4230/LIPIcs.FSCD.2025.2},
  doi          = {10.4230/LIPICS.FSCD.2025.2},
  bibsource    = {dblp computer science bibliography, https://dblp.org}
}

@inproceedings{thies_cpp_2026,
	title = {Computing Solutions for Systems of Multivariate Ordinary Differential Equations in Rocq},
	booktitle = {Proceedings of the 15th {ACM} {SIGPLAN} International Conference on Certified Programs and Proofs ({CPP})},
	author = {Thies, Holger},
	date = {2026},
	pages = {29--44},
	doi = {10.1145/3779031.3779097},
}

@inproceedings{cohen_rouhling_itp_2017,
	title = {A Formal Proof in {Coq} of {LaSalle}'s Invariance Principle},
	booktitle = {Interactive Theorem Proving ({ITP})},
	author = {Cohen, Cyril and Rouhling, Damien},
	date = {2017},
	publisher = {Springer},
	series = {LNCS},
	volume = {10499},
	pages = {148--163},
	doi = {10.1007/978-3-319-66107-0_10},
}

@inproceedings{rouhling_cpp_2018,
	title = {A Formal Proof in {Coq} of a Control Function for the Inverted Pendulum},
	booktitle = {Proceedings of the 7th {ACM} {SIGPLAN} International Conference on Certified Programs and Proofs ({CPP})},
	author = {Rouhling, Damien},
	date = {2018},
	pages = {28--41},
	doi = {10.1145/3167101},
}

@software{daggitt_2026_20529849,
  author       = {Daggitt, Matthew and
                  Komendantskaya, Ekaterina and
                  Sirman, Alistair and
                  Bruni, Alessandro and
                  Teuber, Samuel and
                  Smart, Josh and
                  Passmore, Grant},
  title        = {Compositional Neural-Cyber-Physical System
                   Verification in the Interactive Theorem Prover of
                   Your Choice Artefact
                  },
  month        = jun,
  year         = 2026,
  publisher    = {Zenodo},
  doi          = {10.5281/zenodo.20529848},
  url          = {https://doi.org/10.5281/zenodo.20529848},
}

\end{document}